\documentclass[reqno,11pt]{amsart}

\usepackage[backend=biber,style=alphabetic,maxbibnames=50,doi=false,url=false,isbn=false,giveninits=true]{biblatex}
\DeclareNameAlias{author}{family-given}
\DeclareNameAlias{editor}{family-given}
\DeclareNameAlias{translator}{family-given}
\renewbibmacro{in:}{}

\usepackage{graphicx}
\usepackage{amssymb}
\usepackage{amsmath}
\usepackage{amsfonts}
\usepackage{amsthm}
\usepackage[svgnames]{xcolor}
\usepackage{hyperref}
\usepackage[mathscr]{eucal}
\usepackage{enumerate}
\usepackage{times}
\usepackage{authblk}
\usepackage[foot]{amsaddr}
\usepackage{comment}
\usepackage{mathrsfs}

\usepackage{upgreek}
\usepackage{mathrsfs}
\usepackage{graphicx}

\usepackage{physics}
\usepackage{cleveref}

\usepackage[normalem]{ulem}

\usepackage{cancel}
\usepackage{tikz}

\numberwithin{equation}{section}
\allowdisplaybreaks[1]

\theoremstyle{plain}
\newtheorem{Def}{Definition}[section]
\newtheorem{Thm}[Def]{Theorem}
\newtheorem{Prop}[Def]{Proposition}
\newtheorem{Lemma}[Def]{Lemma}

\theoremstyle{definition}
\newtheorem{Remark}[Def]{Remark}

\newtheorem{Corollary}[Def]{Corollary}

\newcommand{\beq}{\begin{equation}}
\newcommand{\eeq}{\end{equation}}
\newcommand{\Proof}{\begin{proof}}
\newcommand{\QED}{\end{proof} \noindent}

\newcommand{\mm}{\hspace{-.08cm}\cdot \hspace{-.08cm}}

\newcommand{\M}{\mathcal{M}}

\newcommand{\R}{\mathbb{R}}
\newcommand{\NN}{\mathbb{N}}

\newcommand{\Gammati}{\tilde{\Gamma}}
\newcommand{\Riem}{{\rm Riem}}
\newcommand{\Ric}{{\rm Ric}}
\newcommand{\A}{\mathcal{A}}

\newcommand{\T}{\mathcal{T}}

\newcommand{\diam}{{\rm diam}}

\newcommand{\eps}{\varepsilon}

\newcommand{\loc}{\ensuremath{\text{loc}}}

\DeclareMathOperator*{\esssup}{ess\,sup}

\newcommand{\vol}{\mathrm{vol}}
\newcommand{\D}{\ensuremath{{\mathscr{D}}}}

\newcommand{\Wloc}[2][p]{W^{#2,#1}_{\text{loc}}}
\newcommand{\Wo}{\Wloc{1}}

\newcommand{\Calph}{C^{\alpha}_{\text{loc}}}
\newcommand{\U}{\mathbf{U}}

\newcommand{\cH}{w}

\title[Hawking Theorem]{The Hawking Singularity Theorem for H\"older Continuous Metrics with $L^p$-Bounded Curvature}
\author{Michael Kunzinger}
\thanks{M.\ Kunzinger, R.\ Steinbauer, and I.\ Vega-Gonz\'alez:
  Faculty of Mathematics, University of Vienna,
  Oskar-Morgenstern-Platz 1, 1090 Wien, Austria.}
\email{michael.kunzinger@univie.ac.at}

\author{Moritz Reintjes}
\thanks{M.\ Reintjes: Department of Mathematics,
  City University of Hong Kong, SAR Hong Kong.}
\email{moritzreintjes@gmail.com}

\author{Roland Steinbauer}
\email{roland.steinbauer@univie.ac.at}

\author{In\'es Vega-Gonz\'alez}
\email{ines.vega.gonzalez@univie.ac.at}

\begin{document}

\begin{abstract}
We prove a low-regularity version of Hawking's singularity theorem for Lorentzian metrics
$g\in W^{1,p}_{\mathrm{loc}}$ with $\Riem[g]\in L^p_{\mathrm{loc}}$, where $p>2n$ and
$n=\dim \mathcal M$. This extends previous results beyond the Lipschitz regime. Under
suitable lower Ricci bounds and upper mean curvature assumptions, expressed in terms of
temporal functions, we establish both the globally hyperbolic version of Hawking's theorem,
in the form of an upper bound on the time separation from a spacelike Cauchy hypersurface,
and the version with a compact achronal spacelike hypersurface, yielding timelike
RT-geodesic incompleteness.   

The proof combines regularisations, based on the elliptic RT-equations, to raise the regularity of the metric by one derivative, with a refinement of the previously used manifold convolution. We introduce a new smeared-out notion of mean curvature adapted to the low metric regularity before, and the $W^{2,p}$-hypersurfaces arising after regularisation. As further consequences, we show that $W^{1,p}$-Lorentzian metrics with $L^p$-bounded curvature are causally plain, and we prove a corresponding low-regularity version of Myers's theorem in the Riemannian setting.  
\end{abstract}

\maketitle

\tableofcontents

\section{Introduction} 

The celebrated singularity theorems of Penrose and Hawking constituted the breakthrough that established singularity formation as a generic feature of General Relativity. While Penrose's theorem \cite{Pen:65} addresses stellar gravitational collapse,  Hawking treated the cosmological situation \cite{Haw:67}. Under generic conditions the theorems predict singularities in the sense of timelike or null geodesic incompleteness. However, the original proofs rely on the gravitational metric tensor being at least $C^2$-regular, a fact which undermines the physical significance of the results: incompleteness might be avoided by a (physically harmless) drop in regularity, and the theorems might merely predict a rough, but otherwise non-singular geometry.  

The most basic structures of spacetime geometry only require a continuous metric, and one can make sense of solutions of the Einstein equations in a weak sense even for metrics in the so-called Geroch-Traschen class \cite{GT:87}, (i.e., locally uniformly nondegenerate $L^\infty$-metrics\footnote{Naturally all $L^p$ and Sobolev spaces here are local but we refrain from overburdening the notation in the introduction and from making this explicit at all occurrences.} with $L^2$-Christoffel symbols \cite{LM:07,SV:09}), the lowest regularity for which curvature can be stably defined in a distributional sense. So one may ask the question, left open by the classical theory, as to whether the conditions of the singularity theorems (or suitable generalisations thereof) still suffice to imply incompleteness when metrics are of low regularity, that is below $C^2$ and down to the Geroch-Traschen class.

This question already motivated Hawking and Ellis in their classic \cite[Sec.\ 8.4]{HawkingEllis} to discuss low regularity extensions of the results, see also Senovilla's major review article \cite[Sec. 6.2]{Sen:98}. While these sources already suggested that regularisation of the metric via convolution might help to resolve the regularity issues, significant progress only became possible with the causality-preserving regularisation put forward by Chru\'sciel-Grant \cite{CG:12}.
Indeed the first step was to lower the regularity to the $C^{1,1}=C^{2-}$-class, i.e., to $C^1$-metrics with locally Lipschitz Christoffel symbols for the Hawking \cite{KSSV:15}, the Penrose \cite{KSV:15}, and the most refined of the classical results, the Hawking-Penrose theorem \cite{GGKS:18}. In this regularity the exponential map can still be used \cite{Min:15,KSSV:14}, and in a parallel development Lorentzian causality theory \cite{Min:19a} has been extended to continuous metrics \cite{CG:12,Sae:16} and even to cone structures \cite{FS:12,bernard_suhr}, in particular by Minguzzi in the seminal \cite{Minguzzi-cone-structures}. However, it was firmly established that some of the main features of the theory such as the push-up principle cease to hold below the Lipschitz threshold $C^{0,1}$ \cite{GrantKuSaSt}, due to the bubbling phenomenon discovered already in \cite{CG:12}.

Building on these developments and providing a refined analysis of the ODE-theory for geodesic equations, Graf in \cite{Graf} succeeded in lowering the metric regularity down to $C^1$, for both the Hawking and the Penrose theorem.
The Gannon-Lee theorem, which connects the topological structure of spacetime to its singularity structure and the Hawking-Penrose theorem in $C^1$-regularity followed suit, \cite{benedikt,KOSS:22}. 

In the recent paper \cite{CGHKS25} the regularity required for the Hawking theorem was further lowered to $C^{0,1}$, which is the maximal class originally envisaged in \cite[Sec.\ 8.4]{HawkingEllis}, and which meanwhile has also turned out to be the `causality threshold' regularity. This work combined refined estimates for the Ricci curvature of regularising smooth metrics based upon a quite general Friedrichs-type lemma with the recent worldvolume estimates of Graf et al.\ \cite{GKOS:22} that, in the absence of a suitable ODE-theory, replaced the usual focusing techniques for timelike geodesics.
\medskip

In parallel with these analytical developments that all use distributional energy conditions resp.\ Ricci curvature bounds, in recent years a significant development in Lorentzian geometry has occurred,  based on synthetic methods as introduced in \cite{KS:18}. Such synthetic Lorentzian spaces (cf.\ also alternative approaches by Bykov, Minguzzi, Suhr \cite{minguzzi-suhr24,bykov-minguzzi-suhr}, as well as Braun and McCann, \cite{braun_mccann_causal}) are analogues of metric length spaces and enable a definition of Ricci curvature bounds based on optimal transport and displacement convexity of entropy functionals, extending the celebrated theory of Sturm and Lott-Villani \cite{S:06a,Sturm2006II,LV:09} to the Lorentzian setting \cite{MCC:20,MS:23,CM:24}, see \cite{CM:22} for a first overview. A Hawking singularity theorem in this framework has been provided by Cavaletti-Mondino \cite{CM:24}. The Penrose case, which is much more involved since it has to be based on (a synthetic version of) the null energy condition was pioneered in this framework independently by McCann\cite{McC:24} and Ketterer \cite{K:24}, where the latter reference introduced a version of null optimal transport leading to a proof of a Penrose singularity theorem for smooth spacetimes using purely synthetic conditions. Recently this framework has also been employed by Braun-Rotolo in \cite{BR:26} to obtain a synthetic Gannon-Lee theorem for smooth spacetimes. Finally, Cavaletti-Manini-Mondino put forward an alternative version of null optimal transport in \cite{CMM:25a} which allowed them to prove a synthetic Penrose theorem for continuous spacetimes in \cite{CMM:25b}. 

Naturally there arises the question of compatibility between the analytic and the synthetic approach to singularity theorems in the areas where they overlap. This leads to the somewhat subtle question of compatibility of distributional and synthetic energy conditions, which has been considered in \cite{KOV:22,BC:24}, with definitive answers even in the Riemannian case having been obtained only very recently \cite{MR:25}. In addition, one also needs to compare the various bounds on the mean curvature used in the corresponding works. For a detailed comparison and a hint to open questions specifically in the Hawking case see \cite[Rem.\ 4.5]{CGHKS25}. Also for a more general review of the different approaches with an emphasis on the analytical one see \cite{Ste:23}.
\medskip

In the present paper we take the analytical approach to low regularity singularity theorems a significant step forward by extending the Hawking theorem from the $C^{0,1}=W^{1,\infty}$-setting of \cite{CGHKS25}, down to metrics in the Sobolev space $W^{1,p}$ with Riemann curvature bounded in $L^p$, for any $p>2n$, where $n\geq 2$ denotes the dimension of spacetime. By Morrey's inequality, these metrics are only H\"older continuous of order $C^{0,\alpha}$ where the H\"older exponent satisfies $\alpha=1- n/p > 1/2$. Hence this can be thought of as being halfway down to the case of continuous metrics, and close to the $L^2$-threshold on metric derivatives set by the Geroch-Traschen class. In particular, we view our assumption of the Riemann curvature to be bounded in some $L^p$-space as only a minor concession, because any spacetime with curvature that cannot be bounded in any $L^p$-space, (hence exhibiting infinite tidal forces), would by itself be so singular that the question of geodesic completeness would only play a subordinate role.

This new extension of the Hawking theorem is primarily based on the theory of the `Regularity Transformation equations' (\textit{RT-equations}), a system of elliptic partial differential equations whose solutions provide coordinate transformations for regularising spacetime connections, recently developed by Reintjes-Temple in a series of papers \cite{ReintjesTemple_ell1, ReintjesTemple_ell2, ReintjesTemple_ell3, ReintjesTemple_ell4, ReintjesTemple_ell5, ReintjesTemple_ell6}. Based on the existence theory for the RT-equations, 
Reintjes-Temple prove in \cite{ReintjesTemple_essreg} that any affine connection in $L^p$ can be regularized by one derivative to $W^{1,p}$ if and only if its Riemann curvature lies in $L^p$, see also the summary in appendix \ref{Sec_RT_appendix}. 

In this paper, based on our standing $L^p$-curvature assumption, we employ the RT-equations to regularise the H\"older continuous $W^{1,p}$-metric to $W^{2,p}$-regularity\footnote{Recall that a metric is always precisely one derivative more regular than its connection.}. The resulting $W^{2,p}$-metric regularity, using Morrey's inequality again,  brings us into the $C^1$-setting, i.e., to the metric regularity required for Graf's version of the Hawking theorem in \cite{Graf}. However, the regularisation of the metric by the RT-equations comes at a price: In fact the higher metric regularity is achieved in a new atlas for the spacetime manifold, which is merely $W^{2,p}$-compatible with the original one. As a result one can no longer maintain the assumption of smooth (Cauchy) hypersurfaces and smooth orthogonal frames thereon necessary for \cite{Graf}, and in fact common to all prior works on analytic extensions of the singularity theorems to low regularity. This forecloses a direct application of the results in \cite{Graf} and forces us to develop an appropriate notion of mean curvature for $W^{2,p}$-surfaces, a `smeared-out' version of mean curvature significantly more subtle than the one used in \cite[Sec.\ 4]{CGHKS25}. There the objective was to thicken the smooth (partial) Cauchy surface by a smooth flow out to obtain a `smeared out' mean curvature of $L^\infty$-regularity for the Lipschitz metric on an open neighbourhood of the hypersurface in spacetime. Here, however, smoothness of the hypersurface breaks down and we devise a new approach to mean curvature based on time functions and adapted to $W^{1,p}$-H\"older continuous frames, see Section \ref{Sec_mean_curv_timefunction}. To turn the resulting mean curvature bounds into a form applicable in the analytic machinery of focussing estimates for timelike geodesics we need, on top of the RT-regularisation, another regularisation of the metric now via convolution, i.e., we apply a smoothing convolution to the $W^{2,p}$-metric obtained by the RT-regularisation. This procedure then necessitates a new and refined analysis of manifold convolutions and the establishment of a new Friedrichs-type lemma for $L^\infty$-convergence which we consider interesting in their own right, see Section \ref{Sec_refined_conv_lemmas}. 
This finally brings us into a situation to which the classical smooth Hawking theorem applies to imply geodesic incompleteness. To make sense of this very notion at the level of H\"older continuous metrics---a setting \emph{below} the $L^\infty$-connection regularity to which the Filippov method \cite{Fil:88} of differential inclusions applies, cf.\ \cite{Ste:14, LLS:21}---we make use of the notion of weak geodesics based on the RT-regularisation developed in \cite{ReintjesTemple_geod} and which we refer to as \emph{RT-geodesics}.

We remark that the RT-regularisation is required both conceptually and in several steps of our proofs. Indeed, as remarked above, H\"older metrics without $L^p$-curvature bounds (and hence without the RT-regularisation they entail) generally fail to have a well-posed causal structure and
it is not clear at all how to introduce the very notion of a geodesic curve and consequently of geodesic incompleteness. On the other hand, our proof (like other low regularity proofs of the Hawking theorem) requires approximating the rough metric $g$ with both mollified metrics $g_\epsilon$ and with smooth metrics $\check{g}_\eps$ with controlled lightcones. For our Friedrichs-type $L^\infty$-estimate on the mean curvature (see Theorem \ref{Prop:Mean-curvature-convergence}), to compensate blow-up, it is crucial for the difference $\check{g}_\eps - g_\eps$ to converge to $0$ in $C^1$, linearly in $\eps$, accomplished in Lemma~\ref{lem:m+}. This linear-in-$\eps$ convergence requires $g\in C^1$ already at the level of $C^0$ convergence, marking a crucial place were the RT-regularisation enters our method.
Finally, our standing assumption of $L^p$-curvature directly allows for $L^p$-Ricci lower bounds, which are readily seen to be preserved under RT-regularisation (Lemma \ref{lem:prb}) and almost preserved under convolution regularisation (Proposition \ref{Prop:Graf3.13-modified})---compare this to the daunting challenge to extract $L^p$-lower Ricci bounds for convolution  regularisations already of Lipschitz metrics satisfying only a distibutional Ricci bound, cf.\ \cite[Sec.\ 3]{CGHKS25}.

\medskip

This work is organised in the following way. In Section \ref{sec:preliminaries} we introduce our  notation and conventions and assemble some preliminaries including a refined estimate on the convergence of $\check g_\eps-g_\eps$ for convolution regularisations of $C^1$-metrics (Lemma \ref{lem:m+}). In Section \ref{Sec_RT} we state the fundamental result on RT-regularisations (Theorem \ref{Thm_smoothing_preliminaries}) and introduce the notion of RT-geodesics and RT-completeness. Some details on the RT-equations are collected in Appendix \ref{Sec_RT_appendix}. In Section \ref{Sec_causality}, we employ the RT-regularisation to establish that the class of spacetimes with $W^{1,p}$-metrics and $L^p$-bounded curvature possess a well-defined causal structure. In particular, we prove that in such spacetimes the bubbling phenomenon is absent (Proposition \ref{prop:cp}), clarify the relation between causal RT-geodesic and causal maximisers and establish a non-branching result (Propositions \ref{Prop_maximisers->geodesics} and \ref{prop:nb}). In Section \ref{Sec_mean_curv_timefunction} we define the notion of `smeared out' mean curvature we work with for the rest of the paper and state the main result on its approximation (Theorem \ref{Prop:Mean-curvature-convergence}). Its proof relies on refined estimates on the manifold convolution which we state and prove in Section \ref{Sec_refined_conv_lemmas}. In Section \ref{sec:hawking} we proceed to our main results, namely both versions of the Hawking theorem for $W^{1,p}$-metrics with $L^p$-curvature. First, in the globally hyperbolic situation we establish the usual bound on the time separation function from the Cauchy surface $\Sigma$ (Theorem \ref{Th:Hawking-I}). Actually we give a `quantified version' of the theorem balancing the Ricci bounds against the mean curvature bounds of $\Sigma$. For convenience we also state the corresponding smooth result (Theorem \ref{Th:Hawking-smooth}) and since we were unable to find a unified proof of all cases in the literature we provide one in Appendix \ref{Sec_Hawking_proof}. The second version of the Hawking theorem in which we merely assume the existence of a compact achronal spacelike hypersurface asserts geodesic incompleteness in the RT-sense (Theorem \ref{Th:Hawking-II}). Finally, in Section \ref{sec:myers} we provide the Riemannian counterpart of our main result, a Myers's theorem again for metrics in $W^{1,p}$ with the Riemann tensor in $L^p$.

\section{Prerequisites}\label{sec:preliminaries}

In this section we collect some prerequisites to be used throughout the paper thereby also fixing our notations. 
Our standard references are \cite{ON83} for smooth Lorentzian geometry, \cite{Minguzzi-cone-structures} for causality theory, \cite{Sae:16} for the low-regularity (continuous) Lorentzian setting, and \cite[Sec.\ 2]{KOSS:22} for distribution theory on manifolds.

\subsection{Lorentzian geometry}
We will generally be concerned with manifolds $\mathcal{M}$ which we assume to be Hausdorff, second countable, connected and of dimension $n\geq 2$. We will assume that the atlas ${\mathcal A}$ of $\mathcal{M}$ is of regularity at least $W^{2,p}_\text{loc}$, where we use $W^{s,p}_\text{loc}$ to denote local Sobolev spaces of positive integer order $s\geq 1$ modelled on $L^p$-spaces with $1\leq p\leq\infty$. We will generally assume $p>n$ so that by Morrey's inequality local $W^{s,p}$-regularity implies local H\"older continuity of order $s-1$, i.e., $W^{s,p}_\text{loc}\subseteq C^{s-1,\alpha}_\text{loc}$ with $\alpha=1-n/p$. Hence all our manifolds will be at least $C^1$, so there always exists a unique maximal $C^1$-compatible smooth atlas and we will work in this atlas unless explicitly stated otherwise.

We will denote Lorentzian metrics on $\mathcal{M}$ by $g$ which we will assume to be at least continuous and any added regularity will be explicitly indicated. A spacetime will be a pair $(\mathcal{M},g)$ with a time-orientation fixed by a smooth timelike vector field. We will say a pair $(\mathcal{M},g)$ is a $C^{k}$-spacetime if $g$ has the corresponding regularity $k=0,1,\dots$ and similarly for spaces $W^{s,p}_\text{loc}$, in particular, for Lipschitz regularities $W^{1,\infty}_\text{loc}=C^{0,1}_\text{loc}$ and $W^{2,\infty}_\text{loc}=C^{1,1}_\text{loc}$. Sometimes it will be necessary to explicitly include the atlas $\A$ in the notation, in which case we will write $(\M,\A,g)$ for a spacetime and denote $g$'s regularity with respect to $\A$, e.g.\ by writing $g\in W^{s,p}_{\A,\loc}$.

We will especially be interested in $g\in W^{1,p}_\text{loc}$ and observe that, by Morrey's inequality $g\in C^{0,\alpha}_\text{loc}$ ($\alpha=1-n/p$) but not locally Lipschitz, i.e., $g\not\in C^{0,1}_\text{loc}$. It follows that $g$'s inverse $g^{-1}$ is also in $W^{1,p}_\text{loc}$ and $C^{0,\alpha}_\text{loc}$ and as a result, such metrics are in the Geroch-Traschen class, ($g\in W^{1,2}_\text{loc}\cap L^\infty_\text{loc}$ and uniformly nondegenerate, hence $g^{-1} \in W^{1,2}_\text{loc}\cap L^\infty_\text{loc}$ as well) which allows for the stable definition of distributional curvature \cite{GT:87,LM:07,SV:09}.

\medskip

For the Riemann curvature tensor (on smooth enough spacetimes, i.e., with metric $C^2$ or higher), we follow the convention of \cite{ON83}, $R_{XY}Z=\nabla_{[X,Y]}Z-[\nabla_X,\nabla_Y]Z$ for smooth vector fields $X,Y,Z\in\frak{X}(\mathcal{M})$. Note that $\mathrm{Ric}$ is independent of this convention. Explicitly (cf.\ \cite[Lem.\ 3.52]{ON83}), for a given orthonormal frame $E_i$ we have    \begin{equation}\label{eq:Ricci}
 \mathrm{Ric}(X,Y) = \sum \varepsilon_i g(R(X,E_i)Y,E_i ) = - \sum \varepsilon_i g( R(E_i,X)Y,E_i ).
\end{equation}
where $\eps_i:= g(E_i, E_i)=\pm 1$.

Our convention for the mean curvature of a smooth (at least $C^2$) spacelike hypersurface $\Sigma$ (compatible with \cite{Graf}, but different from \cite{ON83}) is as follows: Given the future pointing unit normal vector field $N$ on $\Sigma$, define the \emph{shape operator} for $\Sigma$ as $S_p(v) := \nabla_v N$ $(p\in \Sigma,\ v\in T_p\Sigma)$. The \emph{mean curvature} $H$ of $\Sigma$ then is 
\begin{equation}\label{eq:mean-curvature-def}
\begin{split}
H&:= \mathrm{tr}_\Sigma(S) = \sum_{i=1}^{n-1} g(\nabla_{e_i} N, e_i) =
-\sum_{i=1}^{n-1}g(N,\mathrm{II}(e_i,e_i)),
\end{split}
\end{equation}
where $(e_1,\dots,e_{n-1})$ is an orthonormal basis of $T_p\Sigma$ and $\mathrm{II}$ denotes the second fundamental form of $\Sigma$. Denoting by 
\begin{equation}\label{eq:mean-curvature-vector-def}
   \vec{H}:= \frac{1}{n-1}\sum_{i=1}^{n-1}\mathrm{II}(e_i,e_i) 
\end{equation}
the mean curvature vector field of $\Sigma$, we have $H=-(n-1)g(N,\vec{H})$.

If $\{X_1,\dots,X_{n-1}\}$ is a local frame in $\mathfrak{X}(\Sigma)$ that is not necessarily orthonormal, then denoting by $G^{ij}$ the entries
of the inverse of the matrix $(g(X_k,X_l))_{k,l=1}^{n-1}$, we have
\begin{equation}\label{eq:mean-curvature-non-orth}
H=  \sum_{i,j=1}^{n-1} G^{ij} g(\nabla_{X_i}N,X_j) = -\sum_{i,j=1}^{n-1} G^{ij} g(N,\nabla_{X_i}X_j).
\end{equation}

\subsection{Causality theory}
We base causality on locally Lipschitz curves $\gamma: I\to {\mathcal M}$, calling them \emph{timelike, null, causal, future} or \emph{past directed} if $\dot\gamma$ has the respective properties almost everywhere. Then $p\ll q$ and $p\leq q$ means there exists a future directed timelike resp.\ causal curve from $p$ to $q$ and we set as usual $I^+(A):=\{q\in {\mathcal M}:\, p\ll q\ \mathrm{for\,some}\,p\in A\}$ and $J^+(A):=\{q\in {\mathcal M}:\, p\leq q\ \mathrm{for\,some}\,p\in A\}$. 

We call $({\mathcal M},g)$ \emph{globally hyperbolic} if it is non-totally imprisoning and all causal diamonds $J(p, q):= J^+ (p)\cap J^-(q)$ are compact, which implies strong causality \cite[Prop.\ 5.6]{Sae:16}.
A \emph{Cauchy hypersurface} $\Sigma\subseteq {\mathcal M}$ is met exactly once by every inextendible causal curve; it is always a closed acausal topological hypersurface. Moreover, $({\mathcal M},g)$ is globally hyperbolic if and only if it possesses a Cauchy hypersurface $\Sigma$, in which case ${\mathcal M}\cong\R\times\Sigma$ \cite[Sec.\ 5]{Sae:16}. In accordance with \cite{BS06} we assume all spacelike Cauchy hypersurfaces to be smooth, but to avoid confusion we will always explicitly speak of smooth spacelike Cauchy surfaces.

The \emph{Cauchy development} of an achronal set $S\subseteq\mathcal{M}$ is $D(S):=D^+(S)\cup D^-(S)$, where $D^\pm(S)$ are the sets of points $p\in {\mathcal M}$ such that every past/future inextendible causal curve through $p$ meets $S$. The interior $D(S)^\circ$ of any acausal set $S$ is globally hyperbolic \cite[Cor.\ 5.8]{Sae:16}. The Avez-Seifert theorem extends to continuous globally hyperbolic spacetimes: there exist globally maximising causal curves between causally related points \cite[Prop.\ 6.4]{Sae:16}, though the relation between maximisers and geodesics becomes more subtle, cf.\ e.g.\ \cite[Remark 2.1]{CGHKS25} and we will return to this discussion in Section \ref{Sec_causality}.

For $p,q\in {\mathcal M}$ the future \emph{time separation} is 
\begin{equation}
    \tau(p,q):=\sup\left(\left\{ L(\gamma):\gamma\;\text{future directed causal from }p\text{ to }q\right\} \cup\{0\}\right),\label{eq:point time sep}
\end{equation}
where $L(\gamma):=\int_I\sqrt{|g(\dot{\gamma}(t),\dot{\gamma}(t))|}dt$ is the Lorentzian arc-length of $\gamma: I\to {\mathcal M}$. The time separation to a subset $S$ is
\begin{equation}
    \tau_{S}(p):=\sup_{q\in S}\tau(q,p).\label{eq:subset time sep}
\end{equation}
While basic causality features (push up principle, openness of $I^+$) fail for continuous metrics \cite{CG:12,GrantKuSaSt}, they hold for \emph{causally plain} metrics, which include all $g\in C^{0,1}$ \cite[Thm.\ 1.20]{CG:12} and, as we will show in section \ref{Sec_causality}, all $W^{1,p}_\text{loc}$-metrics with $L^p_\text{loc}$-Riemann tensor.

Though we work in the low regularity \emph{spacetime} setting, we will freely use results from more general frameworks, particularly the closed cone structures of Minguzzi \cite{Min:19a}, and we briefly recall how continuous metrics fit into it.
\begin{Remark}[Cone structures]\label{rem:cs}
     A \emph{cone structure} $({\mathcal M},C)$ is a multivalued map ${\mathcal M}\ni p\mapsto C_p$ where $C_p\subseteq T_p{\mathcal M}\setminus 0$ is a closed sharp convex non-empty cone \cite[Def.\ 2.2]{Min:19a}. $({\mathcal M},C)$ is \emph{closed} if it is a closed subbundle of the slit tangent bundle and it is called \emph{proper} if in addition $\mathrm{ Int}(C)_p\not=\emptyset$ for each $p$ \cite[Defs.\ 2.3, 2.4]{Min:19a}. For $g\in C^{0}$, one defines the associated canonical cone structure
     \begin{equation}
      C_p:=\{v\in T_p{\mathcal M}\,\backslash\,\{0\}:g(v,v)\leq0, v\text{ future directed}\}    
     \end{equation} 
     \cite[Ex.\ 2.1]{Min:19a}. Then $p\mapsto C_p$ is continuous by \cite[Prop.\ 2.4]{Min:19a} and by \cite[Props.\ 2.3, 2.5]{Min:19a} $({\mathcal M},C)$ it is a proper cone structure.
\end{Remark}

\subsection{Distributions and regularisations}

We will generally equip $\mathcal{M}$ with a smooth complete background Riemannian metric $h$ and denote its induced distance function by $d^h$. Estimates on (at least continuous) vector and tensor fields $X$ and $T$ will always be done w.r.t.\ the norms induced by $h$, which we denote by $|X|_h$ and $|T|_h$, respectively or sometimes just by $|X|$ and $|T|$. Note that this is just a matter of convenience since we will only do estimates on compact sets where any such norm is equivalent to the Euclidean norm of the chart components and hence the estimates do not depend on the choice of $h$, i.e., on any compact $K\subseteq \M$ e.g.\ $|T|_h\leq C|T|_{h'}$ and $|T|_{h'}\leq C|T|_h$ for any choice of background Riemannian metrics $h$ and $h'$. Note that here and in the rest of the paper we adhere to the convention that $C$ denotes a generic constant, that may change its value from line to line.

Spaces of test functions will be denoted by $\mathscr{D}\equiv C^\infty_c$ and distributions by $\mathscr{D}'$, and in particular we shall write $\mathscr{D}'{\mathcal T}^r_s(\mathcal{M})$ for the space of distributional tensor fields. 
We shall repeatedly have to regularise distributional tensor fields, following constructions introduced in \cite{gkos2001geometricgeneralized,KSSV:14,Graf}. To this end we will write $K\Subset \mathcal{M}$ if $K$ is a compact subset of $\mathcal{M}$, and the regularisation parameter $\eps$ will generally be taken from $(0,1]$.
Let  $(U_\alpha,\psi_\alpha)_{\alpha\in A}$ be a countable and locally finite family of relatively compact chart neighbourhoods covering $\mathcal{M}$ and let $(\xi_\alpha)_{\alpha\in A}$ be a subordinate partition of unity  such that $\mathrm{supp}(\xi_\alpha)\subseteq U_\alpha$ for all $\alpha$. Pick a family of cut-off functions $\chi_\alpha\in\mathscr{D}(U_\alpha)$ with $\chi_\alpha\equiv 1$ on a
neighbourhood of $\mathrm{supp}(\xi_\alpha)$. Let $\rho\in \D(B_1(0))$ be a non-negative,
spherically symmetric test function with unit integral and set $\rho_{\eps}(x):=\eps^{-n}\rho\left (\frac{x}{\eps}\right)$. Then with $f_*$ (resp.\ $f^*$) denoting the push-forward (resp.\ pull-back) of distributions under a diffeomorphism $f$, for any tensor distribution $T \in \mathscr{D}'\mathcal{T}^r_s(\mathcal{M})$, let      
\begin{equation}\label{eq:M-convolution}
T\star_\mathcal{M} \rho_\eps(p):= \sum\limits_{\alpha\in A}\chi_\alpha(p)\,(\psi_\alpha)^*\Big[\big((\psi_{\alpha})_* (\xi_\alpha\cdot \mathcal{T})\big)*\rho_\eps\Big](p).
\end{equation}
Here, $(\psi_{\alpha})_* (\xi_\alpha\cdot \mathcal{T})$ is a compactly supported distributional tensor field on $\R^n$, and convolution with $\rho_\eps$ is understood component-wise,  yielding a smooth field on $\R^n$. The cut-off functions $\chi_\alpha$ ensure that $(\eps,p) \mapsto T\star_\mathcal{M} \rho_\eps(p)$ is a smooth map on $(0,1] \times \mathcal{M}$. 
Since $\rho\ge 0$, it follows that for any nonnegative scalar distribution $u\in \D'(\mathcal{M})$ we also have $u\star_\mathcal{M} \rho_\eps \ge 0$ for any $\eps\in (0,1]$.

When regularising a low regularity Lorentzian metric $g$, it is essential to also control the relative positioning of the light cones of the approximating metrics in relation to those of $g$.  Given two Lorentzian metrics $g_1,g_2$ on $\mathcal{M}$ we write $g_1\prec g_2$ and say that $g_1$ has \emph{strictly narrower light cones} than $g_2$, if for all non-vanishing vectors $X$
\begin{align}\label{Eq:prec}
g_1(X,X)\leq0\quad \text{implies}\quad g_2(X,X)<0.
\end{align}
It was shown by Chrusciel and Grant (\cite{CG:12}, cf.\ also \cite{KSSV:14,Graf}) that any continuous Lorentzian metric $g$ can be approximated by families of smooth Lorentzian metrics $\check g_\eps \prec g \prec \hat g_\eps$, see Figure \ref{fig:nested_lightcones}. The precise statement we will use is a mild variant of \cite[Lem.\ 2.4] {CGHKS25}.

\begin{Lemma}[Convergence of approximating metrics]\label{Le:approximating metrics}
Let $g\in C^{0}(\mathcal{M})$ be a Lo\-rentzian metric and 
set $g_\eps=g*_{\mathcal{M}}\rho_\eps$. Then there are smooth Lorentzian metrics $\hat g_\eps$ and $\check g_\eps$ on $\mathcal{M}$ with the following properties:
\begin{itemize}
    \item[(i)] $\check{g}_\varepsilon\prec g \prec\hat{g}_\varepsilon$.
    \item[(ii)] $\check g_\eps$, $\hat g_\eps\to g$, and $(\check g_\eps)^{-1}$, $(\hat g_\eps)^{-1}\to g^{-1}$ locally uniformly.
    \item [(iii)] If  $g\in W^{s,p}_{\mathrm{loc}}(\M)$  ($s\in\mathbb{N}_0$, $1\leq p<\infty$) then the convergence in (ii) is in $W^{s,p}_{\mathrm{loc}}(\M)$ as well.
    \item[(iv)] $\check g_\eps-g_\eps \to 0$, $\hat g_\eps-g_\eps\to 0$, and
    $(\check g_\eps)^{-1}-(g_\eps)^{-1}\to 0$,  $(\hat g_\eps)^{-1}-(g_\eps)^{-1}\to 0$, all in $C^\infty(\mathcal{M})$.
    In particular, $\Ric[\check g_\eps] - \Ric[g_\eps] \to 0$ in $C^\infty(\M)$.
    \item[(v)] For any compact subset $K\Subset M$ there exists a sequence $\varepsilon_j\searrow0$ such that $\check g_{\eps_j}\prec\check g_{\eps_{j+1}}$ and $\hat g_{\eps_{j+1}}\prec \hat g_{\eps_j}$ for all $j\in\mathbb{N}$.
\end{itemize}
\end{Lemma}

\begin{figure}[htbp]
    \centering
    \begin{tikzpicture}[scale=1.55,line cap=round,line join=round]
        \coordinate (O) at (0,0);
        \fill[black] (O) circle (1.2pt);
        \node[below=2pt] at (O) {$p$};

        \path[fill=RoyalBlue!12, draw=none, opacity=.35]
            (O) -- (-1.65,2.15)
            arc[start angle=180,end angle=0,x radius=1.65,y radius=0.33]
            -- cycle;
        \path[shade, left color=RoyalBlue!8, right color=RoyalBlue!35,
              middle color=white, shading angle=18, draw=none, opacity=.45]
            (O) -- (-1.65,2.15)
            arc[start angle=180,end angle=360,x radius=1.65,y radius=0.33]
            -- cycle;
        \draw[thick, dashed, RoyalBlue!70!black, opacity=.45]
            (1.65,2.15)
            arc[start angle=0,end angle=180,x radius=1.65,y radius=0.33];
        \draw[thick, dashed, RoyalBlue!85!black] (O) -- (-1.65,2.15);
        \draw[thick, dashed, RoyalBlue!85!black] (O) -- ( 1.65,2.15);
        \draw[thick, dashed, RoyalBlue!85!black]
            (-1.65,2.15)
            arc[start angle=180,end angle=360,x radius=1.65,y radius=0.33];
        \node[RoyalBlue!85!black,font=\small\boldmath] at (1.50,1.62) {$\hat{g}_\varepsilon$};

        \path[fill=black!5, draw=none, opacity=.45]
            (O) -- (-1.18,2.15)
            arc[start angle=180,end angle=0,x radius=1.18,y radius=0.24]
            -- cycle;
        \path[shade, left color=black!4, right color=black!28,
              middle color=white, shading angle=15, draw=none, opacity=.55]
            (O) -- (-1.18,2.15)
            arc[start angle=180,end angle=360,x radius=1.18,y radius=0.24]
            -- cycle;
        \draw[thick, black!60, opacity=.4]
            (1.18,2.15)
            arc[start angle=0,end angle=180,x radius=1.18,y radius=0.24];
        \draw[thick, black!85] (O) -- (-1.18,2.15);
        \draw[thick, black!85] (O) -- ( 1.18,2.15);
        \draw[thick, black!85]
            (-1.18,2.15)
            arc[start angle=180,end angle=360,x radius=1.18,y radius=0.24];
        \node[black!85,font=\small\boldmath] at (1.02,1.58) {$g$};

        \path[fill=FireBrick!12, draw=none, opacity=.40]
            (O) -- (-0.72,2.15)
            arc[start angle=180,end angle=0,x radius=0.72,y radius=0.15]
            -- cycle;
        \path[shade, left color=FireBrick!8, right color=FireBrick!35,
              middle color=white, shading angle=12, draw=none, opacity=.55]
            (O) -- (-0.72,2.15)
            arc[start angle=180,end angle=360,x radius=0.72,y radius=0.15]
            -- cycle;
        \draw[thick, densely dotted, FireBrick!65!black, opacity=.4]
            (0.72,2.15)
            arc[start angle=0,end angle=180,x radius=0.72,y radius=0.15];
        \draw[thick, densely dotted, FireBrick!85!black] (O) -- (-0.72,2.15);
        \draw[thick, densely dotted, FireBrick!85!black] (O) -- ( 0.72,2.15);
        \draw[thick, densely dotted, FireBrick!85!black]
            (-0.72,2.15)
            arc[start angle=180,end angle=360,x radius=0.72,y radius=0.15];
        \node[FireBrick!85!black,font=\small\boldmath] at (0.30,1.58) {$\check{g}_\varepsilon$};

        \draw[thick, dashed, RoyalBlue!85!black]
            (-1.65,2.15)
            arc[start angle=180,end angle=360,x radius=1.65,y radius=0.33];
        \draw[thick, black!85]
            (-1.18,2.15)
            arc[start angle=180,end angle=360,x radius=1.18,y radius=0.24];
    \end{tikzpicture}
    \caption{Nested future light cones corresponding to the approximations $\check{g}_\varepsilon \prec g \prec \hat{g}_\varepsilon$.}
    \label{fig:nested_lightcones}
\end{figure}
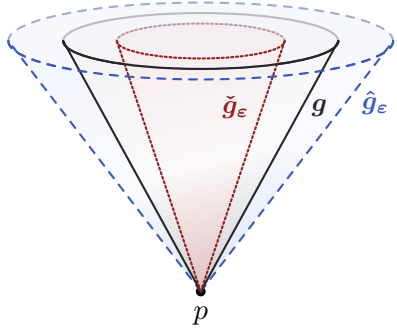

The proof of Theorem \ref{Prop:Mean-curvature-convergence} below will require a refinement of Lemma \ref{Le:approximating metrics} (iv), which amounts to an
extension of \cite[Lem.\ 4.2 and Cor.\ 4.3]{Graf}.
For its formulation, denote by $\nabla_h$ the covariant differential with respect to our fixed background Riemannian metric $h$. For any tensor field $T$ and compact set $K \Subset \mathcal{M}$, let $\|T\|_{\infty, K}$ denote the supremum norm of $T$ on $K$ with respect to $h$. Let $\omega_{\nabla_h g} : [0,\infty) \to [0,\infty)$ denote the local modulus of continuity of the fixed continuous tensor field $\nabla_h g$ on $K$, defined by $\omega_{\nabla_h g}(r) := \sup_{\substack{p,q \in K \\ d_h(p,q) \le r}} \|\nabla_h g(p) - \nabla_h g(q)\|_h$. Note that $\omega_{\nabla_h g}(r) \to 0$ as $r \to 0$.  

\begin{Lemma}[Precise rate of convergence for $\check g_\eps-g_\eps$]\label{lem:m+}
Let $g\in C^1(\M)$   be a Lorentzian metric, let $g_\eps=g\star_{\M}\rho_\eps$, let $\tilde g_\eps$ be either of the approximations $\hat{g}_\eps, \check{g}_\eps$ of Lemma \ref{Le:approximating metrics}. For any compact subset $K \Subset \mathcal{M}$ and any integer $k \ge 0$, there exists a constant $C_{K,k} > 0$ independent of $\eps$ such that for all sufficiently small $\eps$,                  
\begin{equation} \label{eq:metric_diff}
    \| \nabla_h^k (\tilde{g}_\eps - g_\eps) \|_{\infty, K} \le C_{K,k}\, \eps,
\end{equation}
and for $k \in \{0, 1\}$, 
\begin{equation} \label{eq:inv_diff_01}
    \| \nabla_h^k \big(\tilde{g}_\eps^{-1} - g_\eps^{-1}\big) \|_{\infty, K} \le C_{K,k}\,\eps.
\end{equation}
Furthermore, for $k \ge 2$ we have
\begin{equation}
    \|\nabla_h^k \big(\tilde{g}_\eps^{-1} - g_\eps^{-1}\big) \|_{\infty, K} = \mathcal{O}\big((\omega_{\nabla_h g}(\eps)+\eps) \eps^{2-k}\big).
\end{equation}
\end{Lemma}

\begin{proof}
By \eqref{eq:M-convolution},
\begin{equation}\label{eq:M-conv-again}
    g_\eps = g \star_\mathcal{M} \rho_\eps = \sum_{\alpha \in \mathbb{N}} \chi_\alpha (\psi_\alpha)_*^{-1} \big( ((\psi_\alpha)_*(\xi_\alpha g)) * \rho_\eps \big) =: \sum_{\alpha \in \mathbb{N}} \chi_\alpha (\psi_\alpha)_*^{-1} ({}^\alpha g_\eps).
\end{equation}
To strictly widen or narrow the lightcones, a local modified metric ${}^\alpha\tilde{g}_\eps$ is constructed by adding a spatial correction term, cf.\ \cite[eq.\ (1.8)]{CG:12}
\begin{equation}
    {}^\alpha\tilde{g}_\eps := {}^\alpha g_\eps \pm \eta_\alpha(\eps)\ {}^\alpha v,
\end{equation}
where ${}^\alpha v = \sum_{i=1}^{n-1} {}^\alpha \nu_i \otimes {}^\alpha \nu_i$ is a fixed smooth $(0,2)$-tensor field defined on the chart. Here $\eta_\alpha(\eps) = C_\alpha \eps$ for some constant $C_\alpha > 0$, (cf.\ the mean value argument in \cite[Lem.\ 4.2]{Graf}, which requires $g\in C^1$).    A preliminary global approximating metric $\tilde{g}^{\mathrm{pre}}_\eps$ is then defined by patching these local metrics together
\begin{equation}
    \tilde{g}^{\mathrm{pre}}_\eps := \sum_{\alpha \in \mathbb{N}} \chi_\alpha (\psi_\alpha)_*^{-1} ({}^\alpha\tilde{g}_\eps) = g_\eps \pm \eps \sum_{\alpha \in \mathbb{N}} C_\alpha \chi_\alpha (\psi_\alpha)_*^{-1} ({}^\alpha v).
\end{equation}
To ensure the correct signature globally one then sets $\tilde{g}_\eps(p) := \tilde{g}^{\mathrm{pre}}_{u(\eps, p)}(p)$ for a certain smooth map $u$ such that on any compact subset $K \Subset \mathcal{M}$, $u(\eps, p) \equiv \eps$ for all sufficiently small $\eps$. Thus, on $K$, the metric $\tilde{g}_\eps$ coincides precisely with $\tilde{g}^{\mathrm{pre}}_\eps$. 
For such $\eps$ we therefore have on $K$
\begin{equation}
    \tilde{g}_\eps - g_\eps = \pm \eps \sum_{\alpha \in \mathbb{N}} C_\alpha \chi_\alpha (\psi_\alpha)_*^{-1} ({}^\alpha v) =: \pm \eps\, V.
\end{equation}
This definition yields $V$ as a fixed, globally defined smooth tensor field on $\mathcal{M}$, independent of $\eps$. Therefore,
\begin{equation}
    \|\nabla_h^k (\tilde{g}_\eps - g_\eps)\|_{\infty, K} = \eps \|\nabla_h^k V\|_{\infty, K} = C_{K,k} \eps
\end{equation}
for each $k\ge 0$, giving the first claim.

Next, let $A_\eps := \tilde{g}_\eps$ and $B_\eps := g_\eps$. By the above, $A_\eps - B_\eps = \pm \eps V$ on $K$. Using the matrix identity $A^{-1} - B^{-1} = -A^{-1}(A - B)B^{-1}$, we have
\begin{equation}\label{eq:W-eps}
    \tilde{g}_\eps^{-1} - g_\eps^{-1} = \mp \eps A_\eps^{-1} V B_\eps^{-1} =: \eps\, W_\eps,
\end{equation}
which leaves the task of bounding the covariant derivatives of $W_\eps$. 
By Lemma \ref{Le:approximating metrics} $B_\eps = g \star_\mathcal{M} \rho_\eps$ converges to the metric $g$ uniformly in $C^1$ on $K$. Therefore, for sufficiently small $\eps$, $B_\eps$ is uniformly non-degenerate, so $B_\eps^{-1}$ is uniformly bounded on $K$. Furthermore, since $g \in C^1$, the derivatives $\nabla_h B_\eps$ converge uniformly to $\nabla_h g$ on $K$ and are hence uniformly bounded. Differentiating the identity $B_\eps B_\eps^{-1} = I$ gives $\nabla_h(B_\eps^{-1}) = -B_\eps^{-1} (\nabla_h B_\eps) B_\eps^{-1}$, which therefore also is uniformly bounded.  

Also, $A_\eps = B_\eps \pm \eps V \to g$ uniformly, making $A_\eps$ uniformly non-degenerate and its inverse $A_\eps^{-1}$ uniformly bounded for small $\eps$. Similarly, $\nabla_h A_\eps = \nabla_h B_\eps \pm \eps \nabla_h V\to \nabla_h g$ uniformly, hence is uniformly bounded as well. Using $\nabla_h(A_\eps^{-1}) = -A_\eps^{-1} (\nabla_h A_\eps) A_\eps^{-1}$, we conclude uniform boundedness also of $\nabla_h(A_\eps^{-1})$.
Using \eqref{eq:W-eps}, we recover \cite[Cor.\ 4.3(i)]{Graf}, i.e.\ \eqref{eq:inv_diff_01} for $k=0$. For $k=1$, we note that
\begin{equation}
        \nabla_h W_\eps = \mp \big( \nabla_h(A_\eps^{-1}) V B_\eps^{-1} + A_\eps^{-1} (\nabla_h V) B_\eps^{-1} + A_\eps^{-1} V \nabla_h(B_\eps^{-1}) \big).
    \end{equation}
    By our above estimates this implies $\|\nabla_h W_\eps\|_{\infty,K} \le \tilde{C}_{K,1}$, so
    \begin{equation}
        \|\nabla_h \big(\tilde{g}_\eps^{-1} - g_\eps^{-1}\big)\|_{\infty,K} \le \tilde{C}_{K,1} \eps.
    \end{equation}
Finally let $k\ge 2$. From \eqref{eq:M-conv-again} we see that on $K$, $\nabla_h^k B_\eps$ expands into a finite linear combination of terms, each of which is a product of uniformly bounded smooth factors and Euclidean partial derivatives $\partial^m ({}^\alpha g_\eps)$ of order $0 \le m \le k$. Let $f_\alpha := (\psi_\alpha)_*(\xi_\alpha g)$. Because $g \in C^1$ and $\xi_\alpha$ is smooth with compact support in the chart, $f_\alpha$ is a compactly supported $C^1$ tensor field on $\mathbb{R}^n$, and ${}^\alpha g_\eps = f_\alpha * \rho_\eps$. For $m \le 1$, the partial derivatives $\partial^m ({}^\alpha g_\eps) = (\partial^m f_\alpha) * \rho_\eps$ are bounded uniformly in $\eps$.

For $m \ge 2$, passing one derivative to $f_\alpha$ and $m-1$ derivatives to the mollifier yields
\begin{equation}
        \partial^m ({}^\alpha g_\eps)(x) = \int_{\mathbb{R}^n} \partial f_\alpha(x-y) \partial^{m-1} \rho_\eps(y) \, dy.
    \end{equation}
    Since $\rho_\eps$ is supported in the ball $B_\eps(0)$, $\int \partial^{m-1} \rho_\eps = 0$. Thus we may write
    \begin{equation}
        \partial^m ({}^\alpha g_\eps)(x) = \int_{|y| \le \eps} \big( \partial f_\alpha(x-y) - \partial f_\alpha(x) \big) \partial^{m-1} \rho_\eps(y) \, dy
    \end{equation}
    and using $\int |\partial^{m-1} \rho_\eps| = \mathcal{O}(\eps^{1-m})$, this gives
    \begin{equation}
        \|\partial^m ({}^\alpha g_\eps)\|_{\infty} \le \sup_{|y|\le\eps} \|\partial f_\alpha(. -y) - \partial f_\alpha( . )\| \int |\partial^{m-1}\rho_\eps| = \mathcal{O}\big(\omega_{\partial f_\alpha}(\eps) \eps^{1-m}\big).
    \end{equation}
The tensor $f_\alpha$ is constructed from $g$ via multiplication by smooth cutoffs and pushforwards by smooth maps. Therefore, its derivatives $\partial f_\alpha$ depend algebraically on $g$, its first derivatives in coordinates, and smooth factors. It follows that the modulus of continuity of $\partial f_\alpha$ is bounded by a constant multiple of the
modulus of continuity of local first derivatives of $g$ plus $O(\eps)$-terms. Similarly, the modulus of these $\partial g$-terms can be estimated by that of $\nabla_h g$ plus $O(\eps)$-terms.
Altogether, we get $\omega_{\partial f_\alpha}(\eps) = \mathcal{O}(\omega_{\nabla_h g}(\eps) + \eps)$.
Consequently, for $m \ge 2$,
     \begin{equation}
        \|\partial^m ({}^\alpha g_\eps)\|_{\infty} = \mathcal{O}\big(\omega_{\nabla_h g}(\eps) \eps^{1-m} + \eps^{2-m}\big) = \mathcal{O}\big((\omega_{\nabla_h g}(\eps) + \eps) \eps^{1-m}\big).
    \end{equation}
Collecting terms and noting that the dominant $\eps$-power occurs for $m=k$, we obtain
\begin{equation}
        \|\nabla_h^k B_\eps\|_{\infty,K} = \mathcal{O}\big((\omega_{\nabla_h g}(\eps) + \eps) \eps^{1-k}\big).
\end{equation}
Since $A_\eps = B_\eps \pm \eps V$, we also conclude that $\|\nabla_h^k A_\eps\|_{\infty,K} = \mathcal{O}\big((\omega_{\nabla_h g}(\eps) + \eps) \eps^{1-k}\big)$.

Turning now to the inverses, we first note that the $k$-th derivative $\nabla_h^k (B_\eps^{-1})$ is a sum of terms of the form
\begin{equation}
        B_\eps^{-1} (\nabla_h^{m_1} B_\eps) B_\eps^{-1} \dots B_\eps^{-1} (\nabla_h^{m_j} B_\eps) B_\eps^{-1} \quad \text{with } \sum_{i=1}^j m_i = k,\ m_i \ge 1.
\end{equation}
When $j=1$, we immediately read off the bound $\mathcal{O}\big((\omega_{\nabla_h g}(\eps) + \eps) \eps^{1-k}\big)$. For terms with $j \ge 2$, the product consists of factors bounded by $\mathcal{O}(1)$ (when $m_i = 1$) and factors bounded by $\mathcal{O}\big((\omega_{\nabla_h g}(\eps) + \eps) \eps^{1-m_i}\big)$ (when $m_i \ge 2$). Let $p$ be the number of terms with $m_i \ge 2$. If $p=0$, the entire product is bounded by $\mathcal{O}(1)$. If $p > 0$, the product scales as $\mathcal{O}\big((\omega_{\nabla_h g}(\eps) + \eps)^p \eps^{p - \sum_{m_i \ge 2} m_i}\big)$. Since the $j-p$ terms of order $1$ contribute exactly $j-p$ to the sum $k$, we have $\sum_{m_i \ge 2} m_i = k - (j - p)$. The exponent of $\eps$ is thus $p - (k - j + p) = j - k$. Because $j \ge 2$ and $p \ge 1$, we have $j - k \ge 2 - k$. Since $\omega_{\nabla_h g}(\eps) + \eps \to 0$ as $\eps \to 0$, any such product is bounded by $\mathcal{O}\big((\omega_{\nabla_h g}(\eps) + \eps) \eps^{j-k}\big) \le \mathcal{O}\big((\omega_{\nabla_h g}(\eps) + \eps) \eps^{2-k}\big)$. This is subdominant to the $j=1$ term $\mathcal{O}\big((\omega_{\nabla_h g}(\eps) + \eps) \eps^{1-k}\big)$. Consequently, 
\begin{equation}
        \|\nabla_h^k (B_\eps^{-1})\|_{\infty, K} = \mathcal{O}\big((\omega_{\nabla_h g}(\eps) + \eps) \eps^{1-k}\big).
\end{equation}
Again we obtain the same asymptotics for the derivatives of $A_\eps^{-1}$:  $\|\nabla_h^k (A_\eps^{-1})\|_{\infty, K} = \mathcal{O}\big((\omega_{\nabla_h g}(\eps) + \eps) \eps^{1-k}\big)$.

Finally, the $k$-th derivative of $W_\eps = \mp A_\eps^{-1} V B_\eps^{-1}$ (i.e., the term we have to ultimately estimate) decomposes into a sum of terms of the form
    \begin{equation}
        \nabla_h^{k_1} (A_\eps^{-1}) \nabla_h^{k_2} V \nabla_h^{k_3} (B_\eps^{-1}), \quad \text{with } k_1 + k_2 + k_3 = k.
    \end{equation}
Here, $\nabla_h^{k_2} V$ contributes $\mathcal{O}(1)$. If either $k_1 = k$ or $k_3 = k$ then by the above we obtain a bound of $\mathcal{O}\big((\omega_{\nabla_h g}(\eps) + \eps) \eps^{1-k}\big)$. 

For any other combination, $k_1 < k$ and $k_3 < k$. If both $k_1, k_3 \ge 2$, the product scales as $(\omega_{\nabla_h g}(\eps) + \eps)^2 \eps^{2-k_1-k_3} \le (\omega_{\nabla_h g}(\eps) + \eps)^2 \eps^{2-k} \le \mathcal{O}\big((\omega_{\nabla_h g}(\eps) + \eps) \eps^{2-k}\big)$. If only one index is $\ge 2$ (say $k_1 \ge 2$ and $k_3 \le 1$), the product scales as $(\omega_{\nabla_h g}(\eps) + \eps) \eps^{1-k_1} \cdot \mathcal{O}(1) \le \mathcal{O}\big((\omega_{\nabla_h g}(\eps) + \eps) \eps^{2-k}\big)$ since $k_1 \le k-1$. If both $k_1, k_3 \le 1$, the product is $\mathcal{O}(1)$. All of these cases are dominated by $\mathcal{O}\big((\omega_{\nabla_h g}(\eps) + \eps) \eps^{1-k}\big)$.
Hence 
\begin{equation}
\|\nabla_h^k W_\eps\|_{\infty,K} = \mathcal{O}\big((\omega_{\nabla_h g}(\eps) + \eps) \eps^{1-k}\big),
\end{equation}
and we finally arrive at
\begin{equation}
        \|\nabla_h^k \big(\tilde{g}_\eps^{-1} - g_\eps^{-1}\big)\|_{\infty,K} = \eps\, \|\nabla_h^k W_\eps\|_{\infty,K} = \mathcal{O}\big((\omega_{\nabla_h g}(\eps) + \eps)\eps^{2-k}\big).
\end{equation}
\end{proof}

\section{RT-regularisation} \label{Sec_RT}

In this section we detail how RT-regularisation allows us to raise the regularity of the metric by one order from $W^{1,p}_\text{loc}$ to $W^{2,p}_\text{loc}$ in an $W^{2,p}_\text{loc}$-related atlas given an $L^p_\text{loc}$-bound on the curvature. Some background on the theory of RT-regularisations is deferred to Appendix \ref{Sec_RT_appendix}.

Let $\mathcal{M}$ be a manifold with smooth atlas $\mathcal{A}$, and assume $(\M,\A)$ is endowed with a Lorentzian metric $g$ of regularity $W^{1,p}_{\A,\text{loc}}$, (i.e., $g_x \in W^{1,p}_\text{loc}(\Omega_x)$ component-wise for every chart $(x,\Omega)$ contained in the atlas $\A$, where $g_x=x_*g$ denotes the metric represented in coordinates $x$, and $\Omega_x \equiv x(\Omega) \subseteq \R^n$); for brevity, we often refer to $(\M,\A,g)$ as a $W^{1,p}_{\A,\text{loc}}$-spacetime.  Moreover, we assume the Riemann curvature tensor of $g$ is controlled in $L^p_\text{loc}$, that is, $\Riem[g_x] \in L^p_\text{loc}(\Omega_x)$ in each chart $(x,\Omega)$ of the atlas $\mathcal{A}$, which we write as $\Riem[g]\in L^p_{\mathcal{A},\text{loc}}$. 

We can now state the theorem on regularisations of metrics by the RT-equations, which is the starting point of the methods in this paper.

\begin{Thm}[RT-regularisation] \label{Thm_smoothing_preliminaries}  
On a smooth manifold $(\M,\A)$ assume we are given a Lorentzian metric $g \in W^{1,p}_{\A,\text{loc}}$, where $p>{\rm max}\{4,n\}$, $p<\infty$.  Then there exists a smooth atlas $\A'$, which is $W^{2,p}_\text{loc}$-compatible with $\A$, such that $g \in W^{2,p}_{\A',\text{loc}}$, if and only if $\Riem[g]\in L^p_{\mathcal{A},\text{loc}}$.
\end{Thm}   

Here we say that $\A'$ is $W^{2,p}_\text{loc}$-compatible with $\A$, if $\A'$ is a subatlas of the maximal $W^{2,p}_\text{loc}$-extension of $\A$.
In \cite{ReintjesTemple_essreg} the authors actually prove 
this equivalence for general affine connections.\footnote{Theorem \ref{Thm_smoothing_preliminaries} is proven in \cite{ReintjesTemple_essreg} at the level of general affine connections $\Gamma$ in $L^p_\text{loc}$. That is, $\Riem(\Gamma)\in L^p_{\mathcal{A},\text{loc}}$ is indeed a necessary and sufficient for the regularisation of $\Gamma$ to $W^{1,p}_\text{loc}$. Moreover, it is proven in \cite{ReintjesTemple_essreg}, that an implicit regularisation of the curvature takes place in each step until the connection reaches its highest possible regularity (its {\it essential regularity}) in a suitable subatlas. These results are based on the RT-equations \cite{ReintjesTemple_ell1, ReintjesTemple_ell2, ReintjesTemple_ell4}, elliptic equations whose solution furnish the coordinate transformations regularising connections to one derivative above its curvature, (see Appendix \ref{Sec_RT_appendix} for more details). Since a metric is always exactly one derivative more regular than its connection, these results directly imply Theorem \ref{Thm_smoothing_preliminaries}.}

The regularisation of Theorem \ref{Thm_smoothing_preliminaries}, which is our main tool, can be viewed as inducing a diffeomorphism (and automatically an isometry)
\begin{equation}\label{eq:phi3.1}
 \phi=\mathrm{id}:\, (\M,\A,g) \longrightarrow (\M,\A',g)
\end{equation}
of regularity $W^{2,p}_\text{loc}(\M,\M')$ between the $W^{1,p}_\text{loc}$-spacetime $(\M,\A,g)$ and the $W^{2,p}_\text{loc}$-spacetime $(\M,\A',g)$. We will often refer to $\phi$ as the \emph{RT-regularisation}  
of the $W^{1,p}_\text{loc}$-spacetime $(\M,\A,g)$. Observe that again by Morrey's inequality $\phi$ is a $C^{1,\alpha}_\text{loc}$-diffeomorphism.

At the low regularity of $W^{1,p}_\text{loc}$-metrics addressed in this paper, it is not clear how to define geodesic curves in a standard sense, since the associated $L^p$-connection regularity is too low to restrict connections to curves and hence too low for the geodesic equation to make sense, even when taking derivatives in the weak sense of distributions. Moreover, the notion of geodesic curves as solutions in the sense of Fillipov, which has been useful in Lipschitz spacetimes, see e.g.\ \cite{PSSS:14,PSSS:16,LLS:21,CGHKS25} does not appear to apply to connections less regular than $L^\infty_\text{loc}$. In this paper, we thus turn to the notion of geodesic curves introduced in \cite{ReintjesTemple_geod} based on the RT-regularisation, and to which we refer to as ``{\it RT-geodesics}''.\footnote{In \cite{ReintjesTemple_geod}, the authors introduced Definition \ref{Def_weak_geod} in a local sense and worked out existence and uniqueness results for RT-geodesics, as well as stability of RT-geodesics under mollification of the underlying connection.} 

\begin{Def}[RT-geodesics] \label{Def_weak_geod}
Let $(\M,\A,g)$ be a $W^{1,p}_{\A,\text{loc}}$-spacetime, where $p>{\rm max}\{4,n\}$, $p<\infty$. We call a curve $\gamma \in C^0_\A(I,\M)$ an \emph{RT-geodesic} in $(\M,\A,g)$, if there exists a smooth atlas $\A'$ which is $W^{2,p}_\text{loc}$-compatible to $\A$ and with $g \in W^{2,p}_{\A',\text{loc}}$, such that $\gamma$ is a classical solution of the geodesic equation in $(\M,\A',g)$. We call an RT-geodesic \emph{complete} if its affine parameter defined in $(\M,\A',g)$ can be extended to $(-\infty,\infty)$, and we call it \emph{incomplete} otherwise.
\end{Def}

Observe that the classical solutions alluded to in the definition are \emph{non-unique} since the connection regularity lies below Lipschitz, which would guarantee unique solvability of the initial value problem. 
The following lemma establishes regularity of RT-geodesics and shows that the notion is independent of the regularising atlas $\A'$.

\begin{Lemma}[Regularity for RT-geodesics]\label{Lem:reg-RT-geos}
Let $(\M,\A,g)$ be a $W^{1,p}_{\A,\text{loc}}$-spacetime with $\Riem[g]\in L^p_{\mathcal{A},\text{loc}}$, $p>{\rm max}\{4,n\}$, $p<\infty$, and let $\gamma$ be an RT-geodesic in $(\M,\A,g)$. Then the following holds: 
\begin{itemize}
    \item[(i)] The RT-geodesic $\gamma$ has regularity $C^{2,\alpha}_{\A'}(I,\M)$ for $\alpha = 1 -n/p$, and regularity $C^{1,\alpha}_{\A}(I,\M)$ with respect to the original atlas $\A$.  
    \item[(ii)] If $\A''$ is another $C^\infty$-atlas (or $W^{3,p}$-atlas), which is $W^{2,p}$-compatible to $\A$, such that  $g \in W^{2,p}_{\A''}$, then $\gamma$ is a classical solution of the geodesic equation in $(\M,\A'',g)$. Moreover, $\gamma$ is complete in $(\M,\A'',g)$ if and only if $\gamma$ is complete in $(\M,\A',g)$.
\end{itemize}
\end{Lemma}

\Proof 
(i): The regularity $\gamma \in C^{2,\alpha}_{\A'}(I,\M)$ follows directly from the geodesic equation with metric $g \in W^{2,p}_{\A'}\subseteq C^{1,\alpha}$, while regularity $\gamma \in C^{1,\alpha}_{\A}(I,\M)$ follows since $\A$ and $\A'$ are $W^{2,p}$-compatible. 

(ii): Assume $g \in W^{2,p}_{\A''}$ with respect to another smooth atlas $\A''$, such that $\A''$ is $W^{2,p}$ compatible to $\A$. Thus, $\A'$ and $\A''$ are both subatlasses of the maximal $W^{2,p}$-extension $\A^\text{max}$ of $\A$. 
By Lemma 4.1 in \cite{ReintjesTemple_essreg}, the metric regularity $g \in W^{2,p}_{\A''}$ and $g \in W^{2,p}_{\A'}$ implies that all transition maps in $\A^\text{max}$ from $\A'$ to $\A''$ are $W^{3,p}_\text{loc}$-regular (hence also $C^{2,\alpha}_\text{loc}$ for $\alpha=1-n/p$); in particular, the identity is a $W^{3,p}$-diffeomorphism between $(\M,\A',g)$ and $(\M,\A'',g)$. 
This regularity is sufficient to map classical geodesics in $(\M,\A',g)$ to classical geodesics in $(\M,\A'',g)$, and vice versa, by the identity map. Thus, since by assumption $\gamma$ is a classical geodesic curve in $(\M,\A',g)$, it is also a classical geodesic curve in $(\M,\A'',g)$. 

To prove the supplement, assume $\gamma$ is a complete geodesic in $(\M,\A',g)$. That is, its affine parameter can be extended to $(-\infty,\infty)$. Thus, mapping $\gamma$ from $(\M,\A',g)$ to $(\M,\A'',g)$ with the identity map, which is a $W^{3,p}_\text{loc}$-diffeomorphism, it follows that $\gamma$ is a complete geodesic in $(\M,\A'',g)$.  
\QED

\begin{Def}[RT-completeness]\label{def:RT-completeness} Let $(\M,\A,g)$ be a $W^{1,p}_{\A,\text{loc}}$-spacetime and assume $\Riem[g]\in L^p_{\mathcal{A},\text{loc}}$, $p>{\rm max}\{4,n\}$, $p<\infty$. We call $(\M,\A,g)$ (timelike/causally/null) \emph{RT-complete} if every (timelike/causal/null) RT-geodesic $\gamma$ in $M$ is complete. Otherwise $(\M,\A,g)$ is called 
(timelike/causally/null) \emph{RT-incomplete}.
\end{Def}  

By Lemma \ref{Lem:reg-RT-geos}, these notions are independent of the chosen RT-regularisation.

\section{Causality}  \label{Sec_causality}

Recent studies in Lorentzian causality (e.g.\ \cite{CG:12,Minguzzi-cone-structures,GrantKuSaSt}) have revealed that local Lipschitz continuity of the metric is the threshold regularity which allows for a satisfactory theory. Below that, several pillars of causality break down in general, namely the push up property ($I^+\circ J^+\subseteq I^+$ and vice versa) and the openness of $I^+$. Also light cones cease to be hypersurfaces and ``bubble up'' to have non-vanishing interior.  

More specifically, \cite{CG:12} isolates the class of \emph{causally plain} metrics for which no bubbling occurs, for details see their Def.\ 1.16 or the characterisation in \cite[Cor.\ 2.16]{GrantKuSaSt}. Then also $I^+$ is open and the push up principle holds, see the diagramm on p.\ 101 in \cite{GrantKuSaSt}. Finally,  \cite[Cor.\ 1.17]{CG:12} establishes that any $g\in C^{0,1}_\text{loc}$ is causally plain. 

In this section we will use RT-regularisation as a tool in causality theory. First we prove that RT-regularisation boosts metrics $g\in W^{1,p}_\text{loc}$ with $\Riem[g]\in L^p_\text{loc}$ above the critical causality threshold by showing that they are indeed causally plain. Then we will use the same tool to show that maximisers coincide with RT-geodesics which have regularity $C^{1,\alpha}$ and do not branch.
\bigskip

We begin by considering basic notions of causality under $C^1$-diffeomorphisms, $\phi \in C^1(\M,\M')$, between two Lorentzian manifolds, $(\M,g)$ and $(\M',g')$, both endowed with $C^\infty$-atlases $\A$ and $\A'$, and $C^0$-metrics $g$ and $g'$. 

If $\phi$ is an isometry ($\phi_*g=g'$), then clearly the causal character of locally Lip\-schitz curves is preserved, as is the length of causal curves $\gamma:[a,b]\to M$, i.e., $ L_g(\gamma) = L_{\phi_*g}(\phi\circ \gamma)$. Consequently we obtain for the time separation functions that 
\begin{equation}\label{eq:tau=tau}
    \tau_\M(p,q) = \tau_{\M'}(\phi(p),\phi(q))
\end{equation} and the chronological as well as the causal futures and pasts of points are preserved, which we now make explicit.

\begin{Lemma}[Preservation of causality]\label{prop:I^+}
    Any diffeomorphim $\phi \in C^1(\M,\M')$ faithfully translates all causality notions in the sense that $\forall q\in\M$,
    \begin{align}
        \phi\left(J^{\pm}_{g}(q)\right)=J^{\pm}_{\phi_* g} \left(\phi(q)\right), \\
        \phi\left(I^{\pm}_{g}(q)\right)=I^{\pm}_{\phi_* g} \left(\phi(q)\right).
    \end{align}
\end{Lemma}

The above properties hold even when $\phi$ 
is a local {bi-}Lipschitz homeomorphism but we are now interested in the case where $\phi$ is the RT regularisation of \eqref{eq:phi3.1} hence $C^{1,\alpha}_\text{loc}$. We then obtain the following results.

\begin{Prop}[Causal plainness]\label{prop:cp}
Let $(\M,g)$ be a $W^{1,p}_{\text{loc}}$-spacetime and assume $\Riem[g]\in L^p_{\text{loc}}$, $p>{\rm max}\{4,n\}$, $p<\infty$.
Then $(\M,g)$ is \textit{causally plain}.
\end{Prop}
\begin{proof}
    By Theorem \ref{Thm_smoothing_preliminaries} $\phi=\mathrm{id}$ is a $ W^{2,p}_\text{loc}$-isometry from $(\mathcal{M},\mathcal{A})$ to $(\mathcal{M},\mathcal{A'})$, cf.\ \eqref{eq:phi3.1}. In particular, $\phi_* g\in W^{2,p}_\text{loc}\subseteq C^{1,\alpha}_\text{loc}\subseteq C^{0,1}_{\text{loc}}$. Then $\phi_* g$ is causally plain by \cite[Cor.\ 1.17]{CG:12}. Hence no bubbling occurs in $(\mathcal{M,A'})$ and so by Lemma \ref{prop:I^+} no bubbling occurs in $(\mathcal{M,A})$ either, rendering $g$ causally plain.
\end{proof}

\medskip

We now turn to the issue of geodesics vs.\ maximisers in low regularity. Indeed below the threshold of $g \in C^{1,1}_\text{loc}$ the usual equivalence fails. Classical examples show that solutions of the geodesic equations need not be local maximisers already for $g$ locally in the Hölder class $C^{1,\alpha}$ for any $0<\alpha<1$, see \cite{HW:51,SS:18}. Conversely, maximisers (which exist already for $g\in C^0$) are Filippov-geodesics of regularity $C^{1,1}$, once we have $g\in C^{0,1}_\text{loc}$, see \cite[Thm.\ 1.1]{LLS:21} and the discussion following it. Also, such maximisers possess a causal character, i.e, they are either timelike or null throughout \cite{GL:18,LLS:21}. We next show that RT-regularisation allows for an analogous result. 

\begin{Prop}[Maximisers are RT-geodesics] \label{Prop_maximisers->geodesics}
Let $(\M,g)$ be a $W^{1,p}$-spacetime with $\Riem[g] \in L^p$, $p>{\rm max}\{4,n\}$, $p<\infty$. Suppose that $\gamma:[a,b]\to \M$ is a causal maximiser between its endpoints. Then $\gamma$ is of regularity $C^{1,\alpha}$ with $\alpha= 1-n/p$, and $\gamma$ is an RT-geodesic. 
\end{Prop}

\begin{proof} 
Applying the RT-regularisation procedure of Theorem \ref{Thm_smoothing_preliminaries}, cf.\ \eqref{eq:phi3.1}
and noting that the length of curves as well as the time separation $\tau$
is preserved we find that $\tilde \gamma :=\phi\circ \gamma$ is a maximiser in the $C^1$-spacetime  $(\M,\phi_*g)$. Thus by \cite[Thm.\ 3.3]{benedikt}, $\tilde\gamma$ solves the geodesic equation for $\phi_*g$, hence in particular is of regularity $C^2$. Now, recalling that
$\phi$ is a $C^{1,\alpha}_\text{loc}$-isometry we find
that $\gamma=\phi^{-1}\circ\tilde\gamma \in C^{1,\alpha}_\text{loc}$. Finally, $\gamma=\phi^{-1}\circ\tilde\gamma$ is obviously an RT-geodesic in the sense of Definition \ref{Def_weak_geod}.
\end{proof}

We close this section by showing that RT-regularisation also allows us to derive a non-branching result for geodesics and maximisers however, under additional regularity of the curvature. Here we say that a causal maximiser (or geodesic) $\gamma:[0,1] \to\M$ branches if there exists another maximizing (or geodesic) causal curve $\sigma: [0,1] \to\M$ such that $\gamma(t)=\sigma(t)$ for all $0 \leq t \leq a$ for some $0 < a < 1$ and $\gamma(t)\not=\sigma(t)$ for all $a< t\leq 1$.

\begin{Prop}[Non-branching]\label{prop:nb}
Let $(\M,g)$ be a $W^{1,p}_\text{loc}$-spacetime with $\Riem[g]\in W^{1,p}_\text{loc}$, $p>{\rm max}\{4,n\}$, $p<\infty$, then the following holds: 
\begin{enumerate}
    \item The initial value problem for RT-geodesics is uniquely solvable in $(\M,g)$.
    \item Causal RT-geodesics in $(\M,g)$ are non-branching.
    \item Causal maximisers in $(\M,g)$ are non-branching.  
\end{enumerate}     
\end{Prop}

\begin{proof}  
(i): RT-geodesics for $W^{1,p}_\text{loc}$-metrics are by default non-unique, however, as proven in Theorem 2.3 in \cite{ReintjesTemple_geod}, the condition $\Riem[g]\in W^{1,p}_\text{loc}$ implies uniqueness of RT-geodesics in $(\M,\mathcal{A},g)$. This essentially follows because a second regularisation by the RT-equation establishes metric regularity $W^{3,p}_\text{loc} \subseteq C^{2,\alpha}_\text{loc}$ with respect to some new $C^\infty$-atlas $\mathcal{A}'$ which is $W^{2,p}_\text{loc}$-compatible to $\mathcal{A}$ (see \cite[Thm. 2.3]{ReintjesTemple_essreg} for the global regularisation) and this places the metric above the $C^{1,1}_\text{loc}$-threshold regularity for unique solvability of the initial value problem of the geodesic equation. 

(ii): Uniqueness of RT-geodesics (as solutions to the initial value problem of the geodesic equation in the sense of Definition \ref{Def_weak_geod}), directly implies that RT-geodesics are non-branching. More precisely, if two RT-geodesics were branching, they would be identical on some open interval $(t_{-},t_0)$, but distinct on the open interval $(t_0,t_+)$. By $C^1$-regularity of the two RT-geodesics (cf. Definition \ref{Def_weak_geod}), both give rise to the same initial data at $t_0$. Thus, by uniqueness of RT-geodesics, both geodesic are indentical on some open interval $(t_0,t_0^+) \subseteq (t_0,t_+)$, which is a contradiction. 

(iii): By Proposition \ref{Prop_maximisers->geodesics}, maximisers are $C^{1,\alpha}$-RT-geodesics, so the claim follows from (ii).
\end{proof}

\section{Mean curvature bounds from time functions}  \label{Sec_mean_curv_timefunction}

The metric regularities $W^{1,p}_\text{loc}$ addressed in this paper are too low to define the mean curvature (even) for a smooth hypersurface $\Sigma$ because of a loss of regularity under restriction to hypersurfaces entailed by the Trace Theorem (namely  $g|_\Sigma\in W^{1-1/p,p}_\text{loc}$).

Moreover, applying an RT-regularsation does not solve the problem. While then the metric regularity is improved to (above) $C^1$, which would allow for the classical definition of mean curvature for \emph{smooth} hypersurfaces, we face the problem that after RT-regularisation $\Sigma$ is only of regularity $C^{1,\alpha}_\text{loc}$, hence lies below the $C^2$-threshold which allows for the classical treatment.

In the following we develop an approach based on smearing out $\Sigma$, assumed to be a smooth spacelike Cauchy surface using a Cauchy temporal function\footnote{Observe that the approach using the flow-out of $\Sigma$ under a transversal vector field as used in \cite[Sec.\ 4]{CGHKS25} heavily depends on the smoothness of $\Sigma$, which does not hold in our situation, cf.\ also the discussion in the Introduction.}. We will first review the smooth case, and then address the case of a $W^{1,p}_\text{loc}$-metric. This will allow us to formulate a suitable mean curvature condition in the Hawking theorem  for $g\in W^{1,p}_\loc$. We will then treat the case $g\in W^{2,p}_\text{loc}$, $\Sigma \in W^{2,p}_\text{loc}$ (appearing after the RT-regularisation). Finally, we further mollify the $W^{2,p}_\text{loc}$-metric and the Cauchy temporal function in order to make the corresponding notion of mean curvature applicable in our proof of the Hawking theorem. 

\subsection{Smeared out mean curvature: the smooth case}
To motivate our concept of mean curvature bounds for low regularity metrics, we first consider the smooth setting. Thus let $(\M,g)$ be a smooth globally hyperbolic spacetime with smooth spacelike Cauchy surface $\Sigma$. By a \emph{Cauchy temporal function} we mean a $C^1$-function $\T:M\to\R$ with past-directed timelike gradient everywhere, whose level sets $\Sigma_t:=\T^{-1}(t)$ are Cauchy hypersurfaces for all $t\in \R$.  The latter condition can equivalently be replaced by requiring $\T\circ\gamma$ to be onto for all inextendible causal curves $\gamma$. By \cite[Thm.\ 1.2] {BS06}, there exists a smooth Cauchy temporal function $\T$ on $M$ such that $\Sigma = \Sigma_0 = \T^{-1}(0)$. 

We define 
\beq\label{eq:N-def}
N :=  -\,\frac{1}{\|{\rm grad}_g(\T)\|_g}\ {\rm grad}_g (\T),
\eeq 
which is a timelike future-pointing unit vector field perpendicular to each level set $\Sigma_t$. Now given some open set $U$, we introduce a ``smeared-out'' version of the mean curvature of $\Sigma\cap U$, as follows: Choose a smooth local frame $(X_i)_{i=1,...,n-1}$ (not necessarily orthonormal) in a neighbourhood $\Omega_q\subseteq U$ of any point $q\in U$ such that $X_i \perp_{g} N$ in $\Omega_q$ for all $i$. 
Then set 
\beq \label{mean-curv_smear-smooth}
\mathcal{H}_{g,\T,U} :=  - \sum_{i,j=1}^{n-1} G^{ij}\ g(\nabla_{X_i}X_j , N) \in C^\infty(\Omega_q),
\eeq 
where $G^{ij} = \Big(\big(g(X_l,X_m)\big)^{-1}\Big)^{ij}$ is the point-wise inverse of the matrix with components $g(X_i,X_j)$ (cf.\ \eqref{eq:mean-curvature-non-orth}). Then clearly we obtain a well-defined function $\mathcal{H}_{g,\T,U}\in C^\infty(U)$. The restriction of $\mathcal{H}_{g,\T,U}$ to $\Sigma$ is then the restriction to $\Sigma\cap U$ of the classical mean curvature $H_\Sigma$ of $\Sigma$ in $M$, c.f.\ \eqref{eq:mean-curvature-non-orth}, and so by continuity we immediately obtain:

\begin{Lemma}[Smeared out mean curvature in the smooth case]\label{lem:smeared-mean-curv-smooth} 
Let $(\M,g)$ be a smooth spacetime and let $\Sigma=\T^{-1}(0)$ be a smooth spacelike Cauchy surface,  with  $\T$ a smooth Cauchy temporal function. Let $b\in \R$. Then the following are equivalent:
\begin{itemize}
\item[(i)]  $H_\Sigma > b$ (resp.\ $<b$).
\item[(ii)] There exists an open neighbourhood $U$ of $\Sigma$ such that $\mathcal{H}_{g,\T,U}>b$, (resp.\ $<b$).
\end{itemize}
\end{Lemma}
Note that if $\Sigma$ is compact, then a neighbourhood $U$ as above can always be chosen of the form $U=\T^{-1}(t_-,t_+)$, for some appropriate $t_- < 0 < t_+$.
Also observe that the above bounds are independent of the choice of $\T$.

\subsection{Smeared out mean curvature: the rough case}

Turning now to the low-regularity situation, suppose that $g \in W^{1,p}_\text{loc}(\M)$, and assume from here on that $p>2n$. This is slightly more regularity than used in Theorem \ref{Thm_smoothing_preliminaries} but will be needed for approximation by convolution in Lemma \ref{lem:roland_tensor} below and hence for the main result in this Section, Theorem \ref{Prop:Mean-curvature-convergence}.

Since $g$ is continuous it defines a proper cone structure by Remark \ref{rem:cs}.
Moreover, due to \cite[Prop.\ 2.23, Thm.\ 2.51]{Minguzzi-cone-structures}, $g$ induces a closed Lorentz-Finsler structure on $M$. We are thus in a position to apply \cite[Thm.\ 2.14]{Minguzzi-Cauchy-surface} to conclude that, just as in the smooth setting discussed above, for any smooth spacelike Cauchy surface $\Sigma$ in $M$ there exists a smooth Cauchy temporal function $\T: \M\to \R$ such that $\Sigma= \Sigma_0 = \T^{-1}(0)$. This result was also established independently in \cite[Prop.\ 2.4]{bernard_suhr}.

We shall base our definition of mean curvature bounds on $\Sigma$ on an analogue of \eqref{mean-curv_smear-smooth} as follows: Define the normal vector field $N$ as in \eqref{eq:N-def} and note that $N$ is of the same regularity as $g$, that is, $N\in W^{1,p}_\text{loc}$. 
Now, fixing some open set $U$ as before we can introduce on $U\cap\Sigma$ a smeared-out version of the mean curvature of $\Sigma$
, in this low-regularity situation, as follows: Choose a local frame $(X_i)_{i=1,...,n-1}$ in a neighbourhood $\Omega_q \subseteq U$ of any point $q\in U$ such that for all $i$
\beq\label{eq:W1p-frame} 
X_i \perp_{g} N \qquad \text{and} \qquad X_i \in W^{1,p}_\text{loc}(\Omega_q);
\eeq
we show in Lemma \ref{Lemma_mean-curv_smeared} below that such a local frame always exists. Analogously to \eqref{mean-curv_smear-smooth}, we now define the {\it smeared-out mean curvature} in $\Omega_q$ by
\beq \label{mean-curv_smear}
\mathcal{H}_{g,\T,U}:=  - \sum_{i,j=1}^{n-1} G^{ij}\ g(\nabla_{X_i}X_j , N) \in L^p_\text{loc}(\Omega_q).
\eeq 

\begin{Lemma}[Mean curvature for smooth $\Sigma$ and $g\in W^{1,p}_\text{loc}$] \label{Lemma_mean-curv_smeared}
The smeared-out mean curvature $\mathcal{H}_{g,\T,U}$, as defined in \eqref{mean-curv_smear}, exists for any open neighbourhood $U$ of $\Sigma$, lies in $L^p_\text{loc}(U)$, and is independent of the choice of local frame $(X_i)_{i=1,...,n-1}$.
\end{Lemma}

\Proof
To prove existence, we have to show that one can indeed always locally introduce a frame $(X_i)_{i=1,...,n-1}$ as in \eqref{eq:W1p-frame}. Fixing $q\in U$, let $(x^1,\dots,x^n)$ be local coordinates around $q$ and pick $v_i =\sum_{j=1}^n \alpha_i^j \partial_{x^j}|_q$ orthogonal to $N_q$ such that $v_1,\dots,v_{n-1},N_q$ forms a basis of $T_qM$. Let $\tilde X_i := \sum_{j=1}^n \alpha_i^j \partial_{x^j}$ be the constant coefficient extension of $v_i$ to the chart domain of $x$ and set
\begin{equation}\label{eq:5.5}
    X_i := \tilde X_i + g(\tilde X_i,N)N \qquad (i=1,\dots,n-1). 
\end{equation}
Then the $X_i$ are perpendicular to $N$, have the same regularity as $g$, namely $W^{1,p}_\text{loc}$, and since $X_i(q)=v_i$, there exists a neighbourhood $\Omega_q$ of $q$ in $U$ such that $X_1,\dots,$ $X_{n-1},N$ is 
a frame on $\Omega_q$.
The $L^p$-regularity of $\mathcal{H}_{g,\T,U}$ follows directly from the regularity of its constituents.
Finally, the independence of $\mathcal{H}_{g,\T,U}$ of the choice of frame follows from the standard calculation, now taken in the sense of $L^p$-functions. 
\QED

With these building blocks in place, we now define mean curvature bounds in the low regularity regime in analogy to
the smooth setting as follows (cf.\ Lemma \ref{lem:smeared-mean-curv-smooth}):

\begin{Def}[Mean curvature bounds]\label{Def:mean-curvature-bound} Let $(\M,g)$ be a $W^{1,p}_\text{loc}$-spacetime with smooth spacelike Cauchy surface $\Sigma$.  We say that the mean curvature of $\Sigma$ is bounded above (resp.\ below) by $\beta\in \R$ if 
there exists a smooth Cauchy temporal function $\T$ with $\Sigma=\T^{-1}(0)$ and a neighbourhood $U$ of $\Sigma$ in $M$ such that 
\begin{equation}\label{eq:ess-sup_b}
 \esssup \mathcal{H}_{g,\T,U}<\beta \qquad (\text{resp.} >\beta). 
\end{equation}
\end{Def}

Observe that, contrary to the smooth case (cf.\ Lemma \ref{lem:smeared-mean-curv-smooth}), while as remarked above there always is 
a smooth temporal function $\T$ realizing $\Sigma$ as its $0$-level set, bounds as in \eqref{eq:ess-sup_b} will in general depend
on the concrete choice of such a $\T$.

Finally we apply an RT-regularisation and observe that curvature bounds transfer to the corresponding scenario.
\begin{Remark}[Mean curvature bounds for $\Sigma\in W^{2,p}_\text{loc}$ and $g\in W^{2,p}_\text{loc}$]\label{rem:cb}
By Theorem \ref{Thm_smoothing_preliminaries} and using the notation of \eqref{eq:phi3.1} 
we have that $\T$ is a Cauchy temporal function but now of regularity $W^{2,p}_\text{loc}(\M,\R)$ (w.r.t.\ $\A'$) and hence $N \in W^{1,p}_\text{loc}(\M)$ and $X_i \in W^{1,p}_\text{loc}(\M)$ for $i=1,...,n-1$. Indeed $N$ again is future-pointing timelike and $\T \circ \gamma : \R \to \R $ is onto for all $g$-causal curves.

Moreover, one can again introduce the smeared-out mean curvature \eqref{mean-curv_smear}, and since \eqref{mean-curv_smear} indeed transforms as a scalar under any diffeomorphism $\phi \in W^{2,p}_\text{loc}$,
we conclude that the ess-sup bound \eqref{eq:ess-sup_b} is preserved. 
\end{Remark}

\subsection{Mollification of mean curvature}   

In order to produce a situation to which the classical Hawking theorem applies, we next mollify the $W^{2,p}_\text{loc}$-metric obtained previously by regularisation via the RT-equations. Note that a direct mollification of the original $W^{1,p}_\text{loc}$-metrics would not be sufficient, cf.\ the precise account of the required Sobolev regularities in Remark \ref{rem:proof-strategy} (ii) below. 

To begin with, let $\check g_\eps$ be an approximation of $g \in W^{2,p}_\text{loc}$ as in Lemma \ref{Le:approximating metrics}. For $\T \in W^{2,p}_\text{loc}$ we now show the following:

\begin{Lemma} \label{Lemma_T-eps_timefct}
There exists a smooth map $(0,1] \times M \to \R$, $(\eps,p) \mapsto \T_\eps(p)$ with the following properties:
\begin{itemize}
\item[(i)]  For each compact set $K\Subset M$ there exists $\eps_K >0$ such that\footnote{The reason for the $\sqrt{\eps}$-shift in $\T_\eps$ will become apparent in the proof of Theorem \ref{Th:Hawking-I} below.} 
\begin{equation}\label{eq:T-eps}
\T_\eps(p) = \T\star_M \rho_\eps(p) +\sqrt{\eps}
\end{equation}
for all $\eps\in (0,\eps_K]$ and all $p\in K$.
\item[(ii)] $|\T_\eps - \T| <1$ on $M$.
\item[(iii)] $\T_\eps \to \T$ in $W^{2,p}_\text{loc}$ as $\eps\to 0$.
\item[(iv)] Each $\T_\eps$ is a Cauchy temporal function  for $\check g_\eps$, i.e., 
\beq\label{eq:N_eps-def}  
N_\eps := - \frac{1}{\|{\rm grad}_{\check g_\eps}(\T_\eps)\|_{\check g_\eps}} {\rm grad}_{\check g_\eps} (\T_\eps) \in C^\infty 
\eeq 
is $\check g_\eps$-future-pointing timelike and $\T_\eps \circ \gamma : \R \to \R$ is 
onto for any inextendible $\check{g}_\eps$-causal curve $\gamma$.       
\end{itemize}
\end{Lemma}

Observe that by (iv) each level set $\Sigma^\eps_t:=\T_\eps^{-1}(t)$ ($t\in \R$) is a smooth spacelike Cauchy surface for $(M,\check{g}_\eps)$.

\Proof 
Set $\tilde{\T}_\eps := \T\star_M \rho_\eps(p) +\sqrt{\eps}$. Given $K\Subset M$, there exists $\eps_K>0$ such that $|\tilde \T_\eps - \T|<1$
on $K$ for $\eps\in (0,\eps_K]$. Moreover, $\tilde \T_\eps \to \T$ in $W^{2,p}_\text{loc}$, so defining $\tilde N_\eps$ as in \eqref{eq:N_eps-def} but with
$\tilde\T_\eps$ instead of $\T_\eps$, we have that $\tilde N_\eps \to N$ locally uniformly. Hence we may assume $\eps_K$ small enough to additionally guarantee that $\tilde N_\eps$ is $\check{g}_\eps$-future pointing timelike on $K$ for $\eps\in (0,\eps_K]$. We may now apply \cite[Lem.\ 2.4]{KSSV:14}
to obtain a smooth map $(0,1] \times M \to \R$, $(\eps,p) \mapsto \T_\eps(p)$ with properties (i)-(iii) and such that $N_\eps = \tilde N_\eps$
on any compact subset of $M$ for $\eps$ small.

Finally, if $\gamma:\R \to M$ is an inextendible $\check{g}_\eps$-causal curve, then since $\check{g}_\eps \prec g$, $\gamma$ is also $g$-timelike, so
$\T\circ \gamma:\R \to \R$ is surjective. Due to (ii), the same is then true for 
$\T_\eps \circ \gamma : \R \to \R$. 
\QED

Given any open set $U$ in $\M$, we then define $\mathcal{H}_{\check g_\eps,\T_\eps,U}$ as in \eqref{mean-curv_smear-smooth}.
As above, the restriction of $\mathcal{H}_{\check g_\eps,\T_\eps,U}$ to $\Sigma^\eps_t$ is the classical (smooth) mean curvature of $\Sigma^\eps_t\cap U$. To translate mean curvature bounds from $\mathcal{H}_{g,\T,U}$ to $\mathcal{H}_{\check g_\eps,\T_\eps,U}$, we need the uniform convergence as established in the next proposition. However, observe that $\mathcal{H}_{\check g_\eps,\T_\eps,U}\not\to\mathcal{H}_{g,\T,U}$ in $L^\infty_\text{loc}$.

\begin{Thm}\label{Prop:Mean-curvature-convergence} 
For any open set $U\subseteq \M$, 
\[
\mathcal{H}_{\check g_\eps,\T_\eps,U} -  
\mathcal{H}_{g,\T,U}\star_{\M} \rho_\eps \to 0 \qquad \text{in} \  L^\infty_\text{loc}(U). 
\]
\end{Thm}

\begin{Remark}[Proof strategy]\label{rem:proof-strategy}
Since the proof of this statement is analytically involved we start with some preparatory remarks.
\begin{enumerate}[(i)]
\item Writing out the two expressions in the statement of Theorem \ref{Prop:Mean-curvature-convergence} explicitly we find
\begin{equation}\label{eq:H-hat}
    \mathcal{H}_{\check{g}_{\eps}, \mathcal{T}_{\eps}, U}
    =- \sum_{i,j=1}^{n-1} \check G^{ij}_\eps\ \check{g}_{\eps}(\nabla_{X_{i}^{\eps}}^{\check{g}_{\eps}} X_{j}^{\eps}, N_{\eps})
\end{equation} 
where $\check G^{ij}_\eps=\Big(\big(\check g_\eps (X_k^\eps,X_l^\eps)\big)^{-1}\Big)^{ij}$ and  
\begin{equation}\label{eq:H-sing}
    \mathcal{H}_{g, \mathcal{T}, U} \star_{\mathcal{M}} \rho_{\eps}
    = - \sum_{i,j=1}^{n-1} \Big(G^{ij}\ g(\nabla_{X_i}X_j , N)\Big)\star_{\mathcal{M}} \rho_{\eps}.
\end{equation} 
Here $\mathcal{H}_{g, \mathcal{T}, U}$ is built from the low regularity constituents $g$, $\T$, $N$ and an arbitrary frame $X_i$, according to \eqref{mean-curv_smear}, and the resulting $L^p$-function is smoothed by manifold convolution in the end. On the other hand, $\mathcal{H}_{\check{g}_{\eps}, \mathcal{T}_{\eps}, U}$ is made up from smooth constituents built using $\check g_\eps$ of Lemma \ref{Le:approximating metrics}, $\T_\eps$ of Lemma \ref{Lemma_T-eps_timefct}, $N_\eps$ of \eqref{eq:N_eps-def}, and a frame $X^\eps_i$ which we may choose as the $\check g_\eps$-tangential projection of the smoothing $X_i\star_{\M}\rho_\eps$ of the above frame $X_i$, i.e.
\begin{equation}\label{eq:Xieps-projection}
    X_i^\eps := X_i\star_{\M}\rho_\eps  + \check g_\eps \big(X_i\star_{\M}\rho_\eps, N_\eps \big) N_\eps.
\end{equation}
To estimate the difference of these two expressions we will use as intermediate steps an analogous expression with the $\check g_\eps$-based ingredients replaced by pure convolutions of the rough constituents and a term involving only convolutions of products split according to their regularity.   

\item[(ii)] Recalling the regularities of the building blocks we have $g, \T\in W^{2,p}_\text{loc}$ and $N, X_i, G^{ij}\in W^{1,p}_\text{loc}$, so that the most critical term is $\nabla_{X_i}X_j\in L^p_\text{loc}$. By metric compatibility, the weak product rule and orthogonality $g(X_j, N)=0$, we get
\begin{align}\label{eq:H-Hess}
g(\nabla_{X_i} X_j,\, N) &= - g(X_j,\, \nabla_{X_i} N)\\ \nonumber
&=  \frac{1}{\|\mathrm{grad}_g \T\|_g}\, g\big(X_j, \nabla_{X_i}(\mathrm{grad}_g \T)\big)
= \frac{\mathrm{Hess}^\T(X_i,X_j)}{\|\mathrm{grad}_g \T\|_g},
\end{align}
where we have used \eqref{eq:N-def}.
Next we write the Hessian in a way most convenient for our forthcoming analysis, cf.\ Lemma \ref{lem:commutator_estimates}(ii) below.
Using the fixed background Riemannian metric $h$ and its Levi-Civita connection $\nabla^{(h)}$ (also written as $\nabla_h$), we write the Hessian tensorially:
\begin{equation}\label{eq:Hess-coord}
\begin{split}
\mathrm{Hess}^\T(X_i,X_j) &= X_i^\mu X_j^\nu \nabla^{(g)}_\mu \nabla^{(g)}_\nu \mathcal{T}\\ 
&= X_i^\mu X_j^\nu\bigl(\nabla^{(h)}_\mu \nabla^{(h)}_\nu \T
  - C^\lambda_{\mu\nu}\, d\T_\lambda\bigr),
\end{split}
\end{equation}
where $d\T$ denotes the exterior derivative of $\T$, and   
\[
C_{\mu\nu}^\lambda =  (\Gamma_g)^\lambda_{\mu\nu} - (\Gamma_h)^\lambda_{\mu\nu}
= \frac{1}{2} g^{\lambda\rho}\big(\nabla^{(h)}_\mu g_{\nu\rho} + \nabla^{(h)}_\nu g_{\mu\rho} - \nabla^{(h)}_\rho g_{\mu\nu}\big).
\]

Setting the tensorial input data $\mathbf{U} := (g, \nabla_h g, d\T, (X_i))\in W^{1,p}_{\text{loc}}$ and the singular tensor $\cH_{\mu\nu} := \nabla^{(h)}_\mu \nabla^{(h)}_\nu \T \in L^p_\loc$, we arrive at
\begin{equation}\label{eq:H-fh-r}
\mathcal{H}_{g,\T,U} = \Phi_1(\mathbf{U}) + \Phi_2(\mathbf{U})^{\mu\nu} \cH_{\mu\nu},
\end{equation}
where 
\begin{equation}\label{eq:Phi2-explicit}
  \Phi_2(\mathbf{U})^{\mu\nu}
  = -\frac{1}{\|\mathrm{grad}_g \T\|_g}\,
    \sum_{i,j} G^{ij}\, X_i^\mu\, X_j^\nu
\end{equation}
and
\begin{align}\label{eq:Phi1-explicit}
  \Phi_1(\mathbf{U})
  = -\Phi_2(\mathbf{U})^{\mu\nu}\, C^{\lambda}_{\mu\nu} d\T_\lambda.
\end{align}
This isolates all regular (at least $W^{1,p}_{\text{loc}}$) contributions to either act as a summand ($\Phi_1$) or as a prefactor to the singular term ($\Phi_2$).

\item[(iii)] Following the strategy laid out in (i) and taking into account the singularity structure \eqref{eq:H-fh-r}, we must balance the blow-up of the smoothed singular term $\cH^\eps_{\mu\nu}=\cH_{\mu\nu}\star_\M\rho_\eps$ in the $L^\infty$-norm, which by Young's inequality is of order $\mathcal{O}(\eps^{\alpha -1})=\mathcal{O}(\eps^{-n/p})$, by a corresponding fall-off of the regular terms. We must additionally track the commutators generated when swapping covariant/exterior derivatives and $\star_{\mathcal{M}}$. We rely on four auxiliary results: commutator estimates on first derivatives of tensors and on the Hessian of scalar functions (Lemma \ref{lem:commutator_estimates}), a sharp convergence bound for $\check g_\eps-g_\eps$ (Lemma \ref{lem:m+}), a nonlinear commutator estimate on tensor bundles (Lemma \ref{lem:nonlinear-commutator-estimate}), and a tensorial Friedrichs-type lemma for $W^{1,p}\times L^p$ products (Lemma \ref{lem:roland_tensor}). 
\end{enumerate}
\end{Remark}

\begin{proof}[Proof of Theorem \ref{Prop:Mean-curvature-convergence}]
Let $\Omega \subseteq U$ be a relatively compact open set permitting a local $W^{1,p}_\text{loc}$-frame $X_i$ ($i=1,\dots,n-1$) adapted to $N$ as in \eqref{eq:W1p-frame}. Now set 
\beq \label{eq:Xieps}
X_i^\eps := X_i\star_{\M}\rho_\eps  + \check g_\eps \big(X_i\star_{\M}\rho_\eps, N_\eps \big) N_\eps.
\eeq 
In this definition, $X_i$ is extended by $0$ outside of $\Omega$ and we have that $X_i^\eps \to X_i$ in $W^{1,p}_\text{loc}(\Omega)$ as $\eps \to 0$. By construction, $X_i^\eps \perp_{\check g_\eps} N_\eps$. Consequently, $X_i^\eps$ ($i=1,\dots,n-1$) qualifies as a frame adapted to $N_\eps$, hence can be used for calculating $\mathcal{H}_{\check g_\eps,\T_\eps,U}$ on $\Omega$. 

To decompose $\mathcal{H}_{g, \mathcal{T}, U}$ into its regular $W^{1,p}_{\text{loc}}$ and its singular $L^p_{\text{loc}}$ parts as announced in Remark \ref{rem:proof-strategy}, we introduce smooth maps taking either rough input data $\mathbf{U}=(g,\, \nabla_h g,\, d\T,\, (X_i))$ or input data built from $\check g_\eps$-regularised constituents $\check\U_\eps=(\check g_\eps, \nabla_h \check g_\eps, d\T_\eps, (X_i\star_{\M}\rho_\eps))$. 
(It will become clear below that the fourth component really is the pure convolution term $X_i\star_{\M}\rho_\eps$. In fact, the corresponding normal projection will be taken care of by the maps yet to be defined.)
We also use an intermediate term built entirely from pure convolutions: $\U_\eps=(g_\eps, (\nabla_h g)\star_{\M}\rho_\eps, (d\T)\star_{\M}\rho_\eps, (X_i\star_{\M}\rho_\eps))$, where $g_\eps = g\star_{\M}\rho_\eps$.

Now, for `generic' input data $\U'=(g', \nabla_h g', d\T', (V'_i))$, define the smooth maps:
\begin{align}
  \Phi_\sigma(g', d\T')
  &:= \bigl(-g'^{-1}(d\T', d\T')\bigr)^{-1/2}\,,
  \label{eq:Phisigma}\\ 
  \Phi_N(g', d\T')^\mu
  &:= -  \Phi_\sigma(g', d\T')\, g'^{\mu\nu}(d\T')_\nu\,,
  \label{eq:PhiN}\\
  \Phi_X(g', V'_i, d\T')
  &:= V'_i + g'\bigl(V'_i,\, \Phi_N(g', d\T')\bigr)\,\Phi_N(g', d\T')\,,
  \label{eq:PhiX}\\
  \Phi_G(g', V'_1, \ldots, V'_{n-1})^{ij}
  &:= \Big(\bigl[\big(g'(V'_k, V'_l)\big)_{k,l=1}^{n-1}\bigr]^{-1}\Big)^{ij}\,,
  \label{eq:PhiG}\\
  \Phi_C(g', \nabla_h g')^\lambda_{\mu\nu}
  &:= \tfrac{1}{2}\, g'^{\lambda\rho}
     \bigl((\nabla_h)_\mu g'_{\nu\rho} + (\nabla_h)_\nu g'_{\mu\rho}
      - (\nabla_h)_\rho g'_{\mu\nu}\bigr)\,.
  \label{eq:PhiGamma}
\end{align}
These are smooth maps on their natural domains, which is 
$$
\{(g',\partial \T') : {(g')}^{-1}(\partial \T',\partial \T') < 0\}
$$ 
for $\Phi_\sigma$ and the maps derived from it.  To obtain `universal' formulas on all three input levels, define the adapted frames and inverse Gram matrices:
\begin{equation}
 \tilde{X}_i := \Phi_X(g', V_i', d\T'),
 \qquad
 \tilde{G}^{ij} := \Phi_G(g', \tilde{X}_1, \ldots, \tilde{X}_{n-1})^{ij}.  
\end{equation}
We now split the mean curvature formally as $H = \Phi_1 + \Phi_2^{\mu\nu} \cH'_{\mu\nu}$ with $\cH'_{\mu\nu} := \nabla^{(h)}_\mu \nabla^{(h)}_\nu \T'$ via:
\begin{align}\label{eq:Phi2-explicit-general} \nonumber
\Phi_2&(g', d\T', V_1',\dots,V_{n-1}')^{\mu\nu}\\  \nonumber
&:= -\Phi_\sigma(g', d\T')\,
  \sum_{i,j} \tilde{G}^{ij}\, \tilde{X}_i^\mu\, \tilde{X}_j^\nu \\\nonumber
&= -\Phi_\sigma(g', d\T')\,
  \sum_{i,j}  \Phi_G\Big(g', \Phi_X(g', V_1', d\T'), \ldots, \Phi_X(g', V_{n-1}', d\T')\Big)^{ij}\\
  &\hphantom{-\Phi_\sigma(g', d\T')\,\sum_{i,j}  \Phi_G^{ij}} \cdot
  \Phi_X(g', V_{i}', d\T')^\mu\,\Phi_X(g', V_{j}', d\T')^\nu ,
\\ \nonumber \\\nonumber
\label{eq:Phi1-explicit-general} 
\Phi_1&(g', \nabla_h g', d\T', V_1',\dots,V_{n-1}')\\
&:= -\Phi_2(g', d\T', V_1',\dots,V_{n-1}')^{\mu\nu}\,
  \Phi_C(g', \nabla_h g')^\lambda_{\mu\nu}\, (d\T')_\lambda.
\end{align}

On the rough level, for $\mathbf{U} = (g, \nabla_h g, d\T, (X_i))$ we clearly have
\begin{equation}\label{eq:H-fh}
\mathcal{H}_{g,\T,U} = \Phi_1(\mathbf{U}) + \Phi_2(\mathbf{U})^{\mu\nu}\, \cH_{\mu\nu}.
\end{equation}
Note that $\Phi_X(g, X_i, d\T) = X_i$, so $\tilde{X}_i = X_i$ and $\tilde{G}^{ij} = G^{ij}$. Regarding regularities, $\Phi_1(\U), \Phi_2(\U)\in \Wo\subseteq \Calph$ and $\cH_{\mu\nu} \in L^p_{\text{loc}}$.

On the level of the $\check g_\eps$ regularisation, for data $\check\U_\eps=(\check g_\eps, \nabla_h \check g_\eps, d\T_\eps, (X_i\star_{\M}\rho_\eps))$, the function $\Phi_X$ naturally applies the required tangent projection to $X_i \star_\M \rho_\eps$ to produce $X_i^\eps$ directly. Thus, 
\begin{equation}
    \mathcal{H}_{\check{g}_\eps, \T_\eps, U} = \Phi_1(\check{\mathbf{U}}_\eps) + \Phi_2(\check{\mathbf{U}}_\eps)^{\mu\nu} \nabla^{(h)}_\mu \nabla^{(h)}_\nu \T_\eps.
\end{equation} 

Now let $K\Subset\M$. By Lemma \ref{lem:commutator_estimates}(ii), $\|[\nabla_h^2, \star_\M\rho_\eps]\T\|_{L^\infty(K)} = \mathcal{O}(\eps^\alpha)$, so by Lemma \ref{Lemma_T-eps_timefct}(i) on $K$ we may write $\nabla_h^2 \T_\eps = \nabla_h^2 (\T\star_{\M}\rho_\eps + \sqrt{\eps}) = \cH_\eps + \mathcal{O}(\eps^\alpha)$, where $(\cH_\eps)_{\mu\nu} = (\nabla_h^2 \T)_{\mu\nu}\star_{\M}\rho_\eps$. Thus we have
\begin{equation}\label{eq:split-Lp}
    \mathcal{H}_{\check{g}_\eps, \T_\eps, U}
  = \Phi_1(\check{\mathbf{U}}_\eps)
  + \Phi_2(\check{\mathbf{U}}_\eps)\cdot \cH_\eps + \mathcal{O}(\eps^\alpha).
\end{equation} 

Now to start our estimates we use the following splitting 
\begin{align}\label{eq:decomp-new}\nonumber
\mathcal{H}_{\check{g}_\eps, \T_\eps, U} &- \mathcal{H}_{g,\T,U} \star_\M \rho_\eps \\ 
&= \underbrace{
 \Big(\Phi_1(\check{\mathbf{U}}_\eps) + \Phi_2(\check{\mathbf{U}}_\eps)\cdot \cH_\eps\Big)
 - \Big(\Phi_1({\mathbf{U}}_\eps) + \Phi_2({\mathbf{U}}_\eps)\cdot \cH_\eps\Big)
 }_{\text{term~(I)}} \ \ + \ \ \mathcal{O}(\eps^\alpha)\\ \nonumber
&\quad+\underbrace{
 \Big(\Phi_1({\mathbf{U}}_\eps) + \Phi_2({\mathbf{U}}_\eps)\cdot \cH_\eps\Big)
   - \Big(\Phi_1({\mathbf{U}})\star_{\M}\rho_\eps + \Phi_2({\mathbf{U}})\star_{\M}\rho_\eps\cdot \cH_\eps\Big)
 }_{\text{term~(II)}}\\ \nonumber
&\quad+\underbrace{
 \Phi_2({\mathbf{U}})\star_{\M}\rho_\eps\cdot \cH_\eps
 - \big(\Phi_2(\mathbf{U})\cdot \cH\big)\star_{\M}\rho_\eps.
  }_{\text{term~(III)}}
\end{align} 
In the following we treat terms (I) -- (III) separately.

\medskip
\noindent\emph{Term (I): $\mathcal{O}(\eps^\alpha)$ estimates for exchanging $\check{\mathbf{U}}_\eps$ by $\mathbf{U}_\eps$.}\\
All entries of both $\check{\mathbf{U}}_\eps = (\check g_\eps, \nabla_h \check g_\eps, d\T_\eps, (X_i\star_{\M}\rho_\eps))$ and $\mathbf{U}_\eps = (g_\eps, (\nabla_h g)\star_{\M}\rho_\eps, (d\T)\star_{\M}\rho_\eps, (X_i\star_{\M}\rho_\eps))$ are uniformly bounded on compact sets. By Lemma \ref{Lemma_T-eps_timefct}(i), $\T_\eps = \T\star_\M \rho_\eps + \sqrt{\eps}$. By Lemma \ref{lem:commutator_estimates} (i), $d(\T \star_\M \rho_\eps) = (d\T) \star_\M \rho_\eps + \mathcal{O}(\eps)$. Thus, $d\T_\eps = (d\T) \star_\M \rho_\eps + \mathcal{O}(\eps)$. By Lemma \ref{lem:m+}, we have $\nabla_h \check{g}_\eps = \nabla_h g_\eps + \mathcal{O}(\eps)$ and by Lemma \ref{lem:commutator_estimates} (i), (for $p=\infty$ and $\alpha=1$), we find further that $\nabla_h g_\eps = (\nabla_h g) \star_\M \rho_\eps + \mathcal{O}(\eps)$.  Since, moreover, $\check{g}_\eps - g_\eps = \mathcal{O}(\eps)$ and the frame components $(X_i \star_\M \rho_\eps)$ match identically in both tuples, we obtain $\|\check{\mathbf{U}}_\eps - \mathbf{U}_\eps\|_{L^\infty(K)} = \mathcal{O}(\eps)$.        Traversing through the smooth compositions $\Phi_k$ then yields      
\[
\|\Phi_k(\check{\mathbf{U}}_\eps) - \Phi_k(\mathbf{U}_\eps)\|_{L^\infty(K)} = \mathcal{O}(\eps), \qquad k = 1, 2.
\]
Multiplied by the singular factor $\cH_\eps$, where $\|\cH_\eps\|_{L^\infty(K)}=\mathcal{O}(\eps^{\alpha-1})$ (see Remark \ref{rem:proof-strategy} (iii)), we have
\begin{equation}
    \|\text{term (I)}\|_{L^\infty(K)} = \mathcal{O}(\eps) + \mathcal{O}(\eps) \cdot \mathcal{O}(\eps^{\alpha-1}) = \mathcal{O}(\eps^{\alpha})\to 0.
\end{equation}

\medskip
\noindent\emph{Term~(II): $\mathcal{O}(\eps^{\min(3\alpha-1,\,\alpha)})$ estimates by nonlinear commutator relations.}\\ 
Each $\Phi_k$ ($k=1,2$) is a smooth ($\eps$-independent) map defined on a tensor bundle and the inputs $\mathbf{U}$ are sections of tensor bundles of regularity $C^{0,\alpha}_\loc$. 
So by Lemma \ref{lem:nonlinear-commutator-estimate} we obtain
\begin{equation}
\|\Phi_k(\mathbf{U}_\eps) - (\Phi_k(\mathbf{U}))\star_\M\rho_\eps\|_{L^\infty(K)}
= \mathcal{O}(\eps^{2\alpha}),
\qquad k = 1, 2.
\end{equation}
The $\Phi_1$-part is $\mathcal{O}(\eps^{2\alpha})$, and the $\Phi_2$-part is multiplied by $\cH_\eps$, where again $\|\cH_\eps\|_{L^\infty(K)} = \mathcal{O}(\eps^{\alpha - 1})$, to give   
\[
\mathcal{O}(\eps^{2\alpha}) \cdot \mathcal{O}(\eps^{\alpha - 1}) = \mathcal{O}(\eps^{3\alpha - 1}) \to 0
\]
provided $\alpha > 1/3$, i.e., $p > \tfrac{3n}{2}$. So for such a choice of $p$, and since we assume $p>2n$, we have in total 
\begin{equation}
    \|\text{term (II)}\|_{L^\infty(K)}\to 0.
\end{equation}

\medskip
\noindent\emph{Term~(III): $\mathcal{O}(\eps^{1-2n/p})$ estimates by a Friedrichs-type lemma.}\\
Setting $f := \Phi_2(\mathbf{U}) \in W^{1,p}_\loc$ and $f_\eps := f \star_\M \rho_\eps$, term~(III) reads $f_\eps \cdot \cH_\eps - (f\cH)\star_\M\rho_\eps$.
Due to Lemma \ref{lem:roland_tensor}, we have
\begin{equation}\label{eq:lemma57}
 \|\text{term (III)}\|_{L^\infty(K)}=\|f_\eps \, \cH_\eps - (f\cH)_\eps\|_{L^\infty(K)} = \mathcal{O}(\eps^{2\alpha-1}) \to 0
\end{equation}
provided $\alpha>1/2$, i.e., $p > 2n$, as assumed.
\medskip

In view of the splitting \eqref{eq:decomp-new} and the $\mathcal{O}(\eps^\alpha)$ remainder, we are done.
\end{proof}

\section{Refined results on manifold convolution} \label{Sec_refined_conv_lemmas}

In this section we formulate and prove the three missing auxiliary results used in the proof of Theorem \ref{Prop:Mean-curvature-convergence}. These results, establishing basic properties of the manifold convolution \eqref{eq:M-convolution}, are of interest in their own right. We begin with the nonlinear commutator estimates that allowed us to permute the manifold convolution with the $\Phi_k$ in term (II).

\begin{Lemma}[Nonlinear commutator estimate]\label{lem:nonlinear-commutator-estimate} 
Let $E$, $F$ be tensor bundles over $\mathcal{M}$,
and let $\Phi: E \to F$ be a smooth (indeed, $C^2$ is sufficient), fiber-preserving nonlinear map. 
Then for any section $T \in C^{0,\alpha}_\text{loc}(M, E)$ with
$\alpha\in (0,1)$ we have on any compact set $K\Subset \mathcal{M}$
\begin{equation}\label{eq:nonlinear-commutator-estimate}
\left\| \Phi(T\star_\mathcal{M} \rho_\eps) - \Phi(T)\star_\mathcal{M}\rho_\eps \right\|_{L^\infty(K)} = \mathcal{O}\left(\eps^{2\alpha}\right).
\end{equation}
\end{Lemma}
\begin{proof} To simplify notations, in this proof we are going to denote application of $\star_\mathcal{M}\rho_\eps$ simply by a 
subscript $\eps$, i.e.\ $T_\eps = T \star_\mathcal{M}\rho_\eps$, etc. We first derive explicit local representations.

Denote by $T_{(\alpha)}: U_\alpha \to \mathbb{R}^N$ the local representation in the trivialization associated with the chart $(U_\alpha,\psi_\alpha)$\footnote{We use the common indices $\alpha$ and $\beta$ for charts. There should be no danger of confusing this with the Hölder exponent.}, i.e.,
\begin{equation}
    T_{(\alpha)} = \big( (\psi_\alpha)_* T \big) \circ \psi_\alpha.
\end{equation}
Also, let $\Phi_{(\alpha)}$ be the local coordinate representation of $\Phi$, defined by $\Phi_{(\alpha)} := (\psi_\alpha)_* \circ \Phi \circ (\psi_\alpha)^*$. We denote the transition matrix mapping tensor components from chart $\alpha$ to chart $\beta$ at a point $q \in U_\alpha \cap U_\beta$ by $J_{\alpha \to \beta}^E(q) \in \text{GL}(N, \mathbb{R})$. 
Then for $x\in \R^n$ we have 
\begin{align}
\left((\psi_\alpha)_* (\xi_\alpha T)\right)*\rho_\eps(x)
&= \int_{\mathbb{R}^n} \xi_\alpha(\psi_\alpha^{-1}(y))\, T_{(\alpha)}(\psi_\alpha^{-1}(y))\, \rho_\eps(x - y) \, dy.
\end{align}
Setting $S_\alpha:=(\psi_\alpha)^*  \left[ \left( (\psi_\alpha)_* (\xi_\alpha T) \right)  * \rho_\eps \right]$, at $p=\psi_\alpha^{-1}(x)$ we obtain
\begin{align}
(S_\alpha)_{(\alpha)} (p) = \int_{\mathbb{R}^n} \xi_\alpha(\psi_\alpha^{-1}(y))\, T_{(\alpha)}(\psi_\alpha^{-1}(y))\, \rho_\eps(\psi_\alpha(p) - y) \, dy.
\end{align}
Next we express $T_\eps(p) = \sum_\alpha \chi_\alpha(p) S_\alpha(p)$ in some fixed chart $\psi_\beta$. We apply the transition matrix $J_{\alpha \to \beta}^E(p)$ to map the components from frame $\alpha$ to frame $\beta$ at $p$:
\begin{equation}
(S_\alpha)_{(\beta)}(p) = J_{\alpha \to \beta}^E(p) \,(S_\alpha)_{(\alpha)}(p),
\end{equation}
so that 
\begin{equation}\label{eq:mid_step}
(S_\alpha)_{(\beta)}(p) = \int_{\mathbb{R}^n} \xi_\alpha(\psi_\alpha^{-1}(y))\, J_{\alpha \to \beta}^E(p)\, T_{(\alpha)}(\psi_\alpha^{-1}(y))\, \rho_\eps(\psi_\alpha(p) - y) \, dy.
\end{equation}
Letting $q = \psi_\alpha^{-1}(y)$ and recalling that $T_{(\alpha)}(q) = J_{\beta \to \alpha}^E(q) T_{(\beta)}(q)$, we arrive at the 
local expression for $T_\eps$ in the frame associated with the chart $\psi_\beta$:
\begin{equation}\label{eq:moll_def}
\begin{split}
&(T_\eps)_{(\beta)}(p)\\ 
&= \sum_\alpha \chi_\alpha(p) \int_{\mathbb{R}^n} \xi_\alpha(q) \underbrace{\left[ J_{\alpha \to \beta}^E(p) J_{\beta \to \alpha}^E(q) \right]}_{=: P_\alpha^E(p,q)} T_{(\beta)}(q)\, \rho_\eps(\psi_\alpha(p) - y) \, dy.
\end{split}
\end{equation}
Here, the matrix $P_\alpha^E(p,q) \in \text{GL}(N, \mathbb{R})$ acts as a coordinate parallel transport operator. It maps a tensor from frame $\beta$ to frame $\alpha$ at the source point $q$, and then maps it from frame $\alpha$ back to frame $\beta$ at the evaluation point $p$. If $p=q$, the transition matrices are exact inverses, yielding $P_\alpha^E(p,p) = I_E$. By the smoothness of the transition maps, 
\begin{equation}\label{eq:transport-O-epsilon}
P_\alpha^E(p,q) = I_E + \mathcal{O}(d^h(p,q)) = I_E + \mathcal{O}(\eps)    
\end{equation}
uniformly on the compact support of the mollifier $\rho$.
\medskip

Fix $p \in K$ and a reference chart $\psi_\beta$, and let $\overline{T} := (T_\eps)_{(\beta)}(p)$. Applying \eqref{eq:moll_def} to the tensor field $\Phi(T)$, its mollification is given by:
\begin{equation}\label{eq:phi-T-moll}
\begin{split}
(\Phi(T)_\eps&)_{(\beta)}(p)\\ 
&= \sum_\alpha \chi_\alpha(p) \int_{\mathbb{R}^n} \xi_\alpha(q)\, P_\alpha^F(p,q)\, \Phi_{(\beta)}(T_{(\beta)}(q))\, \rho_\eps(\psi_\alpha(p) - y) \, dy,
\end{split}
\end{equation}
where $P_\alpha^F(p,q) = J_{\alpha \to \beta}^F(p) J_{\beta \to \alpha}^F(q)$ is the transport operator for the target bundle $F$. 
We Taylor expand $\Phi_{(\beta)}$ around $\overline{T}$:
\begin{equation}\label{eq:taylor-T}
\Phi_{(\beta)}(T_{(\beta)}(q)) = \Phi_{(\beta)}(\overline{T}) + D\Phi_{(\beta)}(\overline{T})\big[T_{(\beta)}(q) - \overline{T}\big] + R(q).
\end{equation}
Let $q = \psi_\alpha^{-1}(y)$ be such that $y$ is in the support of $\rho_\eps(\psi_\alpha(p) - .)$. Then we claim that
\begin{equation}\label{eq:remainder-bound}
|R(q)| \le C\eps^{2\alpha},
\end{equation}
for some constant $C > 0$ independent of $\eps$.

To see this, first note that by Taylor's theorem and the local boundedness of $D^2\Phi_{(\beta)}$ we have
$|R(q)| \le C \left| T_{(\beta)}(q) - \overline{T} \right|^2.$ Here,
\begin{equation} \label{eq:triangle_T}
\left| T_{(\beta)}(q) - \overline{T} \right| \le \left| T_{(\beta)}(q) - T_{(\beta)}(p) \right| + \left| T_{(\beta)}(p) - \overline{T} \right|.
\end{equation}
In the first term, $y=\psi_\alpha(q)$ ranges over the support of $\rho_\eps(\psi_\alpha(p) - .)$, so $|y - \psi_\alpha(p)| \le \eps$. Since chart transition maps are bi-Lipschitz on compact sets, we have $d^h(p, q) \le C_1 \eps$. The assumption that $T \in C^{0,\alpha}_\text{loc}(M, E)$ therefore gives
\begin{equation}\label{eq:first-term-est}
\left| T_{(\beta)}(q) - T_{(\beta)}(p) \right| \le C \eps^\alpha.
\end{equation}
To estimate the second term in \eqref{eq:triangle_T}, using \eqref{eq:moll_def} and setting
\begin{equation}
M_E(p) := \sum_\alpha \chi_\alpha(p) \int_{\mathbb{R}^n} \xi_\alpha(w) P_\alpha^E(p,w) \rho_\eps(\psi_\alpha(p) - z) \, dz,
\end{equation}
we can write:
\begin{equation}
\begin{split}
\overline{T} - &T_{(\beta)}(p)\\ 
&= \sum_\alpha \chi_\alpha(p) \int_{\mathbb{R}^n} \xi_\alpha(w) P_\alpha^E(p,w) \big( T_{(\beta)}(w) - T_{(\beta)}(p) \big) \rho_\eps(\psi_\alpha(p) - z) \, dz\\
&+ \big( M_E(p) - I_E \big) T_{(\beta)}(p),
\end{split}
\end{equation}
where $w = \psi_\alpha^{-1}(z)$ and $I_E$ is the $N\times N$ identity matrix. On the support of $\rho$ we again have $d^h(p, w) \le C \eps$, so $\left| T_{(\beta)}(w) - T_{(\beta)}(p) \right| \le C \eps^\alpha$. The parallel transport operator by \eqref{eq:transport-O-epsilon} is uniformly bounded with $P_\alpha^E(p, w) = I_E + \mathcal{O}(\eps)$. Thus the first integral term is bounded by $\mathcal{O}(\eps^\alpha)$.

Furthermore,  
\begin{equation}\label{eq:M_E-estimate}
    M_E(p) = I_E + \mathcal{O}(\eps^2).
\end{equation}
To see this, change variables to $u = \psi_\alpha(p) - z$. Then the integral becomes $\int_{\mathbb{R}^n} \xi_\alpha(w) P_\alpha^E(p,w) \rho_\eps(u) \, du$ where $w = \psi_\alpha^{-1}(\psi_\alpha(p) - u)$. Expanding the integrand $\xi_\alpha(w) P_\alpha^E(p,w)$ around $u=0$ (i.e., $w=p$), the constant term takes the form $\xi_\alpha(p)$$P_\alpha^E(p,p) = \xi_\alpha(p) I_E$. Because $\rho_\eps(u)$ is spherically symmetric, the linear term in $u$ integrates to $0$. Summing over $\alpha$ and using $\chi_\alpha \equiv 1$ on $\text{supp}(\xi_\alpha)$ along with $\sum_\alpha \xi_\alpha = 1$, we obtain $M_E(p) = \sum_\alpha \chi_\alpha(p) \xi_\alpha(p) I_E + \mathcal{O}(\eps^2) = I_E + \mathcal{O}(\eps^2)$. Consequently,
\begin{equation}
\left| \overline{T} - T_{(\beta)}(p) \right| \le \mathcal{O}(\eps^\alpha) + \mathcal{O}(\eps^2) \le C \eps^\alpha,
\end{equation}
since $\alpha \in (0,1)$. Together with \eqref{eq:first-term-est} this establishes \eqref{eq:remainder-bound}.
\medskip

To show \eqref{eq:nonlinear-commutator-estimate} we need to estimate $|(\Phi(T)_\eps)_{(\beta)}(p) - \Phi_{(\beta)}(\overline{T})|$ uniformly for $p\in K$. To this end, we insert the expansion \eqref{eq:taylor-T} into \eqref{eq:phi-T-moll}. This induces a splitting of 
$(\Phi(T)_\eps)_{(\beta)}$ into three terms, corresponding to the constant, linear, and remainder term in \eqref{eq:taylor-T}
\begin{equation}\label{eq:Phi-T-eps-splitting}
    (\Phi(T)_\eps)_{(\beta)}(p) = I_C + I_L + I_R,
\end{equation}
and we are going to estimate each of these terms in turn. 

\noindent \underline{$I_C$ and $I_R$:} 
Let $M_F(p) = \sum_\alpha \chi_\alpha(p) \int_{\mathbb{R}^n} \xi_\alpha(q) P_\alpha^F(p,q) \rho_\eps(\psi_\alpha(p) - y) \, dy$. 
As in the case of $M_E$ above, since $\rho$ is symmetric and $P_\alpha^F(p,p) = I_F$, the linear terms in the Taylor expansion of the integrand vanish, so $M_F(p) = I_F + \mathcal{O}(\eps^2)$. Thus 
\begin{equation}
I_C = M_F(p) \Phi_{(\beta)}(\overline{T}) = \Phi_{(\beta)}(\overline{T}) + \mathcal{O}(\eps^2).
\end{equation}
Concerning $I_R$, since $P_\alpha^F$ is bounded, we obtain $I_R = \mathcal{O}(\eps^{2\alpha})$ from \eqref{eq:remainder-bound}.

\noindent \underline{The linear term $I_L$:} We have
\begin{equation}
I_L = \sum_\alpha \chi_\alpha(p) \int_{\mathbb{R}^n} \xi_\alpha(q) P_\alpha^F(p,q) D\Phi_{(\beta)}(\overline{T}) \big[T_{(\beta)}(q) - \overline{T}\big] \rho_\eps(\psi_\alpha(p) - y) \, dy.
\end{equation}
Since $P_\alpha^F(p,q) D\Phi_{(\beta)}(\overline{T}) \neq D\Phi_{(\beta)}(\overline{T}) P_\alpha^E(p,q)$ in general, we quantify the resulting
mismatch by the following commutator
\begin{equation}
\mathcal{E}_\alpha(p,q) := P_\alpha^F(p,q) D\Phi_{(\beta)}(\overline{T}) - D\Phi_{(\beta)}(\overline{T}) P_\alpha^E(p,q).
\end{equation}
At $q=p$, both $P_\alpha^F$ and $P_\alpha^E$ evaluate to the identity matrices on their respective bundles, so $\mathcal{E}_\alpha(p,p) = 0$. Smoothness of the transition maps therefore implies $|\mathcal{E}_\alpha(p,q)| = \mathcal{O}(d^h(p,q)) = \mathcal{O}(\eps)$. Next, we split $I_L = I_{L1} + I_{L2}$ using this commutator
\begin{align}\nonumber
I_{L1} &= D\Phi_{(\beta)}(\overline{T}) \left[ \sum_\alpha \chi_\alpha(p) \int_{\mathbb{R}^n} \xi_\alpha(q) P_\alpha^E(p,q) \big[T_{(\beta)}(q) - \overline{T}\big] \rho_\eps(\psi_\alpha(p) - y) \, dy \right], \\
I_{L2} &= \sum_\alpha \chi_\alpha(p) \int_{\mathbb{R}^n} \xi_\alpha(q) \mathcal{E}_\alpha(p,q) \big[T_{(\beta)}(q) - \overline{T}\big] \rho_\eps(\psi_\alpha(p) - y) \, dy.
\end{align}

To estimate $I_{L1}$, note that by \eqref{eq:moll_def}, integrating the transport operator against $T_{(\beta)}(q)$ yields exactly $\overline{T}$. Integrating against the constant $\overline{T}$ yields $M_E(p)\overline{T}$, where $M_E(p) = I_E + \mathcal{O}(\eps^2)$ by \eqref{eq:M_E-estimate}. Thus, the bracket evaluates to $\overline{T} - (I_E + \mathcal{O}(\eps^2))\overline{T} = \mathcal{O}(\eps^2)$, meaning $I_{L1} = \mathcal{O}(\eps^2)$.

In the term $I_{L2}$ the integrand is the product of the commutator $\mathcal{E}_\alpha(p,q) = \mathcal{O}(\eps)$ and the H\"older deviation $T_{(\beta)}(q) - \overline{T} = \mathcal{O}(\eps^\alpha)$. Therefore, the integral is pointwise bounded by their product
\begin{equation}
I_{L2} = \mathcal{O}(\eps) \cdot \mathcal{O}(\eps^\alpha) = \mathcal{O}(\eps^{1+\alpha}).
\end{equation}
Altogether, we arrive at
\begin{align}\nonumber
| (\Phi(T)_\eps)_{(\beta)}(p) - \Phi_{(\beta)}(\overline{T}) | &\le |I_C - \Phi_{(\beta)}(\overline{T})| + |I_{L1}| + |I_{L2}| + |I_R| \nonumber \\
&= \mathcal{O}(\eps^2) + \mathcal{O}(\eps^2) + \mathcal{O}(\eps^{1+\alpha}) + \mathcal{O}(\eps^{2\alpha}),
\end{align}
uniformly for $p\in K$. As $\alpha \in (0,1)$, this concludes the proof of \eqref{eq:nonlinear-commutator-estimate}.
\end{proof}

The next result is the Friedrichs-type lemma which allowed us to estimate term (III) in the proof of Theorem \ref{Prop:Mean-curvature-convergence}. Observe that it contains the most critical exponents, which enforces the choice of $p>2n$.

\begin{Lemma}[Tensorial Friedrichs-type estimate]\label{lem:roland_tensor}
Let $E$ and $F$ be tensor bundles over $\mathcal{M}$, and let 
$\mathcal{C} \colon E \times F \to \R$ be a smooth, pointwise bilinear pairing (e.g.\  
the tensor contraction). Suppose that $f$, $\cH$ are sections of regularity $f \in W^{1,p}_\text{loc}(\mathcal{M}, E)$ and $\cH \in L^p_\text{loc}(\mathcal{M}, F)$ for some $p > 2n$. Setting $f_\eps := f \star_\M \rho_\eps$, $\cH_\eps := \cH \star_\M \rho_\eps$, and $(\mathcal{C}(f, \cH))_\eps := \mathcal{C}(f, \cH) \star_\M \rho_\eps$, we have for any compact set $K \Subset \mathcal{M}$ that
\begin{equation}
    \left\| \mathcal{C}(f_\eps, \cH_\eps) - (\mathcal{C}(f, \cH))_\eps \right\|_{L^\infty(K)} = \mathcal{O}(\eps^{1 - 2n/p}) \to 0 \quad \text{as } \eps \to 0.
\end{equation}
\end{Lemma}
\begin{proof}
By Morrey's inequality, $p > 2n$ implies that $\alpha := 1 - n/p > 1/2$ and $f \in C^{0,\alpha}_\text{loc}(\mathcal{M}, E)$. We measure the pointwise norms of tensors, denoted by $|\cdot|$, with respect to the  fixed background Riemannian metric $h$. Let $K \Subset \mathcal{M}$ be compact. We evaluate the convolutions at an arbitrary point $p \in K$.\footnote{We keep to the convention of denoting points in $\M$ by $p,q$, as there should be no confusion with the Sobolev index $p$.}

As in the proof of Lemma \ref{lem:nonlinear-commutator-estimate}, we fix a reference chart $(U_\beta, \psi_\beta)$ around $p$ and denote the local components of tensor fields in this chart by a subscript $(\beta)$. 

By \eqref{eq:moll_def}, for any tensor field $T$ in, e.g., $E$ we have
\begin{equation}\label{eq:tensor_moll_def}
(T_\eps)_{(\beta)}(p) = \sum_\alpha \chi_\alpha(p) \int_{\R^n} \xi_\alpha(q)\, P_\alpha^E(p,q)\ T_{(\beta)}(q)\ \rho_\eps(\psi_\alpha(p) - y) \, dy,
\end{equation}
where $y \in \R^n$, $q = \psi_\alpha^{-1}(y)$, and $P_\alpha^E(p,q) = J_{\alpha \to \beta}^E(p) J_{\beta \to \alpha}^E(q)$ is the coordinate parallel transport matrix from $q$ to $p$ for the bundle of $E$. For a scalar function $u$, the transport is trivial: $P_\alpha^\R(p,q) \equiv 1$.

Because the pairing $\mathcal{C}$ is a coordinate-invariant tensor contraction, the local components satisfy
\begin{equation}\label{eq:pairing_invariance}
\mathcal{C}\big(P_\alpha^E(p,q) v, P_\alpha^F(p,q) w\big) = \mathcal{C}(v, w)
\end{equation}
for any $v \in E_q$ and $w \in F_q$ (the transition matrices $J_{\beta \to \alpha}^E(q)$ and $J_{\beta \to \alpha}^F(q)$ cancel against $J_{\alpha \to \beta}^E(p)$ and $J_{\alpha \to \beta}^F(p)$ within the contraction).

Applying \eqref{eq:tensor_moll_def} to the scalar function $\mathcal{C}(f, \cH)$, we have 
\begin{equation}\label{eq:fh_eps_def}
(\mathcal{C}(f, \cH))_\eps(p) = \sum_\alpha \chi_\alpha(p) \int_{\R^n} \xi_\alpha(q) \mathcal{C}\big(f_{(\beta)}(q), \cH_{(\beta)}(q)\big) \rho_\eps(\psi_\alpha(p) - y) \, dy.
\end{equation}
On the other hand, \eqref{eq:tensor_moll_def} gives
\begin{equation}\label{eq:f_eps_h_eps}
\begin{split}
&\mathcal{C}(f_\eps(p), \cH_\eps(p)) = \\
&\sum_\alpha \chi_\alpha(p) \int_{\R^n} \xi_\alpha(q)\ \mathcal{C}\Big((f_\eps)_{(\beta)}(p), P_\alpha^F(p,q)\, \cH_{(\beta)}(q)\Big)\, \rho_\eps(\psi_\alpha(p) - y) \, dy.
\end{split}
\end{equation}
By \eqref{eq:pairing_invariance}, $\mathcal{C}\big(f_{(\beta)}(q), \cH_{(\beta)}(q)\big) = \mathcal{C}\big(P_\alpha^E(p,q) f_{(\beta)}(q), P_\alpha^F(p,q) \cH_{(\beta)}(q)\big)$. Thus we can write the difference $\Delta_\eps(p) := \mathcal{C}(f_\eps(p), \cH_\eps(p)) - (\mathcal{C}(f, \cH))_\eps(p)$ as
\begin{equation}\label{eq:difference}
\begin{split}
& \Delta_\eps(p) = 
\sum_\alpha \chi_\alpha(p) \\ 
& \cdot\int_{\R^n} \xi_\alpha(q)\ \mathcal{C}\Big((f_\eps)_{(\beta)}(p) - P_\alpha^E(p,q)\, f_{(\beta)}(q), \, P_\alpha^F(p,q)\, \cH_{(\beta)}(q)\Big) \rho_\eps \, dy.
\end{split}
\end{equation}
To estimate the first argument of $\mathcal{C}$ in the integrand, we rewrite it as
\begin{equation}\label{eq:f_diff_split}
\begin{split}
(f_\eps)_{(\beta)}(p) - P_\alpha^E(p,q) f_{(\beta)}(q) &= \big( (f_\eps)_{(\beta)}(p) - f_{(\beta)}(p) \big)  \\
&\quad + \big( f_{(\beta)}(p) - P_\alpha^E(p,q) f_{(\beta)}(q) \big).
\end{split}
\end{equation}
We analyze this on the support of the mollifier $\rho_\eps(\psi_\alpha(p) - y)$, where the distance satisfies $d(p,q) \le C \eps$.

Concerning the first term here, note 
that for $\eps$ sufficiently small, $\chi_\alpha(p) \equiv 1$ on the support of the integral. Using $\sum_\alpha \xi_\alpha \equiv 1$, we can write 
\begin{equation}\label{eq:f_eps_error}
\begin{split}
&(f_\eps)_{(\beta)}(p) - f_{(\beta)}(p) \\
&= \sum_\alpha \chi_\alpha(p) \int_{\R^n} \xi_\alpha(q) \big( P_\alpha^E(p,q) f_{(\beta)}(q) - f_{(\beta)}(p) \big)\, \rho_\eps(\psi_\alpha(p) - y) \, dy \\
&\quad + \sum_\alpha \chi_\alpha(p) \int_{\R^n} \big( \xi_\alpha(q) - \xi_\alpha(p) \big)\, f_{(\beta)}(p)\, \rho_\eps(\psi_\alpha(p) - y) \, dy.
\end{split}
\end{equation}
Since $f \in C^{0,\alpha}_\text{loc}$ and $P_\alpha^E(p,q) = I_E + \mathcal{O}(\eps)$ (see \eqref{eq:transport-O-epsilon}), the integrand of the first term is bounded by 
\begin{equation}\label{eq:PEalpha}
\begin{split}
\left| P_\alpha^E(p,q) f_{(\beta)}(q) - f_{(\beta)}(p) \right|
&\le \left| f_{(\beta)}(q) - f_{(\beta)}(p) \right|\\
&\hspace*{-5em} + \left| P_\alpha^E(p,q) - I_E \right|\, \left| f_{(\beta)}(q) \right|
\le C \varepsilon^\alpha + C \varepsilon .
\end{split}
\end{equation}
As $\xi_\alpha$ is smooth, the second term in \eqref{eq:f_eps_error} is bounded by $O(\eps)$. Due to $\alpha \in (1/2, 1)$ we thus arrive at
\begin{equation}\label{eq:f_eps_bound}
\| (f_\eps)_{(\beta)} - f_{(\beta)} \|_{L^\infty(K)} \le C \eps^\alpha.
\end{equation}
The remaining term in \eqref{eq:f_diff_split} is precisely of the form \eqref{eq:PEalpha},
so we conclude that
\begin{equation}\label{eq:f_diff_total_bound}
\left| (f_\eps)_{(\beta)}(p) - P_\alpha^E(p,q) f_{(\beta)}(q) \right| \le C \eps^\alpha.
\end{equation}
Since the pairing $\mathcal{C}$ is bilinear, it satisfies the continuous bound $|\mathcal{C}(v, w)| \le C |v| |w|$ relative to the local norms. Applying this to \eqref{eq:difference} and using \eqref{eq:f_diff_total_bound}, we obtain
\begin{align}\label{eq:Delta_eps_bound}
|\Delta_\eps(p)|\le C \eps^\alpha \sum_\alpha \chi_\alpha(p) \int_{\R^n} \xi_\alpha(q) \left| \cH_{(\beta)}(q) \right| \rho_\eps(\psi_\alpha(p) - y) \, dy,
\end{align}
where we absorbed the uniform bound for the smooth transport matrix $P_\alpha^F(p,q)$ into the constant $C$. 

The remaining sum of integrals is precisely the local representation of the manifold convolution applied to the scalar function $|\cH|$ (the pointwise norm of $\cH$), so
\begin{equation}\label{eq:Delta_eps_bound2}
|\Delta_\eps(p)| \le C \eps^\alpha (|\cH|)_\eps(p).
\end{equation}

Since $\cH \in L^p_\text{loc}(\mathcal{M}, F)$, its pointwise norm $|\cH|$ is in $L^p_\text{loc}(\mathcal{M}, \R)$. Taking the supremum over $p \in K$ and using Young's inequality, we have
\begin{equation*}
\left\| (|\cH|)_\eps \right\|_{L^\infty(K)} \le C \| \cH \|_{L^p(L)} \| \rho_\eps \|_{L^{p'}(\R^n)},
\end{equation*}
where $1/p + 1/{p'} = 1$ and $L$ is some compact neighbourhood of $K$. Here, $\| \rho_\eps \|_{L^{p'}(\R^n)} \le C \eps^{-n/p}$, so
\begin{equation*}
\| \Delta_\eps \|_{L^\infty(K)} \le  C \eps^{\alpha - n/p}.
\end{equation*}
Since $\alpha - n/p = 1 - 2n/p>0$, this finishes the proof.
\end{proof}

Finally we turn to first and second order commutator estimates for $L^p$ and $W^{2,p}$-sections. The Hessian commutator of item (ii) of the following statement was essential in isolating the critical $L^p$-term in the splitting \eqref{eq:split-Lp} in the proof of Theorem \ref{Prop:Mean-curvature-convergence}. Item (i) on the other hand was used in combination with Lemma \ref{lem:m+} to establish the order of $\eps$ bound on the difference $\|\check{\mathbf{U}}_\eps-\mathbf{U}_\eps\|_{L^\infty(K)}$ used to estimate term (I).

\begin{Lemma}[Commutator estimates for manifold convolution]\label{lem:commutator_estimates}
Let $p > n$, so that $\alpha = 1 - n/p \in (0,1)$, and let $K \Subset \mathcal{M}$ be a compact set and denote by $\nabla_h$ the Levi-Civita connection of the background metric. Then for $\eps \in (0,1]$ sufficiently small, the following commutator estimates hold:
\begin{itemize}
 \item[(i)] First derivative commutator for $L^p_\loc$-tensors:\label{lem:first_deriv} \\
 For any tensor bundle $E$ over $\M$ and any section $T \in L^p_\loc(\M, E)$,
\begin{equation}\label{eq:first_deriv_Lp}
 \| \nabla_h (T \star_\M \rho_\eps) - (\nabla_h T) \star_\M \rho_\eps \|_{L^\infty(K)} = \mathcal{O}(\eps^\alpha),
\end{equation}
where $\nabla_h T$ is the distributional covariant derivative. Furthermore, for $p=\infty$,
i.e., $T \in L^\infty_\loc(\M, E)$, \eqref{eq:first_deriv_Lp} holds with $\alpha=1$.
\item[(ii)] Hessian commutator for $W^{2,p}_\loc$-scalars:\label{lem:hess-comm} \\
For any scalar function $\T \in W^{2,p}_\loc(\M)$,
\begin{equation}\label{eq:hessian_comm}
\| \nabla_h^2 (\T \star_\M \rho_\eps) - (\nabla_h^2 \T) \star_\M \rho_\eps \|_{L^\infty(K)} = \mathcal{O}(\eps^\alpha).
\end{equation}
\end{itemize}
\end{Lemma}
\begin{proof} 
Again we will employ the notations from the proof of Lemma \ref{lem:nonlinear-commutator-estimate}.
Fix $p \in K$ within a fixed reference chart $(U_\beta, \psi_\beta)$.
We denote the coordinates of the evaluation point~$p$ and the
integration point~$q$ in the integration ($\alpha$-) and reference
($\beta$-) charts by
\beq\label{eq:coords}
  x^\mu := \psi^\mu_\alpha(p),\quad
  \bar x^\mu := \psi^\mu_\beta(p),\qquad
  y^\mu := \psi^\mu_\alpha(q),\quad
  \bar y^\mu := \psi^\mu_\beta(q),
\eeq
so that unbarred coordinates refer to the $\alpha$-chart and barred
ones to the $\beta$-chart throughout.
Local components of tensor fields in the $\beta$-chart are denoted
by a subscript~$(\beta)$.
Let $P^E_\alpha(p,q) :=
J^E_{\alpha\to\beta}(p)\,J^E_{\beta\to\alpha}(q)$
be the coordinate parallel transport matrix for the bundle~$E$.
We drop the outer cut-off $\chi_\alpha(p)$ since
$\chi_\alpha \equiv 1$ on the support of the mollifier for
small~$\eps$.
 
It will be convenient to introduce the pulled-back partial
derivative operator $\tilde\partial^{(p,q)}_\mu$.
Geometrically, this operator represents the coordinate parallel
transport of the differential operator
$\partial/\partial \bar y^\nu$
(acting on the $\beta$-chart coordinates at~$q$) from the
integration point $q \in U_\alpha$ to the evaluation point
$p \in U_\beta$.
Recall that the coordinate parallel transport matrix for the
cotangent bundle $T^*\!M$ from $q$ to $p$ is defined as the product
of the chart transition matrices:
$P^{T^*\!M}_\alpha(p,q) :=
J^{T^*\!M}_{\alpha\to\beta}(p)\,J^{T^*\!M}_{\beta\to\alpha}(q)$.

Because covectors ($1$-forms) transform via the inverse transposed Jacobian, the transition matrix mapping $1$-form components from chart $\beta$ to chart $\alpha$ at the point $q$ is given by $J^{T^*\!M}_{\beta\to\alpha}(q)^\nu_{\ \sigma}
= \frac{\partial \bar y^\nu}{\partial y^\sigma}(q)$.
Conversely, the transition matrix mapping from chart~$\alpha$ back
to chart~$\beta$ at the evaluation point~$p$ is
$J^{T^*\!M}_{\alpha\to\beta}(p)^\sigma_{\ \mu}
= \frac{\partial x^\sigma}{\partial \bar x^\mu}(p)$.
Thus,
\beq\label{eq:Pcot}
  P^{T^*\!M}_\alpha(p,q)^\nu_{\ \mu}
  = J^{T^*\!M}_{\alpha\to\beta}(p)^\sigma_{\ \mu}\,
    J^{T^*\!M}_{\beta\to\alpha}(q)^\nu_{\ \sigma}
  = \frac{\partial x^\sigma}{\partial \bar x^\mu}(p)\,
    \frac{\partial \bar y^\nu}{\partial y^\sigma}(q).
\eeq
Applying this transport matrix to the standard partial derivative
$\partial/\partial \bar y^\nu$, we define
\beq\label{eq:transported-deriv}
  \tilde\partial^{(p,q)}_\mu
  := P^{T^*\!M}_\alpha(p,q)^\nu_{\ \mu}\,
     \frac{\partial}{\partial \bar y^\nu}
  = \frac{\partial x^\sigma}{\partial \bar x^\mu}(p)\,
    \biggl(
      \frac{\partial \bar y^\nu}{\partial y^\sigma}(q)\,
      \frac{\partial}{\partial \bar y^\nu}
    \biggr)
  = \frac{\partial x^\sigma}{\partial \bar x^\mu}(p)\,
    \frac{\partial}{\partial y^\sigma}\,.
\eeq
To unify notations, we shall also write distributional actions as
integrals, with the understanding that distributional derivatives
(of $L^p$-tensors) are defined via integration by parts. With these
preparations in place we can now proceed with the proof.

(i) By \eqref{eq:moll_def},
\begin{equation}\label{eq:Teps_def}
(T_\eps)_{(\beta)}(p) = \sum_\alpha \int_{\R^n} \xi_\alpha(q) P_\alpha^E(p,q) T_{(\beta)}(q) \rho_\eps(\psi_\alpha(p) - y) \, dy.
\end{equation}

Then $\nabla^h_\mu (T_\eps)_{(\beta)}(p) = \partial_{\bar x^\mu} (T_\eps)_{(\beta)}(p) + \Gamma_\mu^E(p) (T_\eps)_{(\beta)}(p)$. Applying $\partial_{\bar x^\mu}$ to $(T_\eps)_{(\beta)}(p)$, note that $\partial_{\bar x^\mu} [\rho_\eps(\psi_\alpha(p) - y)] = \frac{\partial \psi_\alpha^\sigma}{\partial \bar x^\mu}(p) (\partial_{z^\sigma} \rho_\eps)(\psi_\alpha(p) - y)$, where $z = \psi_\alpha(p) - y=x-y$. Since $\partial_{z^\sigma}$ is equivalent to $-\partial_{y^\sigma}$ when acting on $\rho_\eps(x - y)$, this gives
\begin{align}
   \partial_{\bar x^\mu} [\rho_\eps(\psi_\alpha(p) - y)]&
   = - \frac{\partial \psi_\alpha^\sigma}{\partial \bar x^\mu}(p) \partial_{y^\sigma} [\rho_\eps(\psi_\alpha(p) - y)]\nonumber \\ \label{eq:x-z}
   & = - \tilde{\partial}_\mu^{(p,q)} [\rho_\eps(\psi_\alpha(p) - y)]
\end{align}

We then (distributionally) integrate by parts with respect to the $y$ variables, shifting $-\tilde{\partial}_\mu^{(p,q)}$ from the mollifier to the smooth terms $\xi_\alpha(q) P_\alpha^E(p,q)$ and the distributional tensor $T_{(\beta)}(q)$. This gives
\begin{align*}
\nabla^h_\mu &(T_\eps)_{(\beta)}(p) = \sum_\alpha \int_{\R^n} \rho_\eps(\psi_\alpha(p) - y) \Big\{ \big(\tilde{\partial}_\mu^{(p,q)} \xi_\alpha(q)\big) P_\alpha^E(p,q) T_{(\beta)}(q) \\
&\qquad + \xi_\alpha(q) \big(\tilde{\partial}_\mu^{(p,q)} P_\alpha^E(p,q)\big) T_{(\beta)}(q) + \xi_\alpha(q) P_\alpha^E(p,q) \big(\tilde{\partial}_\mu^{(p,q)} T_{(\beta)}(q)\big) \\
&\qquad + \xi_\alpha(q) \big(\partial_{\bar x^\mu} P_\alpha^E(p,q)\big) T_{(\beta)}(q) + \Gamma_\mu^E(p) \xi_\alpha(q) P_\alpha^E(p,q) T_{(\beta)}(q) \Big\} \, dy.
\end{align*}
On the other hand, when we expand $((\nabla_h T)_\eps)_{(\beta)\mu}(p)$, using the convolution formula \eqref{eq:moll_def} for the (distributional) section $\nabla^h T \in \D'(M, T^*\!M \otimes E)$ we need its local components at the integration point $q$, which are $\partial_{\bar y^\nu} T_{(\beta)}(q)+ \Gamma^E_\nu(q)\, T_{(\beta)}(q)$. To transport this $T^*M \otimes E$-valued tensor back to $p$, we apply $P_\alpha^{T^*M \otimes E}(p,q) = P_\alpha^{T^*M}(p,q)_\mu^\nu \otimes P_\alpha^E(p,q)$, to obtain
\begin{align*}
&((\nabla_h T)_\eps)_{(\beta)\mu}(p) =\\ 
&\sum_\alpha \int_{\R^n} \xi_\alpha(q)\ \rho_\eps(\psi_\alpha(p) - y) \Big[ P_\alpha^{T^*M}(p,q)_\mu^\nu\ P_\alpha^E(p,q) \Big( \partial_{\bar y^\nu} T_{(\beta)}(q)+ \Gamma^E_\nu(q)\, T_{(\beta)}(q) \Big) \Big] \, dy \\
&= \sum_\alpha \int_{\R^n} \xi_\alpha(q)\ \rho_\eps(\psi_\alpha(p) - y)\ P_\alpha^E(p,q)     \big(\underbrace{P^{T^*\!M}_\alpha(p,q)^\nu_{\ \mu}\,
    \frac{\partial}{\partial \bar y^\nu}}_{
      =\,\tilde\partial^{(p,q)}_\mu}\,\big)
    T_{(\beta)}(q)\, dy \\
&\quad + \sum_\alpha \int_{\R^n} \xi_\alpha(q)\ \rho_\eps(\psi_\alpha(p) - y)\ P_\alpha^E(p,q)\, P_\alpha^{T^*M}(p,q)_\mu^\nu\, \Gamma_\nu^E(q)\, T_{(\beta)}(q) \, dy.
\end{align*}
Now calculating the commutator $C_\eps(T) := \nabla_h (T \star_\M \rho_\eps) - (\nabla_h T) \star_\M \rho_\eps$, we see that the principal derivative term $\sum_\alpha \int_{\R^n} \rho_\eps(\psi_\alpha(p) - y) \xi_\alpha(q) P_\alpha^E(p,q) \big(\tilde{\partial}_\mu^{(p,q)} T_{(\beta)}(q)\big) \, dy$ cancels exactly. Factoring $T_{(\beta)}(q)$ out of the remaining terms, we are left with:
\begin{equation}\label{eq:Ceps_integral}
C_\eps(T)_{(\beta)\mu}(p) 
= \sum_\alpha \int_{\R^n} \rho_\eps(\psi_\alpha(p) - y)\, M_{\alpha, \mu}(p, q)\, T_{(\beta)}(q) \, dy,
\end{equation}
where $M_{\alpha, \mu}(p,q)$ is given by:
\begin{equation*}
M_{\alpha, \mu}(p,q) = \big(\tilde{\partial}_\mu^{(p,q)} \xi_\alpha(q)\big) P_\alpha^E(p,q) + \xi_\alpha(q) \mathcal{E}_{\alpha, \mu}(p,q),
\end{equation*}
with 
\begin{align*}
\mathcal{E}_{\alpha, \mu}(p,q) :=& 
\partial_{\bar x^\mu} P_\alpha^E(p,q) + \tilde{\partial}_\mu^{(p,q)} P_\alpha^E(p,q)\\& + \Gamma_\mu^E(p) P_\alpha^E(p,q) - P_\alpha^E(p,q) P_\alpha^{T^*M}(p,q)_\mu^\nu \Gamma_\nu^E(q).
\end{align*}
Next we claim that $\mathcal{E}_{\alpha, \mu}(p,p)=0$. To see this, note first that $q=p$
implies $y = \psi_\alpha(p)$. Now $P_\alpha^E(p,p) = J_{\alpha \to \beta}^E(p) J_{\beta \to \alpha}^E(p) = I_E$. Furthermore, since $P_\alpha^{T^*M}(p,p)_\mu^\nu = \delta_\mu^\nu$, the operator $\tilde{\partial}_\mu^{(p,p)}$ reduces to $\partial_{\bar y^\mu}$ at $p$, i.e., to $\partial_{\bar x^\mu}$.
Since $z \mapsto P^E_\alpha(z,z) \equiv I_E$, we have
$\partial_{\bar x^\mu} P^E_\alpha(p,p)
+ \partial_{\bar y^\mu} P^E_\alpha(p,p) = 0$
so the first two terms in $E_{\alpha,\mu}(p,p)$ cancel,
as do the remaining terms due to
$\Gamma^E_\mu(p)\, I_E
- I_E\, \delta^\nu_\mu\, \Gamma^E_\nu(p) = 0$. So indeed $\mathcal{E}_{\alpha, \mu}(p,p)=0$, leaving $M_{\alpha, \mu}(p,p) = \big( \partial_{\bar x^\mu} \xi_\alpha(p) \big) I_E$. Since $\sum_\alpha \xi_\alpha \equiv 1$,
\begin{equation}\label{eq:kernel_cancel}
\sum_\alpha M_{\alpha, \mu}(p,p) = \Big( \partial_{\bar x^\mu} \sum_\alpha \xi_\alpha(p) \Big) I_E = 0.
\end{equation}
We can therefore rewrite the commutator in the form
\begin{equation*}
C_\eps(T)_{(\beta)\mu}(p) 
 = \sum_\alpha \int_{\R^n} \rho_\eps(\psi_\alpha(p) - y) \big[ M_{\alpha, \mu}(p, q) - M_{\alpha, \mu}(p,p) \big] T_{(\beta)}(q) \, dy.
\end{equation*}
Since $M_{\alpha, \mu}$ is smooth, Taylor expansion bounds the bracket by $\mathcal{O}(d^h(p,q)) \le \mathcal{O}(\eps)$ on the support of $\rho_\eps$, so that
\begin{equation*}
|C_\eps(T)_{(\beta)\mu}(p)| 
\le C \eps \int_{\R^n} \rho_\eps(\psi_\alpha(p) - y) |T_{(\beta)}(q)| \, dy.
\end{equation*}
If $p < \infty$, applying H\"older's inequality with conjugate exponents $p$ and $p'$ 
yields:
\begin{equation*}
|C_\eps(T)_{(\beta)\mu}(p)| 
\le C \eps \, \| \rho_\eps \|_{L^{p'}(\R^n)} \, \| T_{(\beta)} \|_{L^p(L)}, 
\end{equation*}
where $L$ is a compact neighbourhood of $K$. 
Since $\| \rho_\eps \|_{L^{p'}} = \eps^{-n/p} \|\rho\|_{L^{p'}}$, we arrive at
\begin{equation}\label{eq:holder_bound}
\|C_\eps(T)\|_{L^\infty(K)} \le C \eps \cdot \eps^{-n/p} \| T \|_{L^p(L)} = \mathcal{O}(\eps^{1-n/p}) = \mathcal{O}(\eps^\alpha).
\end{equation}

If $p = \infty$, we can bound $|T_{(\beta)}(q)|$ directly by its $L^\infty$-norm on a compact neighbourhood $L$ of $K$, which gives
\begin{equation*}
|C_\eps(T)_{(\beta)\mu}(p)| 
\le C \eps \, \| T_{(\beta)} \|_{L^\infty(L)} \int_{\R^n} \rho_\eps(\psi_\alpha(p) - y) \, dy = C \eps \, \| T_{(\beta)} \|_{L^\infty(L)}.
\end{equation*}
Thus $\|C_\eps(T)\|_{L^\infty(K)} = \mathcal{O}(\eps)$.

This completes the proof of (i).

(ii) We begin by rewriting the Hessian commutator $H_\eps(\T) = [\nabla_h^2, \star_\M\rho_\eps]\T$ in terms of the first-derivative commutator $C_\eps$ from (i) corresponding to $\T$ and $\nabla_h \T$, respectively:
\begin{equation}\label{eq:Hess_decomp}
\begin{split}
H_\eps(\T) &= \nabla_h \big[ \nabla_h (\T \star_\M \rho_\eps) - (\nabla_h \T) \star_\M \rho_\eps \big]\\ 
&\hphantom{-\ \ \,}+ \big[ \nabla_h \big( (\nabla_h \T) \star_\M \rho_\eps \big) - \big(\nabla_h (\nabla_h \T)\big) \star_\M \rho_\eps \big]\\ 
&= \nabla_h C_\eps(\T) + C_\eps(\nabla_h \T).
\end{split}
\end{equation}
Since $\nabla_h \T \in W^{1,p}_\loc \subseteq L^p_\loc$, part (i) immediately yields $\| C_\eps(\nabla_h \T) \|_{L^\infty(K)} = \mathcal{O}(\eps^\alpha)$. The remaining term takes the form 
\[
\nabla^h_\nu C_\eps(\T)_\mu(p) = \partial_{\bar x^\nu} C_\eps(\T)_\mu(p) - \Gamma^\lambda_{\nu\mu}(p) C_\eps(\T)_\lambda(p)
\]
Applying (i) to $\T$ gives $\|C_\eps(\T)\|_{L^\infty} = \mathcal{O}(\eps^\alpha)$, so the second term is bounded by $\mathcal{O}(\eps^\alpha)$, and we are left with bounding $\partial_{\bar x^\nu} C_\eps(\T)_\mu(p)$, where $C_\eps$ is given by \eqref{eq:Ceps_integral}. In our present case, $T=\T$ is scalar, so the corresponding bundle is trivial ($P_\alpha^\R \equiv 1$, $\Gamma^\R \equiv 0$), and the kernel simplifies to $M_{\alpha, \mu}(p,q) = \tilde{\partial}_\mu^{(p,q)} \xi_\alpha(q)$. Applying $\partial_{\bar x^\nu}$ yields
\begin{equation*}
\partial_{\bar x^\nu} C_\eps(\T)_\mu(p) = \sum_\alpha \int_{\R^n} (\partial_{\bar x^\nu} \rho_\eps) M_{\alpha, \mu} \T \, dy + \sum_\alpha \int_{\R^n} \rho_\eps (\partial_{\bar x^\nu} M_{\alpha, \mu}) \T \, dy.
\end{equation*}
Next we rewrite $\partial_{\bar x^\nu} [\rho_\eps(\psi_\alpha(p) - y)]$ using \eqref{eq:x-z}
and integrate the first term by parts with respect to $y$ 
to obtain
\begin{align*}
\int (-\tilde{\partial}_\nu^{(p,q)} \rho_\eps) M_{\alpha, \mu} \T \, dy &= \int \rho_\eps \tilde{\partial}_\nu^{(p,q)}(M_{\alpha, \mu} \T) \, dy \\
&= \int \rho_\eps (\tilde{\partial}_\nu^{(p,q)} M_{\alpha, \mu}) \T \, dy + \int \rho_\eps M_{\alpha, \mu} (\tilde{\partial}_\nu^{(p,q)} \T) \, dy.
\end{align*}
Substituting this back, we obtain
\begin{equation}\label{eq:Hess_bracket}
\begin{split}
\partial_{\bar x^\nu} C_\eps(\T)_\mu(p) &= \sum_\alpha \int_{\R^n} \rho_\eps \underbrace{\big[ \partial_{\bar x^\nu} M_{\alpha, \mu} + \tilde{\partial}_\nu^{(p,q)} M_{\alpha, \mu} \big]}_{=: B_{\alpha, \mu, \nu}(p,q)} \T(q) \, dy\\ 
&\quad + \sum_\alpha \int_{\R^n} \rho_\eps M_{\alpha, \mu} \big( \tilde{\partial}_\nu^{(p,q)} \T(q) \big) \, dy.
\end{split}
\end{equation}
Now recalling that at $q=p$ the operator $\tilde{\partial}_\nu^{(p,p)}$ reduces to $\partial_{\bar y^\nu}\big|_p = \partial_{\bar x^\nu}$. Consequently, $B_{\alpha, \mu, \nu}(p,p)$ is the derivative of the kernel along the diagonal
\begin{equation*}
B_{\alpha, \mu, \nu}(p,p) = \partial_{\bar x^\nu} M_{\alpha, \mu}(p,p) + \partial_{\bar y ^\nu} M_{\alpha, \mu}(p,p) = \frac{\partial}{\partial \bar x^\nu}\bigg|_{p} \big( z \mapsto M_{\alpha, \mu}(z,z) \big).
\end{equation*}
Since by \eqref{eq:kernel_cancel} we have $\sum_\alpha M_{\alpha, \mu}(z,z) = 0$ for any $z \in \mathcal{M}$, the derivative along the diagonal vanishes, i.e., $\sum_\alpha B_{\alpha, \mu, \nu}(p,p) = 0$ and similar to the above we may rewrite the first integral in \eqref{eq:Hess_bracket} as
\begin{equation*}
\sum_\alpha \int_{\R^n} \rho_\eps \big[ B_{\alpha, \mu, \nu}(p,q) - B_{\alpha, \mu, \nu}(p,p) \big] \T(q) \, dy.
\end{equation*}
Since $B_{\alpha, \mu, \nu}$ is smooth, the bracket is bounded on the support of $\rho_\eps$ by $\mathcal{O}(d^h(p,q)) \le \mathcal{O}(\eps)$. Because $\T \in W^{2,p}_\loc$ is locally bounded, the first integral is bounded by $\mathcal{O}(\eps)$. 

Similarly using \eqref{eq:kernel_cancel}, we can rewrite the second integral in \eqref{eq:Hess_bracket} in the form
\begin{equation*}
\sum_\alpha \int_{\R^n} \rho_\eps(\psi_\alpha(p) - y) \Big[ M_{\alpha, \mu}(p,q) \big( \tilde{\partial}_\nu^{(p,q)} \T(q) \big) - M_{\alpha, \mu}(p,p) \partial_{\bar x^\nu} \T(p) \Big] \, dy.
\end{equation*}
We split this difference into the two terms
\begin{align*}
I_1 &= \sum_\alpha \int_{\R^n} \rho_\eps(\psi_\alpha(p) - y) \big[ M_{\alpha, \mu}(p,q) - M_{\alpha, \mu}(p,p) \big] \big( \tilde{\partial}_\nu^{(p,q)} \T(q) \big) \, dy, \\
I_2 &= \sum_\alpha \int_{\R^n} \rho_\eps(\psi_\alpha(p) - y) M_{\alpha, \mu}(p,p) \big[ \tilde{\partial}_\nu^{(p,q)} \T(q) - \partial_{\bar x^\nu} \T(p) \big] \, dy
\end{align*}
and bound each term separately.

For $I_1$, by a Taylor expansion we obtain a bound of $\mathcal{O}(d^h(p,q)) \le \mathcal{O}(\eps)$ for the bracket on the support of $\rho_\eps$. Also, $\tilde{\partial}_\nu^{(p,q)} \T(q)$ is bounded pointwise by $C|\nabla_h \T(q)|$ which is locally bounded and so we obtain
\begin{equation*}
|I_1| \le C \eps \int_{\R^n} \rho_\eps |\nabla_h \T(q)| \, dy \le C\, \eps\, \|\nabla_h \T\|_{L^\infty(L)} = \mathcal{O}(\eps).
\end{equation*}
For $I_2$, the matrix $M_{\alpha, \mu}(p,p) = \partial_{\bar x^\mu} \xi_\alpha(p) I_E$ is $\mathcal{O}(1)$. The term in the brackets is the difference between the transported 1-form $P_\alpha^{T^*M}(p,q) d\T(q)$ and the local 1-form $d\T(p)$. Adding and subtracting $P_\alpha^{T^*M}(p,q) \nabla_h \T(p)$, we bound this by
\begin{equation*}
\big| P_\alpha^{T^*M}(p,q) \big( \nabla_h \T(q) - \nabla_h \T(p) \big) \big| + \big| \big( P_\alpha^{T^*M}(p,q) - I \big) \nabla_h \T(p) \big|.
\end{equation*}
Here, $\nabla_h \T \in C^{0,\alpha}$ with $\alpha = 1 - n/p$. The first term is thus bounded by $C d^h(p,q)^\alpha \le C \eps^\alpha$. For the second term, since $P_\alpha^{T^*M}(p,p)=I$, the bracket is bounded by $\mathcal{O}(d^h(p,q)) \le \mathcal{O}(\eps)$. As $\nabla_h \T(p)$ is locally bounded, the entire bracket
in $I_2$ is $\mathcal{O}(\eps^\alpha) + \mathcal{O}(\eps) = \mathcal{O}(\eps^\alpha)$ on the support of $\rho_\eps$. Consequently, $|I_2| = \mathcal{O}(\eps^\alpha)$.

Combining these bounds yields $\partial_{\bar x^\nu} C_\eps(\T)_\mu(p) = \mathcal{O}(\eps) + \mathcal{O}(\eps) + \mathcal{O}(\eps^\alpha) = \mathcal{O}(\eps^\alpha)$. In view of \eqref{eq:Hess_decomp}, this completes the proof.

\end{proof}

\section{The Hawking Singularity Theorem} \label{sec:hawking}

In this section we will formulate and prove the Hawking theorem for $W^{1,p}_\text{loc}$-spacetimes with $L^p_\text{loc}$-bounded curvature ($p>2n$). 
We will present the statement in its two standard forms: in the globally hyperbolic case, we obtain an explicit bound for the time-separation function, whereas without global hyperbolicity we merely conclude timelike incompleteness. Both versions will be given in a general, quantified form under corresponding Ricci and mean curvature bounds. As preparation, we first briefly discuss how Ricci bounds are to be understood in our low-regularity setting, and then recall the classical smooth formulation.

\subsection{Ricci bounds} \label{sec:ricbounds}

As announced, we will not only deal with the `classical' case of the Hawking theorem where one assumes the \emph{strong energy condition} (SEC), $\Ric(X,X) \geq 0$ for all timelike vectors, but with general lower Ricci bounds. We begin by formulating these for $g \in W^{1,p}_\text{loc}$. By our underlying assumption $\Riem[g] \in L^p_\text{loc}$, this is quite straightforward. In fact, it suffices to assume $g\in C^0$ and $\Ric\in L^p_\text{loc}$, and we do so in the following. 

\begin{Def}[Ricci bounds]
Let $\kappa\in \R$. We say that a continuous spacetime $(\M,g)$ with $\Ric\in L^p_\text{loc}$ has Ricci curvature bounded below by $\kappa$ if  
\begin{equation}\label{eq:lrb}
  \Ric(X,X)\geq (n-1) \kappa \quad\mbox{point-wise almost everywhere}
\end{equation} 
for all smooth timelike unit vector fields $X$.
\end{Def}

Obviously condition \eqref{eq:lrb} is equivalent to 
\begin{enumerate}
    \item for almost all points $p\in M$ we have $\Ric(v,v)\geq (n-1)\kappa$ for all timelike unit tangent vectors $v$ in $T_pM$, and to
    \item condition \eqref{eq:lrb}  as above but for all continuous (or even merely measurable) timelike vector fields.
\end{enumerate}
Note that RT-regularisation preserves Ricci bounds. Even more generally, let $\phi$ be a $C^1$-isometry between spacetimes as in Section \ref{Sec_causality}. Then it maps null sets to null sets and so preserves the tangent space condition (1) above by the fact that $T_p\phi$ is an isometry of the tangent spaces. 

\begin{Lemma}[Preservation of Ricci bounds]\label{lem:prb}
  Let $\kappa\in\R$ and let $\phi$ be a $C^1$-diffeo\-mor\-phism. Then $\Ric[g]$ has curvature bounded below by $\kappa$ if and only if  $\Ric[\phi_*g]$ has curvature bounded below by $\kappa$.
\end{Lemma}

We will also be concerned with Ricci bounds for $g\in C^1$ with no additional regularity assumptions on the curvature, as discussed in \cite[Sec.\ 3.2]{Graf}. In this setting, the  Ricci tensor is a distribution and bounds have to be formulated in the weak sense, i.e., 
\begin{equation}\label{eq:drb}
    \langle\Ric(X,X)-(n-1)\kappa,\varphi\rangle \geq 0
\end{equation}
for all smooth timelike unit vector fields $X$ and all non-negative test densities $\varphi$, where $\langle\ ,\ \rangle$ denotes the distributional action. It is then clear that if $\Ric\in L^p_\text{loc}$, conditions \eqref{eq:drb} and \eqref{eq:lrb} are equivalent. 

We need the following result that guarantees that Ricci bounds are almost preserved under smoothing.

\begin{Prop}\label{Prop:Graf3.13-modified}
Let $(\M,g)$ be a $C^1$-spacetime. Suppose that $\mathrm{Ric}(X, X) \ge (n-1)\kappa$ for some $\kappa\in\R$ and all smooth unit timelike vector fields $X$ on $M$. Then for any compact set $K\Subset TM$, and any $\delta>0$ there exists $\eps_0=\eps_0(K,\delta)>0$ such that for all $\eps\in (0,\eps_0)$ and all $X\in K$ with $\check{g}_\eps(X,X) = -1$ we have
\begin{equation*}
 \Ric[\check g_\eps](X,X) \ge (n-1)(\kappa-\delta).
\end{equation*}
\end{Prop}
\begin{proof} For $g\in C^{1,1}$, this is precisely \cite[Lem.\ 3.13]{Graf_Volume}. That the result holds even for $g\in C^1$ then follows from the proof of \cite[Lem.\ 4.6]{Graf}: Since 
${g}\star_M \rho_\eps \to g$ even in $C^1$, it follows that also 
\[
(\Ric[\check{g}_\eps] - (n-1)\kappa \check{g}_\eps) - (\Ric[g]-(n-1)\kappa g)*_M \rho_\eps \to 0
\]
locally uniformly on $TM$.    
\end{proof}

\subsection{Hawking Theorem - version 1}
Our aim in this section is to derive a low regularity version of Hawking's singularity theorem in the globally hyperbolic setting for general (quantified) assumptions on the energy condition and the mean curvature of the Cauchy surface. In the smooth setting, such versions have appeared in \cite{Borde} and \cite{Andersson-Galloway}, cf.\ also \cite{Graf_Volume}. 
The precise formulation is as follows:

\begin{Thm}[Hawking, $C^\infty$-case]\label{Th:Hawking-smooth} Let $(\M, g)$ be a smooth globally hyperbolic spacetime with a smooth spacelike Cauchy hypersurface $\Sigma$. 
Suppose there exist $\beta, \kappa\in \R$ such that
\begin{enumerate}
    \item[(i)] $\mathrm{Ric}(X, X) \ge (n-1)\kappa$ for all unit timelike vectors $X$, and
    \item[(ii)] the mean curvature $H$ of $\Sigma$ satisfies $H \le \beta$.
\end{enumerate}
Then $\tau_\Sigma \le b_{\kappa,\beta}$, 
where $b_{\kappa,\beta}$ is defined as follows:
\begin{enumerate}
\item If $\kappa>0$, then
\[
b_{\kappa,\beta}
=
\frac{1}{\sqrt{\kappa}}\cot^{-1}\!\left(-\frac{\beta}{(n-1)\sqrt{\kappa}}\right),
\qquad \beta\in\mathbb{R}.\]
where $\cot^{-1}:\R\to(0,\pi)$ denotes the principal branch.\item If $\kappa=0$, then
\[
b_{\kappa,\beta}
=
\begin{cases}
\infty, & \beta\ge 0,\\[0.3em]
-\dfrac{n-1}{\beta}, & \beta<0.
\end{cases}
\]
\item If $\kappa<0$, then
\[
b_{\kappa,\beta}
=
\begin{cases}
\infty,
& \displaystyle \frac{\beta}{(n-1)\sqrt{-\kappa}}\ge -1,\\[0.8em]
\dfrac{1}{\sqrt{-\kappa}}\,
\coth^{-1}\!\left(-\dfrac{\beta}{(n-1)\sqrt{-\kappa}}\right),
& \displaystyle \frac{\beta}{(n-1)\sqrt{-\kappa}}<-1,
\end{cases}
\]
where $\coth^{-1}:(1,\infty)\to(0,\infty)$ is the standard inverse hyperbolic cotangent.
\end{enumerate}
\end{Thm}
For completeness, we include a proof of this version\footnote{Independent of the one in \cite{Graf_Volume}.} of the result in Appendix \ref{Sec_Hawking_proof}. We shall require the following observation on the regularity of $b_{\kappa,\beta}$:
\begin{Lemma}\label{Lem:b-kappa-beta-continuous} Fix $\beta\in\R$. Then the map $\kappa\mapsto b_{\kappa,\beta}$ defined in  Theorem \ref{Th:Hawking-smooth} is continuous as a function $\R\to[0,\infty]$.
\end{Lemma}
\begin{proof}
Straightforward calculation.    
\end{proof}
The following result is a variant of \cite[Cor.\ 2.15]{Graf} (dropping the smoothness assumption on the Cauchy hypersurface). For completeness, we include a streamlined proof.
\begin{Prop} \label{prop_maximizing_geo}
    Let $(\M,g)$ be a globally hyperbolic $C^1$-spacetime and let $\Sigma\subseteq\M$ be a Cauchy hypersurface. Then for any $q\in I^+(\Sigma)$ there exists at least one timelike geodesic from $\Sigma$ to q maximizing the distance to $\Sigma$. Further, such a geodesic can be obtained as a $C^1$-limit of a sequence of $\check{g}_{\eps_k}$-geodesics $\gamma_{\eps_k}$, maximizing the $\check{g}_{\eps_k}$-distance from $\Sigma$ to $q$.
\end{Prop}
\begin{proof}
    Let $q\in I^+(\Sigma)$. Since $J^-(q)\cap J^+(\Sigma)$ is compact, we may use approximations $\check g_{\eps_k}$ as in Lemma \ref{Le:approximating metrics}(iv), i.e.\ the width of the $\check{g}_\eps$-lightcones is increasing as $k$ increases.
    Consequently $q\in I^+(\Sigma)$ implies that $q\in I^+_{\check{g}_{k}}(\Sigma)$ for all large $k$.

    Using that $\check{g}_{\eps_k}$ is globally hyperbolic, there exists a $\check{g}_{\eps_k}$-maximizing $\check{g}_{\eps_k}$-geodesic $\gamma_{\eps_k}$ from $\Sigma$ to $q$, i.e.
    \begin{equation}\label{eq:conttau}
        L_{\check{g}_{\eps_k}}(\gamma_{\eps_k})=\tau^{\check{g}_{\eps_k}}_{\Sigma}(q) \to \tau_\Sigma(q),
    \end{equation}
    using \cite[Lem.\ 5.5]{CGHKS25}. As in the proof of \cite[Prop.~2.13]{Graf}, passing to a subsequence if necessary, we may assume that the $\gamma_{\eps_k}$ converge to a $g$-geodesic $\gamma :[0,1]\to M$ in $C^1([0,1])$. Then
    \begin{equation}\label{eq:Lconv}
        L_{\check{g}_{\eps_k}}(\gamma_{\eps_k}) \to L_g(\gamma).
    \end{equation}

    Combining \eqref{eq:conttau} and \eqref{eq:Lconv}, we find $L_g(\gamma)=\tau_\Sigma(q)$, which implies that $\gamma$ is a maximizer.
\end{proof}

\begin{Thm}[$W^{1,p}$-Hawking singularity theorem, I]\label{Th:Hawking-I}
Let $(\M,g)$ be an $n$-dimen\-sional globally hyperbolic $W^{1,p}_\text{loc}$-spacetime with $\Riem[g]\in L^p_\text{loc}$ for some $p>2n$. Suppose there is a smooth spacelike Cauchy hypersurface $\Sigma$ and there are $\beta, \kappa\in \R$ such that
\begin{enumerate}
    \item[(i)] $\mathrm{Ric}(X, X) \ge (n-1)\kappa$ a.e.\ for all smooth unit timelike vector fields $X$.
    \item[(ii)] 
    $\Sigma$ has mean curvature bounded above by $\beta$ in the sense of Definition \ref{Def:mean-curvature-bound}.
\end{enumerate}
Then $\tau_\Sigma \le b_{\kappa,\beta}$, with $b_{\kappa,\beta}$ as defined in Theorem \ref{Th:Hawking-smooth}. 
\end{Thm}

\begin{proof}
Applying the RT-regularisation procedure of Theorem \ref{Thm_smoothing_preliminaries} to $g$, but keeping the same letters for the objects involved, we obtain a spacetime $(\M,g)$ with $W^{2,p}_\text{loc}$-metric $g$ satisfying (i) and (ii) by Lemma \ref{lem:prb} and Remark \ref{rem:cb}, respectively. 
More precisely $\Sigma$ is a $W^{2,p}_\text{loc}$-hypersurface, realized as the $0$-level set of $\T\in W^{2,p}_\text{loc}$ and there exists an open neighbourhood $U$ of $\Sigma$ such that 
\begin{equation}\label{eq:mean-curvature-improved}
\esssup \mathcal{H}_{g,\T,U}<\beta. 
\end{equation}
Since $\tau_\Sigma$ remains unchanged during the RT-regularisation by \eqref{eq:tau=tau}, our goal remains to show boundedness of $\tau_\Sigma$ in this setup.

Assume, to the contrary, that there exists some $q\in I^+(\Sigma)$ with $\tau_\Sigma(q) > b_{\kappa,\beta}+\eta$ for some $\eta>0$, see Figure \ref{fig:hawking-proof-schematic} for the setup. Continuity of $\kappa \mapsto b_{\kappa,\beta}$ (Lemma \ref{Lem:b-kappa-beta-continuous}) then allows us to pick $\delta >0$ so small that 
\begin{equation}\label{eq:indirect-assumption-estimate}
\tau_\Sigma(q) > b_{\kappa,\beta} + \eta > b_{\kappa-\delta,\beta}.
\end{equation}

Since $\Sigma_{-1}=\T^{-1}(-1)$ is a Cauchy surface, the set $C:=J^+(\Sigma_{-1}) \cap J^-(q)$ is compact and by Lemma \ref{Le:approximating metrics}(v) we may use an approximating sequence $\check g_{\eps_{k}}$ with $\check g_{\eps_{k}}\prec \check g_{\eps_{k+1}} \prec g$ on $C$.
To simplify notation, we set $g_k:=\check{g}_{\eps_k}$.  

Now Proposition \ref{prop_maximizing_geo} implies the existence of $g_k$-maximizing $g_k$-geodesics $\gamma_k$ from $\Sigma$ to $q$ that converge in $C^1$ to a $g$-maximizing geodesic $\gamma$ from $\Sigma$ to $q$ and such that 
\begin{equation}\label{eq:gamma_k-lengths}
L_{g_k}(\gamma_k) = \tau^{g_k}_\Sigma(q) \to L_g(\gamma) = \tau_\Sigma(q).
\end{equation}
Since $\T\in W^{2,p}_\text{loc}(M)\subseteq C^1(M)$ and $C$ is compact, we have $\| \T -\T\star_M \rho_\eps \|_{L^\infty(C)} = O(\eps)$. Together with Lemma \ref{Lemma_T-eps_timefct}(i) and the fact that $\Sigma= \T^{-1}(0)$ this implies that
\[
\T_{\eps_k}(p) = \T \star_{\M} \rho_{\eps_k}(p) + \sqrt{\eps_k} > 0 
\]
for $k$ large, uniformly for $p\in C\cap \Sigma$. Moreover, for $k$ large we also have $\T_{\eps_k}(p) < 0$ for all $p\in \Sigma_{-1}\cap C$. Altogether, there exists some $k_0$ such that (with $\Sigma^{\eps_k}_0=\T_{\eps_k}^{-1}(0)$)
\begin{equation}\label{eq:Cauchy-surface-positioning}
\Sigma_{-1}\cap C \subseteq J_{g_k}^-(\Sigma^{\eps_k}_0) \cap C 
\subseteq J_{g_k}^-(\Sigma)\cap C \subseteq J^-(\Sigma)\cap C,
\end{equation}
for each $k\ge k_0$. 

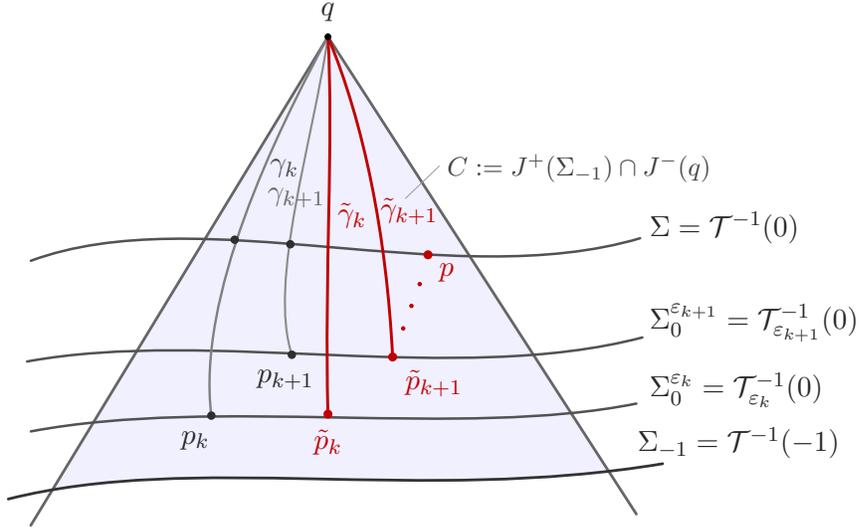
\begin{figure}[htbp]
    \centering
    \makebox[\textwidth][l]{\hspace*{1.10cm}%
    \begin{tikzpicture}[scale=1.0,line cap=round,line join=round]

        \coordinate (q)  at (0,6.7);
        \coordinate (L)  at (-3.95,0.20);
        \coordinate (R)  at ( 4.10,0.35);

        \path[fill=blue!6, draw=blue!20!black, line width=0.25pt]
            (q) -- (-3.6623,0.6734)
            .. controls (-1.4825,1.0466) and (1.0129,0.6051) .. (3.7324,0.9193)
            -- cycle;

        \draw[black!60, line width=1.05pt] (q) -- (L);
        \draw[black!60, line width=1.05pt] (q) -- (R);

        \draw[black!80, line width=1.05pt]
            (-4.25,0.55) .. controls (-1.75,1.18) and (1.20,0.48) .. (4.45,1.02);
        \draw[black!68, line width=0.95pt]
            (-3.95,1.45) .. controls (-1.45,1.92) and (1.75,1.35) .. (4.10,1.82);
        \draw[black!68, line width=0.95pt]
            (-4.00,2.38) .. controls (-1.55,2.83) and (1.70,2.05) .. (4.18,2.70);
        \draw[black!68, line width=0.95pt]
            (-3.95,3.72) .. controls (-1.70,4.55) and (1.65,3.28) .. (4.15,4.02);

        \draw[black!50, line width=0.9pt]
            (q) .. controls (-0.70,5.40) and (-1.75,3.25) .. (-1.55,1.66);
        \draw[black!45, line width=0.8pt]
            (q) .. controls (-0.20,5.30) and (-0.80,3.45) .. (-0.48,2.47);

        \draw[red!75!black, line width=1.1pt]
            (q) .. controls (0.08,5.30) and (-0.05,3.55) .. (0.00,1.68);
        \draw[red!75!black, line width=1.1pt]
            (q) .. controls (0.48,5.55) and (0.78,4.05) .. (0.86,2.44);

        \fill[black!82] (-1.24,4.00) circle (1.6pt);
        \fill[black!82] (-0.50,3.94) circle (1.6pt);
        \fill[red!75!black] (1.33,3.80) circle (1.7pt);

        \fill[black!82] (-0.48,2.47) circle (1.6pt);
        \fill[red!75!black] (0.86,2.44) circle (1.7pt);

        \fill[black!82] (-1.55,1.66) circle (1.6pt);
        \fill[red!75!black] (0.00,1.68) circle (1.7pt);

        \fill[red!75!black] (1.00,2.82) circle (0.85pt);
        \fill[red!75!black] (1.11,3.12) circle (0.85pt);
        \fill[red!75!black] (1.22,3.41) circle (0.85pt);

        \node[black!85, below left=2pt, xshift=0.18cm] at (-1.55,1.66) {$p_k$};
        \node[black!85, below left=2pt, xshift=0.50cm] at (-0.48,2.47) {$p_{k+1}$};

        \node[red!75!black, below=2pt] at (0.00,1.68) {$\tilde p_k$};
        \node[red!75!black, below right=1pt] at (0.86,2.44) {$\tilde p_{k+1}$};
        \node[red!75!black, right=3pt] at (1.24,3.56) {$p$};

        \node[black!72] at (-0.58,4.95) {$\gamma_k$};
        \node[black!60] at (-0.42,4.56) {$\gamma_{k+1}$};

        \node[red!75!black] at (0.32,4.32) {$\tilde\gamma_k$};
        \node[red!75!black] at (1.08,4.38) {$\tilde\gamma_{k+1}$};

        \fill[black] (q) circle (1.3pt);
        \node[black!85, above=2pt] at (q) {$q$};

        \node[black!75, anchor=west, font=\small] at (1.45,4.95)
            {$C := J^+(\Sigma_{-1}) \cap J^-(q)$};
        \draw[black!40, line width=0.25pt] (1.48,4.98) -- (1.03,4.55);

        \node[anchor=west, black!85] at (4.14,4.16)
            {$\Sigma = \mathcal{T}^{-1}(0)$};

        \node[anchor=west, black!85] at (4.14,2.90)
            {$\Sigma_0^{\varepsilon_{k+1}} = \mathcal{T}_{\varepsilon_{k+1}}^{-1}(0)$};

        \node[anchor=west, black!85] at (4.14,2.00)
            {$\Sigma_0^{\varepsilon_k} = \mathcal{T}_{\varepsilon_k}^{-1}(0)$};

        \node[anchor=west, black!85] at (3.98,1.30)
            {$\Sigma_{-1} = \mathcal{T}^{-1}(-1)$};

    \end{tikzpicture}}
    \caption{Configuration used in the proof of Hawking's theorem.}
    \label{fig:hawking-proof-schematic}
\end{figure}

We now extend, for $k$ large, the $\gamma_k$ as $g_k$-geodesics into the past of $\Sigma$. By \eqref{eq:Cauchy-surface-positioning} and since $\Sigma^{\eps_k}_0$ is a Cauchy surface for $g_k$, $\gamma_k$ will intersect it in some $p_k\in \Sigma^{\eps_k}_0$. 
In addition, each $\gamma_k$ is contained in $C$ until it reaches $p_k$,
and the $g_k$-length of $\gamma_k$ up to $p_k$ is at least that of $\gamma_k$ until $\Sigma$, hence by \eqref{eq:gamma_k-lengths} and our indirect assumption \eqref{eq:indirect-assumption-estimate} exceeds $b_{\kappa-\delta,\beta}$ for $k$ large, i.e., $L_{g_k}(\gamma_k)>b_{\kappa-\delta,\beta}$.

Since $\Sigma^{\eps_k}_0$ is a Cauchy surface for the smooth, globally hyperbolic metric $g_k$, there exist $g_k$-maximizing $g_k$-geodesics $\tilde\gamma_k:[0,1]\to M$ from $\Sigma^{\eps_k}_0$ to $q$, 
emanating $g_k$-orthogonally from $\Sigma^{\eps_k}_0$ at some point $\tilde p_k\in \Sigma^{\eps_k}_0$. In particular,
\begin{equation}\label{eq:gamma_k-tilde-length}
\tau_{\Sigma_0^{\eps_k}}^{g_k}(q) = L_{g_k}(\tilde \gamma_k)\ge L_{g_k}(\gamma_k) > b_{\kappa-\delta,\beta}.
\end{equation}
Since $\T_{\eps_k}(\tilde p_k)=0$ and $|\T - \T_{\eps_k}|<1$ on $M$ by Lemma \ref{Lemma_T-eps_timefct} (ii), $\T(\tilde p_k)>-1$. Thus since $\T$ is a temporal function for $g$ and $\tilde \gamma_k$ is $g$-causal, it follows that $\tilde \gamma_k$ is entirely contained in $C$. In particular we may, without loss of generality, suppose that $\tilde p_k\to p$ for some $p\in C$ (indeed, $p\in \Sigma$).

It follows from this and the proof of \cite[Prop.\ 2.13]{Graf} that a (non-relabeled) subsequence of $\tilde \gamma_k$ converges in $C^1$. In particular, 
\begin{equation}\label{eq:rel-comp-set}
\bigcup_{k\in \NN}\dot{\tilde \gamma}_k([0,1])
\end{equation}
is relatively compact in $TM$. Next, we reparametrize the $\tilde \gamma_k$ to $g_k$-unit speed. Since
$\tilde \gamma_k$ converges in $C^1$, relative compactness of the set in \eqref{eq:rel-comp-set} is not affected by this operation. From this, assumption (i) together with Proposition \ref{Prop:Graf3.13-modified} imply that we may assume $k_0$ so large that, for all $k\ge k_0$ we have the Ricci-bound for $g_k$
\begin{equation}\label{eq:Ric-gk-bound}
\Ric[g_k](\dot{\tilde \gamma}_k(t),\dot{\tilde \gamma}_k(t)) > (n-1)(\kappa -\delta)  \qquad \forall t\in [0,L_{g_k}(\tilde  \gamma_k)].
\end{equation}

Finally we derive a matching mean curvature bound: 
From Lemma \ref{lem:smeared-mean-curv-smooth}, Theorem \ref{Prop:Mean-curvature-convergence} and \eqref{eq:mean-curvature-improved} we obtain that the $g_k$-mean curvature of $\Sigma^{\eps_k}_0$ is strictly less than $\beta$ on $\Sigma^{\eps_k}_0\cap C$, again for $k$ large. 

This, together with \eqref{eq:Ric-gk-bound} and \eqref{eq:gamma_k-tilde-length}
provides a contradiction to (the estimates on the length of maximizing curves derived in the proof of) Theorem \ref{Th:Hawking-smooth}, applied to the smooth spacetime $(\M,g_k)$ with smooth Cauchy surface $\Sigma_0^{\eps_k}$.
\end{proof}

\subsection{Hawking Theorem - version 2}

We now turn to the non-globally hyperbolic case of Hawking's singularity theorem. This will require some preparation.
\begin{Prop}\label{prop:compact-level-sets-regularisation} Let $\M$ be a smooth 
manifold and suppose that $f:\M \to \R$ is $C^0$  
with the property that the $0$-level set $Z:=f^{-1}(0)$ is compact. 
Let $U$ be any neighbourhood of $Z$. Then there exist data $(U_\alpha,\psi_\alpha)$, $(\xi_\alpha)_\alpha$, $(\chi_\alpha)_{\alpha}$ such that, setting $f_\eps := f\star_\M \rho_\eps$ as in \eqref{eq:M-convolution} we have: For each $\eps\in (0,1]$, $f_\eps^{-1}(0) \subseteq U$.  
In particular, if $U$ is relatively compact, then each $f_\eps^{-1}(0)$ is compact.
\end{Prop}
\begin{proof} 
Without loss of generality, we may suppose that $U$ is of the form $U=U_\eta:=\{p\in \M \mid d_h(p,Z) \le \eta \}$ for some $\eta>0$. Choosing first a finite cover of the compact set $U_{\eta/2}$ by relatively compact chart neighbourhoods and then covering $\M\setminus U_{\eta/2}$ by countably many compact chart neighbourhoods and choosing a partition of unity subordinate to this cover, we can arrange the data in \eqref{eq:M-convolution} to have the following properties:
\begin{itemize}
\item[(i)] There is a finite index set $A_0\subseteq A$ with $U_\alpha\subseteq U$ for all $\alpha\in A_0$ and $\bigcup_{\alpha\in A_0} U_\alpha$ covering $U_{\eta/2}$. 
\item[(ii)] For every $\alpha\in A\setminus A_0$, we have $U_\alpha \subseteq \M\setminus U_{\eta/2}(Z)$.
\end{itemize}
Moreover, we can assume each set $\mathrm{supp}\,\xi_\alpha$ to be connected. 

Fix $\eps\in(0,1]$ and let $p\in \M\setminus U$. Then for every $\alpha\in A_0$ we have $\chi_\alpha(p)=0$, so 
\begin{equation}\label{eq:f-eps}
    f_\eps(p)=\sum_{\alpha\in A\setminus A_0} \chi_\alpha(p)\, (\psi_\alpha)^*\!\left(\big((\psi_\alpha)_*(\xi_\alpha f)\big)*\rho_\eps\right)\!(p).
\end{equation}
For any $\alpha\in A\setminus A_0$, the function $(\psi_\alpha)_*(\xi_\alpha f)$ has constant sign on its support because $\mathrm{supp}\,\xi_\alpha$ lies in a connected component of $\M\setminus Z = \{f\ne 0\}$, and this sign equals that of $f$ on any connected neighbourhood of $p$. Now because $\rho_\varepsilon\ge 0$, $\int\rho_\varepsilon=1$, and $\mathrm{supp}\,\rho_\varepsilon\subseteq \overline{B_\varepsilon(0)}$, we can write
\[
\big((\psi_\alpha)_*(\xi_\alpha f)\big)*\rho_\eps(x)
=\int_{B_\varepsilon(0)} (\psi_\alpha)_*(\xi_\alpha f)(x-y)\,\rho_\eps(y)\,dy.
\]
Thus $(\psi_\alpha)^*\!\left(\big((\psi_\alpha)_*(\xi_\alpha f)\big)*\rho_\eps\right)$ has the same sign as $f$ near $p$ (or is
identically zero, hence does not contribute to $f_\eps(p)$). It follows that every non-vanishing summand in  \eqref{eq:f-eps} has the same sign
and since at least one $\xi_\alpha f$ has to be non-zero near $p$, $f_\eps(p)$ is non-zero. Altogether, we conclude that
$f_\eps^{-1}(0) \subseteq U$.
\end{proof}

\begin{Remark}[Mean curvature bounds for achronal hypersurfaces]\label{rem:achronal-acausal-mean-curvature}\ \\
(i) Let $\Sigma$ be a smooth, connected and closed spacelike hypersurface in a $W^{1,p}_\text{loc}$-spacetime. Then the proof of \cite[Prop.\ 14.48]{ON83} carries over
to this low-regularity setup, showing that there exists a $W^{1,p}_\text{loc}$-Lorentzian covering $\mathbf{k}: \widetilde \M \to \M$ and a closed spacelike hypersurface $\widetilde \Sigma$ in $\widetilde \M$ that is achronal and isometric under $\mathbf{k}$ to $\Sigma$. Therefore, and also since one can always restrict attention to one connected component of $\Sigma$, in the assumption of the non-globally hyperbolic version of the Hawking theorem below, one may without loss of generality assume the compact spacelike hypersurface to be achronal. 

(ii) Supposing that $g\in W^{1,p}_\text{loc}$ and $\Riem[g]\in L^p$, we now show that achronality of  $\Sigma$ already implies acausality. As this is a question of causality, we may equivalently consider the problem after applying RT-regularisation. In this atlas, both $\Sigma$ and $g$ are of regularity $W^{2,p}_\text{loc} \subseteq C^{1,\alpha}_\text{loc}$. Suppose that there exists a future directed causal curve $\gamma: [0,1] \to \M$ with $p:=\gamma(0), q:=\gamma(1)\in \Sigma$. Since $g$ is causally plain, by \cite[Thm.\ 2.12]{GrantKuSaSt}, it has the push-up property 
(see also Section \ref{Sec_causality}.)
Therefore, $\gamma$ has to be a null curve (i.e., $\dot{\gamma}(t)$ is null a.e.), since otherwise there would exist a timelike curve with the same endpoints, contradicting achronality of $\Sigma$. 

Moreover, $\gamma$ has to be maximizing from $p$ to $q$. Otherwise, there would exist some causal curve $\hat{\gamma}$ from $p$ to $q$ of positive length, which again by push-up would lead to a contradiction with achronality of $\Sigma$. By \cite[Thm.\ 1.1]{LLS:21}, $\gamma$ can be parametrized as a $C^{1,1}_\text{loc}$-curve. Then \cite[Lem.\ 5.11]{CGHKS25} shows that $\gamma$ must emanate from $\Sigma$ orthogonally, which implies that $\dot{\gamma}(0)$ is timelike, a contradiction. More precisely, we note that the conclusion of \cite[Lem.\ 5.11]{CGHKS25} remains valid even though our $\Sigma$ here is
merely $C^{1,\alpha}_\text{loc}$: the only change required in the proof of \cite[Lem.\ 5.11]{CGHKS25} is to note that $\alpha(s) -sy = O(s^{1+\alpha})$ (instead of $O(s^{2}$)), which still allows one to draw the same conclusion since $\alpha>0$.

(iii) Based on the above observations, we are now going to establish a natural way of imposing mean curvature bounds on an achronal spacelike hypersurface $\Sigma$ in a $W^{1,p}_\text{loc}$-spacetime with $\Riem[g]\in L^p_\text{loc}$.
Namely, since $\Sigma$ is acausal by (ii), \cite[Thm.\ 2.43]{Minguzzi-cone-structures} implies that the Cauchy development $D(\Sigma)$ is open in $\M$ and globally hyperbolic. Thus, replacing $\M$ by $D(\Sigma)$ we are
precisely in the setting of Definition \ref{Def:mean-curvature-bound}: $\Sigma$ is realized as the $0$-level set of a smooth temporal function $\T: D(\Sigma)\to \R$ and we say that the mean curvature of $\Sigma$ is bounded from above/below if the condition from Definition \ref{Def:mean-curvature-bound} is satisfied for such a $\T$.
\end{Remark}

In the theorem below we shall require the following notation, due to \cite{Graf}: Given a $C^1$-spacetime $(\M,g)$, for a compact $K\Subset T\M$, $N>0$ we denote the union of the images of all $g$-geodesics $\dot{\gamma}$ (restricted to $[0,N]$ if they exist longer) with $\dot{\gamma}(0)\in K$ by $F_{K,N}$, i.e.,
\begin{equation}\label{eq:F}
 F_{K,N}:=\bigcup_{\{ \dot{\gamma}:\; \mathrm{g-geodesic \; with}\; \dot{\gamma}(0)\in K \}}\mathrm{im}(\dot{\gamma}|_{[0,N]})\subseteq T\M.
\end{equation}
We define $F_{\eps,K,N}$ the same way, only using $\check g_\eps$-geodesics instead of $g$-geodesics.

\begin{Thm}[$W^{1,p}$-Hawking singularity theorem, II]\label{Th:Hawking-II} Let $(\M,g)$ be an $n$-dimen\-sional spacetime with $g\in W^{1,p}_\text{loc}$ and $\Riem[g]\in L^p_\text{loc}$ for some $p>2n$. Suppose that there exist $\beta, \kappa\in \R$ such that
\begin{itemize}
\item[(i)] $\mathrm{Ric}(X, X) \ge (n-1)\kappa$ for all smooth unit timelike vector fields $X$ on $\M$.
\item[(ii)]  There exists a compact achronal smooth spacelike hypersurface $\Sigma$ in $\M$ with mean curvature bounded above by $\beta$ in the sense of Remark \ref{rem:achronal-acausal-mean-curvature} (iii). 
\item[(iii)] If $\kappa=0$, additionally suppose that $\beta<0$. If $\kappa< 0$, additionally suppose that
$\frac{\beta}{(n-1)\sqrt{-\kappa}}<-1$.
\end{itemize}
Then $\M$ is timelike RT-geodesically incomplete.
\end{Thm}
For the notion of RT-completeness see Definition \ref{def:RT-completeness}. Also observe that the Theorem does not hold for arbitrary $\kappa$ and $\beta$ but only for such values that guarantee a finite $b_{\kappa,\beta}$.
\begin{proof} Our assumption on the mean curvature of $\Sigma$ means that for some 
representation of $\Sigma$ as the zero-level set of a Cauchy temporal function $\T$ on $D(\Sigma)$ and some open neighbourhood $U$ of $\Sigma$ in $D(\Sigma)$,
$\esssup \mathcal{H}_{g,\T,U}<\beta$. As in the proof of Theorem \ref{Th:Hawking-I}, we apply the RT-regularisation procedure to this setup, arriving at 
a spacetime $(\M,g)$ with $g\in W^{2,p}_\text{loc}$, $\Sigma\in W^{2,p}_\text{loc}$, and $\T\in W^{2,p}_\text{loc}$, while $D(\Sigma)$ remains unchanged. We may assume that $\bar{U} \Subset D(\Sigma)$. 

Employing Proposition \ref{prop:compact-level-sets-regularisation} to $\Sigma$ as a submanifold of $D(\Sigma)$, and applying \cite[Lem.\ 2.4]{KSSV:14} to $\T\star_{D(\Sigma)} \rho_\eps:D(\Sigma) \to \R$ as in the proof of Lemma \ref{Lemma_T-eps_timefct},
we obtain regularisations $\T_\eps$
of $\T$ with properties (ii), (iii) and (iv) of Lemma \ref{Lemma_T-eps_timefct} on $D(\Sigma)$. Additionally, $\T_\eps^{-1}(0)\subseteq U$ for each $\eps$. In particular, each level set $\Sigma^\eps_0 = \T_\eps^{-1}(0)$ is compact. 

To construct global regularisations of $g$ on $\mathcal{M}$ we follow the same procedure as in Theorem \ref{Th:Hawking-I}, only choosing the data $(U_\alpha,\psi_\alpha)$, $(\xi_\alpha)_\alpha$, $(\chi_\alpha)_{\alpha}$ for the definition of $\star_\M$
in such a way that $\star_\M$ and $\star_{D(\Sigma)}$ coincide on $U$. 
This can be achieved as in the proof of Proposition \ref{prop:compact-level-sets-regularisation}, by simply adding further open sets $U_\alpha$ to the covering of $D(\Sigma)$ so as to cover all of $\M$ while assuring that all the additional $U_\alpha$ are disjoint from $\bar{U}$.

Then there exists $\eps_0>0$ such that $\Sigma^\eps_0$ is a $\check g_\eps$-spacelike hypersurface in $D(\Sigma)$ for each $\eps\in (0,\eps_0]$.
By Theorem \ref{Prop:Mean-curvature-convergence}, we can additionally assume $\eps_0$ so small that 
$\esssup\mathcal{H}_{\check g_\eps,\T_\eps,U} < \beta$ on $\bar U$ for such $\eps$. Defining the corresponding mean curvature vector as
\[
\vec{H}_\eps := \frac{1}{n-1} \sum_{k=1}^{n-1} \mathrm{II}_\eps(e_i,e_i)
\]
(cf.\ \eqref{eq:mean-curvature-vector-def}), this means that 
the $\check{g}_\eps$-mean curvature of $\Sigma_0^\eps$ satisfies
\begin{equation*}\label{eq:mean-curv-in-hawking2}
  H_{\Sigma_0^\eps} =  -(n-1)\max\check{g}_\eps(\vec{H}_\eps,N_\eps)<\beta
\end{equation*}
for $\eps\in (0,\eps_0)$, where $N_\eps$ was defined in \eqref{eq:N_eps-def}.

Let us now indirectly assume that the chosen RT-regularisation of $(\M,g)$ is timelike geodesically complete. 
With $N_\eps$ as above and $N$ on $D(\Sigma)$ defined by \eqref{eq:N-def},  set 
\begin{align*}
C:= \bigcup_{0<\eps\le 1} \{N_\eps(p) \mid p\in \Sigma^\eps_0\}    
\cup \{N(p) \mid p\in \Sigma\}.
\end{align*}
Since $N_\eps \to N$ uniformly on $\bar{U}$, $C$ is relatively compact in $T\M$, and we set $K:=\bar{C}$. 

By our assumptions on $\beta$ and $\kappa$, Proposition \ref{Lem:b-kappa-beta-continuous} allows us to choose $\delta>0$ so small
that $b_0:=b_{\kappa-\delta,\beta}<\infty$: This is immediate in case $\kappa<0$. On the other hand, if $\kappa=0$ and $\beta<0$ one simply chooses
$\delta>0$ such that $\frac{\beta}{(n-1)\sqrt{\delta}}<-1$.
By \cite[Prop.\ 2.11]{Graf}, $F_{\le\eps_{0},K,b_0}\subseteq T\M$ (cf.\ \eqref{eq:F}) is relatively compact in $T\M$, and we set 
$\tilde{K}:= \overline{F_{\le\eps_{0},K,b_0}}\subseteq T\M $.

According to Proposition \ref{Prop:Graf3.13-modified} there exists $\eps_{1}\in (0,\eps_0)$ (depending only on $\tilde{K}$ and $\delta$ 
such that
\[ 
\forall\eps<\eps_{1}\ \forall X\in\tilde{K} \text{ with } \check{g}_{\eps}(X, X) = -1:\  \text{Ric}_{\eps}(X, X) \ge (n-1)(\kappa-\delta) \text{.} 
\]
Fix $\eps<\eps_1$. Then $\dot{\gamma}\in F_{\le\eps_{0},K,b_0} \subseteq\tilde{K}$ for all $\check{g}_{\eps}$-unit speed $\check{g}_{\eps}$-geodesics 
$\gamma: [0,b]\rightarrow \M$ starting $\check{g}_{\eps}$-orthogonally to $\Sigma_\eps$, so for any such geodesic $\text{Ric}_{\eps}(\dot{\gamma},\dot{\gamma})\ge (n-1)(\kappa-\delta)$ and 
$-(n-1)\check{g}_{\eps}(\vec{H}_{\eps},\dot{\gamma}(0))=-(n-1)\check{g}_{\eps}(\vec{H}_{\eps},N_{\eps})<\beta$. 

It therefore follows from (the proof of) Theorem \ref{Th:Hawking-smooth} that
any future directed $\check{g}_{\eps}$-unit speed $\check{g}_{\eps}$-geodesic starting $\check{g}_{\eps}$-orthogonally to $\Sigma^\eps_0$ will stop maximizing the $\check{g}_{\eps}$-distance to $\Sigma^\eps_0$ after $b_{0}=b_{\kappa-\delta,\beta}<\infty$.

This shows that $D_{\eps}^{+}(\Sigma^\eps_0)$ is contained in the compact set $F_{\eps,K,b_{0}}$, which implies that the future Cauchy horizon $H_{\eps}^{+}(\Sigma^\eps_0)$ is compact and non-empty:  Indeed, by \cite[Prop.\ 2.9]{Graf}
we have $\emptyset\ne F_{\eps,K,b_0+1}\backslash F_{\eps,K,b_{0}}\subseteq I_{\eps}^{+}(\Sigma^\eps_0) \backslash D_{\eps}^{+}(\Sigma^\eps_0))$. We then arrive at a contradiction by the standard argument (cf., e.g, [35, Theorem 14.55B]). 
\end{proof}

\section{A low-regularity Myers' theorem}\label{sec:myers}

To close the paper we provide a $W^{1,p}$-version of Myers' theorem, the Riemannian analogue of Hawking's theorem, see e.g.\ \cite{AS:15}. Here the application of the RT-regularisation is straight forward since there is no mean curvature assumption involved. We also discuss an alternative synthetic version of the theorem.

For a continuous Riemannian metric $g$, we denote by $L_g(\gamma)$ the $g$-length of an absolutely continuous curve $\gamma$, and by $d_g$ the associated intrinsic distance. Since in our setting $g$ lies in $W^{1,p}_\text{loc}$ with $p>n$, hence in $C^{0,\alpha}_\text{loc}$ with $\alpha=1-n/p>0$,  the metric space $(\mathcal{M},d_g)$ is well defined. We shall say that $g$ is \emph{complete} if the metric space $(\mathcal{M},d_g)$ is complete.  In this Riemannian setting we use RT-regularization to reduce to the $C^{1,\alpha}$-case and then appeal to the low-regularity Myers' theorem from \cite[Appendix~A]{Graf}.

\begin{Thm}[$W^{1,p}$-Myers theorem]\label{Thm:Myers-low-regularity}
Let $(\mathcal{M},g)$ be a Riemannian manifold with $g\in W^{1,p}_\text{loc}$ and
$\Riem[g]\in L^p_\text{loc}$ for some $p>\mathrm{max}(4,n)$, as in Theorem~\ref{Thm_smoothing_preliminaries}.
Assume that $g$ is (metrically) complete and that there exists a constant
$\lambda>0$ such that
\begin{equation}\label{eq:myers1}
    \Ric(X,X)\ge (n-1)\lambda\, g(X,X)
\end{equation}
pointwise almost everywhere on $\mathcal{M}$ for every smooth local vector field $X$.
Then
\[
\diam(\mathcal{M},d_g)\le \frac{\pi}{\sqrt{\lambda}}.
\]
\end{Thm}

\begin{proof}
By Theorem~\ref{Thm_smoothing_preliminaries}, after passing to an RT-regularised atlas, $g \in W^{2,p}_\text{loc}\subseteq C^{1,\alpha}_\text{loc}$.
This does not alter the induced distance $d_g$, in particular, completeness of $(\mathcal{M},d_g)$ is unchanged.

By the discussion in Section~\ref{sec:ricbounds}, the almost-everywhere Ricci lower bound
\eqref{eq:myers1} is equivalent to the corresponding distributional Ricci lower bound.
Therefore the assumptions of \cite[Theorem~A.1]{Graf} are satisfied for the complete
$C^{1,\alpha}_\text{loc}$-metric $g$. Applying that theorem, we conclude that
\[
\diam(\mathcal{M},d_g)\le \frac{\pi}{\sqrt{\lambda}}.
\]
\end{proof}

\begin{Remark}[Alternative proof via synthetic geometry]
We point out that Theorem \ref{Thm:Myers-low-regularity} can alternatively be proved using the synthetic approach to lower Ricci curvature bounds. In fact, this approach reveals that the assumption $\Riem[g]\in L^p_\text{loc}$ can be dropped in the case of a Riemannian metric: the $W^{1,p}_\text{loc}$-regularity of the metric and the lower Ricci bound alone are sufficient.

To see this, recall that for $p>n$, $g\in W^{1,p}_\text{loc} \subseteq C^0$. Thus the triple $(\mathcal{M}, d_g, \vol_g)$ forms a complete and separable metric measure space. As discussed in Section \ref{sec:ricbounds}, the pointwise almost everywhere lower bound \eqref{eq:myers1} implies the corresponding distributional lower bound $\Ric \ge (n-1)\lambda g$.

By a recent result of Mondino and Ryborz \cite[Prop.\ 7.5]{MR:25}, for $g\in  W^{1,p}_\text{loc}$ with $p>n$, this distributional lower bound implies the required exponential volume growth condition (cf.\ \cite[Eq.\ (4.1)]{MR:25}). Consequently, by \cite[Cor.\ 7.7]{MR:25} 
this distributional bound is equivalent to the space $(\mathcal{M}, d_g, \vol_g)$ satisfying the reduced curvature-dimension condition $\mathrm{CD}^*((n-1)\lambda, n)$. 

Since metric measure spaces induced by continuous Riemannian metrics are infinitesimally Hilbertian \cite[Cor.\ 4.27]{MR:25}, this condition is equivalent to the Riemannian curvature-dimension condition $\mathrm{RCD}^*((n-1)\lambda, n)$ \cite[Rem.\ 7.8]{MR:25}. Furthermore, as noted in \cite[Remark 4.9]{MR:25}, for $\sigma$-finite measures the $\mathrm{RCD}^*((n-1)\lambda, n)$ condition coincides with $\mathrm{RCD}((n-1)\lambda, n)$, which in turn implies the standard Sturm--Lott--Villani condition $\mathrm{CD}((n-1)\lambda, n)$.

One can then directly apply the generalized Bonnet-Myers theorem for metric measure spaces due to Sturm \cite[Cor.\ 2.6]{Sturm2006II}. For $K=(n-1)\lambda>0$ and $N=n$, it immediately yields that the diameter of $\mathcal{M}$ is bounded by 
\[
\diam(\mathcal{M}, d_g) \le \pi\sqrt{\frac{N-1}{K}} =  \frac{\pi}{\sqrt{\lambda}}.
\]
\end{Remark}

\appendix

\section{The RT-equations} \label{Sec_RT_appendix}

Following the summary laid out in \cite{Reintjes_SCC, ReintjesTemple_essreg}, we now review the theory of the {\it Regularity Transformation (RT-) equations}, developed by B. Temple and M. Reintjes in a series of papers \cite{ReintjesTemple_ell1, ReintjesTemple_ell2, ReintjesTemple_ell3, ReintjesTemple_ell4, ReintjesTemple_ell5, ReintjesTemple_ell6, ReintjesTemple_essreg}. Solutions of the RT-equations furnish local coordinate transformations which regularise connections by one derivative, to one derivative of regularity above their Riemann curvature in the starting coordinate system \cite{ReintjesTemple_ell2, ReintjesTemple_ell4}.  In \cite{ReintjesTemple_essreg}, based on the RT-equations, a necessary and sufficient condition is given for when a singularity in an affine connection is removable by a coordinate transformation, together with a computable procedure for removing the singularity by regularising the connection, potentially by multiple derivatives, all the way up to its {\it essential} (highest possible) {\it regularity}, a geometric notion of regularity introduced in \cite{ReintjesTemple_essreg}. This is accomplished both locally and globally.    

To state the theorems in \cite{ReintjesTemple_ell2,ReintjesTemple_ell4} on local regularisations of connections which we require here, consider a fixed chart $(x,\Omega)$ on an $n$-dimensional manifold $\mathcal{M}$, where $\Omega_x \equiv x(\Omega) \subseteq \R^n$ (the image of $\Omega$ under the coordinate map)  is open and bounded with smooth boundary. Let $\Gamma_x$ denote the collection of components $\Gamma^k_{ij}(x)$ of an affine connection $\Gamma$ in $x$-coordinates.  Now, view $\Gamma_x$ as a matrix valued $1$-form in $x$-coordinates, $(\Gamma_x)^\mu_{\nu} \equiv (\Gamma_x)^\mu_{\nu j} dx^j$, using the Einstein summation convention of summing over repeated indices from $1$ to $n$. Let $d\Gamma_x$ denote its exterior derivative, $d(\Gamma_x)^\mu_\nu \equiv \partial_i (\Gamma_x)^\mu_{\nu j} dx^i \wedge dx^j$, where $\mu, \nu = 1,...,n$ denote indices of the matrix. Writing the Riemann curvature tensor as a matrix valued $2$-form, ${\rm Riem}(\Gamma_x) = d\Gamma_x +\Gamma_x \wedge \Gamma_x$, it follows that the assumption $\Gamma_x \in L^{p}_\text{loc}(\Omega_x)$ and ${\rm Riem}(\Gamma_x) \in L^{p/2}_\text{loc}(\Omega_x)$ is equivalent to the assumption $\Gamma_x \in L^{p}_\text{loc}(\Omega_x)$ and $d\Gamma_x \in L^{p/2}_\text{loc}(\Omega_x)$; the latter assumption is used here and in \cite{ReintjesTemple_ell2,ReintjesTemple_ell4}. The main idea for establishing optimal regularity in \cite{ReintjesTemple_ell4} is to derive from the connection transformation law a non-invariant system of {\it elliptic} PDE's for the regularising Jacobian $J$, an idea motivated by the Riemann-flat condition in \cite{ReintjesTemple_geo}. This idea led to the formulation of the {\it RT-equations} in \cite{ReintjesTemple_ell1}.   We now review the main steps in the derivation of the RT-equations, and explain how they furnish optimal regularity.

\subsection{Derivation of the RT-equations}
To review the derivation in \cite{ReintjesTemple_ell1}, assume there exists a coordinate transformation with Jacobian $J$ mapping $\Gamma_x$ to $\Gamma_y$ (the connection of optimal regularity), and write the connection transformation law as
\beq \label{opt_eqn1}
\Gammati = \Gamma - J^{-1} dJ,
\eeq
where $\Gammati^k_{ij} = (J^{-1})^k_\gamma J^\alpha_i J^\beta_j (\Gamma_y)^\gamma_{\alpha\beta}$ is the connection $\Gamma_y$ transformed as a tensor to $x$-coordinates and for ease in notation $\Gamma \equiv \Gamma_x$. Differentiating \eqref{opt_eqn1} by the exterior derivative $d$ and by the co-derivative $\delta$, implies the following two equations
\begin{eqnarray} 
\Delta \Gammati &=& \delta d\Gamma - \delta\big(dJ^{-1} \wedge dJ\big) + d\delta \Gammati ,     \label{opt_eqn2} \\
\Delta J &=& \delta(J\Gamma) - \langle dJ ; \Gammati \rangle - J \delta\Gammati ,  \label{opt_eqn3}
\end{eqnarray}
where $\Delta \equiv \delta d + d \delta = \partial_{x^0}^2 + ... + \partial_{x^n}^2$ is the Euclidean Laplacian, $\langle \cdot\; ; \cdot \rangle$ is a matrix-valued inner product and $\wedge$ the wedge product on matrix valued differential forms, (see \cite[Ch.3]{ReintjesTemple_ell1} or \cite[Ch.5]{ReintjesTemple_ell4} for detailed definitions). At this stage, equations \eqref{opt_eqn2} - \eqref{opt_eqn3} neither appear solvable, nor need a solution $J$ be a true Jacobian that is integrable to coordinates, i.e., satisfying ${\rm Curl}(J) =0$. To complete the equations, view $A\equiv \delta\Gammati$ as a free matrix valued function---this choice was motivated by the Riemann-flat condition in \cite{ReintjesTemple_geo}, which only involves $d\Gammati$, but not $\delta\Gammati$.     Now, substituting $A\equiv \delta \Gammati$ in \eqref{opt_eqn2} - \eqref{opt_eqn3}, and viewing $A$ as a new unknown matrix valued function, one next imposes on equation \eqref{opt_eqn3} the condition ${\rm Curl}(J) =0$ for integrability.  For this, the vectorization $\vec{J}^\mu = J^\mu_\nu dx^\nu$ of $J$ is introduced in \cite{ReintjesTemple_ell1}, so that ${\rm Curl}(J)\equiv d\vec{J}$, and the integrability condition is imposed in the equivalently form $d\vec{J} =0$.   Finally, by a fortuitous cancellation the regularities in different terms of the equation become consistent, and the computations in \cite{ReintjesTemple_ell1} eventually lead to the {\it RT-equations}:
\begin{align} 
\Delta \Gammati &= \delta d\Gamma - \delta \big(d J^{-1} \wedge dJ \big) + d(J^{-1} A ), \label{PDE1} \\
\Delta J &= \delta ( J \Gamma ) - \langle d J ; \tilde{\Gamma}\rangle - A , \label{PDE2} \\
d \vec{A} &= \overrightarrow{\text{div}} \big(dJ \wedge \Gamma\big) + \overrightarrow{\text{div}} \big( J\, d\Gamma\big) - d\big(\overrightarrow{\langle d J ; \tilde{\Gamma}\rangle }\big),   \label{PDE3}\\
\delta \vec{A} &= v.  \label{PDE4}
\end{align}

Equation \eqref{PDE3} on the auxiliary field $A$ results from imposing $d\vec{J}=0$ on \eqref{PDE2}, and one can prove that integrability of $J$ follows from the coupled equations \eqref{PDE2} and \eqref{PDE3}. The unknowns $(\Gammati,J,A)$ in the RT-equations, together with the given non-optimal connection components $\Gamma$, are viewed as matrix valued differential forms. Arrows denote ``vectorization'', mapping matrix valued $0$-forms to vector valued $1$-forms, (e.g. $\vec{A}^\mu = A^\mu_i dx^i$) and $\overrightarrow{\text{div}}$ is a divergence operation which maps matrix valued $k$-forms to vector valued $k$-forms. The vector $v$ in \eqref{PDE4} is free to impose, representing a ``gauge''-type freedom in the equations, reflecting the fact that smooth transformations preserve optimal connection regularity. The operations on the right hand side are formulated in terms of the Cartan Algebra of matrix valued differential forms based on the Euclidean metric in $x$-coordinates, and these objects depend on the starting coordinate system and are not invariant under change of coordinates, (see \cite{ReintjesTemple_ell1} for detailed definitions and proofs).  As shown in \cite{ReintjesTemple_essreg}, a regularity below essential regularity is not a geometric property of a connection, so it makes sense that non-invariant equations are required to regularise them; the RT-equations \eqref{PDE1} - \eqref{PDE4} are this non-invariant system of PDE's, and they are elliptic regardless of the metric signature.

\subsection{How the RT-equations yield optimal regularity}

 We now review how solutions of the RT-equations locally furnish coordinate transformations to optimal regularity. Theorem 2.5 in \cite{ReintjesTemple_ell2} establishes optimal regularity $\Gamma_y \in W^{m+1,p}_\text{loc}$ for curvature $\Riem(\Gamma_x)$ in $W^{m,p}_\text{loc}$ at the higher level $m\geq 1$, $p>n$. The case of lower regularities $m=0$, regularising a non-optimal connection $\Gamma_x \in L^{p}_\text{loc}$ with $\Riem(\Gamma_x)$ in $L^{p/2}_\text{loc}$, $\max\{4,n\} < p < \infty$, is accomplished in Theorem 3.1 in \cite{ReintjesTemple_ell4}. This lower regularity case is significantly more challenging because the earlier iteration scheme for solving the RT-equations at higher regularity in \cite{ReintjesTemple_ell2} does not close when $\Gamma$ is only in $L^p$, due to the non-linear term $d J^{-1} \wedge dJ$ in \eqref{PDE1}. This problem was eventually resolved  in \cite{ReintjesTemple_ell4} by the discovery of an internal ``gauge''-type transformation for solutions of the RT-equations, which allows one to separate off the troublesome equation \eqref{PDE1} from the remaining equations. 

The resulting separated-off system is linear in the unknowns $(J,B)$. This system is referred to as the {\it reduced} RT-equations in \cite{ReintjesTemple_ell4}. It takes the form: 
\begin{eqnarray} 
\Delta J &=& \delta ( J \mm \Gamma ) - B , \label{RT_withB_2} \\
d \vec{B} &=& \overrightarrow{\text{div}} \big(dJ \wedge \Gamma\big) + \overrightarrow{\text{div}} \big( J\, d\Gamma\big) ,   \label{RT_withB_3} \\
\delta \vec{B} &=& v'.  \label{RT_withB_4}
\end{eqnarray}
The iteration scheme in \cite{ReintjesTemple_ell4}, which is based on solving the linear Poisson equation at each stage, applies to the reduced RT-equations \eqref{RT_withB_2} - \eqref{RT_withB_4} at the low regularity of $L^p$-connections with $d\Gamma \in L^{p/2}_\text{loc}$, and locally establishes existence of solutions $(J,B)$, (i.e, in neighbourhoods $\Omega'$ of points), such that $J$ is point-wise an invertible matrix, cf. \cite[Thm 6.4]{ReintjesTemple_ell4}. 

That any solution $J$ is a Jacobian integrable to coordinates is a built-in property of \eqref{RT_withB_2} - \eqref{RT_withB_4}, provided that the integrability condition $d\vec{J}\equiv {\rm Curl}(J)=0$ holds on the boundary $\partial\Omega'$, c.f. \cite[Thm 6.4]{ReintjesTemple_ell4}. That is, combining \eqref{RT_withB_2} with \eqref{RT_withB_3}, a computation shows that $\omega\equiv d\vec{J}$ is a solution of Laplace's equation $\Delta \omega=0$, which together with the boundary data $d\vec{J}|_{\partial\Omega'}=0$ implies that $\omega=0$ throughout $\Omega'$. This implies that $J$ is integrable to coordinates. 

Given now a solution $(J,B)$ of the reduced RT-equation \eqref{RT_withB_2} - \eqref{RT_withB_4} with $J$ an integrable and invertable Jacobian, one recovers a solution $(J,\Gammati,A)$ of the full RT-equations \eqref{PDE1} - \eqref{PDE4} by introducing\footnote{The second and third equation in \eqref{Gammati'} define the ``gauge'' transformations of the RT-equations, while the first equation defines a projection onto the space of solution of the Riemann-flat condition.}
\beq  \label{Gammati'} 
\Gammati \equiv \Gamma - J^{-1} dJ, 
\hspace{.5cm} 
A \equiv B - \langle d J ; \Gammati \rangle ,  
\hspace{.5cm} \text{and} \hspace{.5cm}
v \equiv v' - \delta \overrightarrow{\langle d J ; \Gammati \rangle},
\eeq
as can be verified by direct computation using \eqref{RT_withB_2} to eliminate uncontrolled terms involving $\delta\Gamma$. Interior elliptic estimates, applied to the first RT-equations \eqref{PDE1}, then imply that $\Gammati$ is in $W^{1,p}$ on any open set $\Omega'_c$ compactly contained in $\Omega'$, a gain of one derivative over $\Gamma\equiv \Gamma_x$. Finally, defining 
\beq \label{Gamma_y_reverse}
(\Gamma_y)^\gamma_{\alpha\beta} \equiv J_k^\gamma (J^{-1})^i_\alpha  (J^{-1})^j_\beta   \; \Gammati^k_{ij},
\eeq
substitution into the first equation in \eqref{Gammati'} yields the connection transformation law \eqref{opt_eqn1}, which implies that $\Gamma_y$ is the transformed connection of optimal regularity, $\Gamma_y \in W^{1,p/2}_\text{loc}(\Omega'_c)$.

These are the key ideas underlying the proof of the optimal regularity result Theorem 3.1. in \cite{ReintjesTemple_ell4} for $L^p$ connections with Riemann curvature in $L^{p/2}$. The precise statement of Theorem 3.1. in \cite{ReintjesTemple_ell4} is recorded as Theorem \ref{Thm_Smoothing_low} below. We state a simplified version of Theorem 3.1 in \cite{ReintjesTemple_ell4} without uniform bounds on the regularised connections; these bounds form the basis of the extension of Uhlenbeck compactness to Lorentzian and non-Riemannian geometry in \cite{ReintjesTemple_ell4, ReintjesTemple_ell5, ReintjesTemple_ell6}, but are not required for our purposes here.

\begin{Thm} \label{Thm_Smoothing_low}  
Assume $\Gamma_x \in L^{p}_\text{loc}(\Omega_x)$ and $\Riem(\Gamma_x) \in L^{p/2}_\text{loc}(\Omega_x)$ in $x$-coor\-di\-nates, for some $p>{\rm max}\{4 ,n\}$, $p<\infty$, $n\geq2$. Then for any point $P\in \Omega$ there exists a neighbourhood $\Omega' \subseteq \Omega$ of $P$ (depending on $\Omega_x, n, p$ and $\Gamma$) and a coordinate transformation $x \to y$ with Jacobian $J=\frac{\partial y}{\partial x}\, \in W^{1,p}(\Omega'_x)$, such that $\Gamma_y \in W^{1,p/2}(\Omega'_y)$. 
\end{Thm}

Theorem \ref{Thm_Smoothing_low} implies that the connection regularity can always be lifted from $\Gamma_x\in L^p$ to $\Gamma_y\in W^{1,p/2}$---one derivative more regular than the starting curvature $\Riem(\Gamma_x)\in L^{p/2}$---in neighbourhoods (which can be taken to be sufficiently small balls) of points. For the purposes of this paper, it is more convenient to use the following corollary of Theorem \ref{Thm_Smoothing_low}, proven in \cite{ReintjesTemple_essreg} by consecutive use of Theorem \ref{Thm_Smoothing_low}, which assumes the same value of $p>n$ for the curvature and the connection. That is, $\Gamma_x\in L^p_\text{loc}$ is regularised to $\Gamma_x\in W^{1,p}_\text{loc}$ under the stronger assumption that $\Riem(\Gamma_x)\in L^p_\text{loc}$, in place of $\Riem(\Gamma_x)\in L^{p/2}_\text{loc}$, assuming again that $p>\max\{4,n\}$.

\begin{Corollary} \label{Thm_Smoothing_low-cor}  
Assume $\Gamma_x \in L^{p}_\text{loc}(\Omega_x)$ and $\Riem(\Gamma_x) \in L^{p}_\text{loc}(\Omega_x)$ in $x$-coor\-di\-nates, for some $p>{\rm max}\{4,n\}$, $p<\infty$, $n\geq2$. Then for any point $P\in \Omega$ there exists a neighbourhood $\Omega' \subseteq \Omega$ of $P$ (depending on $\Omega_x, n, p$ and $\Gamma$) and a coordinate transformation $x \to y$ with Jacobian $J, J^{-1} \in W^{1,p}(\Omega'_x)$, such that $\Gamma_y \in W^{1,p}(\Omega'_y)$.   
\end{Corollary}

Finally, for this paper, we need the following extension of Corollary \ref{Thm_Smoothing_low-cor} from a local to a global setting, taken from Theorem 2.3 in \cite{ReintjesTemple_essreg}, and stated as Theorem \ref{Thm_smoothing_global} below. To introduce the global setting in \cite{ReintjesTemple_essreg} which we require in this paper, consider a manifold $\mathcal{M}$ with atlas $\mathcal{A}$, and assume $\Gamma$ is an affine connection $\Gamma$ defined on $(\M,\A)$. A connection $\Gamma$ is said to have regularity $W^{s,p}_\text{loc}$, if the components of $\Gamma$ have regularity $W^{s,p}_\text{loc}$ in every coordinate patch $(x,\Omega)$ of $\mathcal{A}$, (i.e., if $\Gamma_x \in W^{s,p}_\text{loc}(\Omega_x)$ where $\Omega_x \equiv x(\Omega) \subseteq \R^n$), in which case we write $\Gamma \in W^{s,p}_{\mathcal{A},\text{loc}}(\mathcal{M})$ following the notation in \cite{ReintjesTemple_essreg}. Moreover, we write $\Riem(\Gamma)\in W^{s,p}_{\mathcal{A},\text{loc}}(\mathcal{M})$ if the components of the Riemann curvature tensor have regularity $W^{s,p}_\text{loc}$ in each coordinate system $(x,\Omega)$ of the atlas $\mathcal{A}$, (i.e., $\Riem(\Gamma_x) \in W^{s,p}_\text{loc}(\Omega_x)$), and we write $L^p_{\mathcal{A},\text{loc}}(\mathcal{M})$ for $W^{0,p}_{\mathcal{A},\text{loc}}(\mathcal{M})$.     
Note finally that an atlas that preserves $W^{s,p}_\text{loc}$-connection regularity always has transition maps of regularity $W^{s+2,p}_\text{loc}$, as a consequence of the connection transformation law; thus, $\Gamma \in W^{s,p}_{\A,\text{loc}}(\mathcal{M})$ implies that $\A$ has regularity  $W^{s+2,p}_\text{loc}$, cf. \cite[Lemma 2.1]{ReintjesTemple_essreg}.

\begin{Thm}\label{Thm_smoothing_global}  
Let $\Gamma$ be an affine connection defined on an $n$-dimensional manifold $\mathcal{M}$ with atlas $\mathcal{A}$. Assume $\Gamma \in L^p_{\mathcal{A},\text{loc}}(\mathcal{M})$, where $p>{\rm max}\{4,n\}$, $p<\infty$, $n\geq2$. Then there exists a subatlas $\A'$, contained in the maximal $W^{2,p}_\text{loc}$-extension of the atlas $\A$, such that $\Gamma \in W^{1,p}_{\A',\text{loc}}(\mathcal{M})$,  if and only if $\Riem(\Gamma)\in L^p_{\mathcal{A},\text{loc}}(\mathcal{M})$.
\end{Thm} 

Since, by Christoffel's formula, a metric is always exactly one derivative more regular than its connection, Theorem \ref{Thm_smoothing_global} directly implies Theorem \ref{Thm_smoothing_preliminaries}, on which the methods in this paper are based.

Theorem 2.3 in \cite{ReintjesTemple_essreg} asserts further that, in successive regularisation of a $W^{s,p}$ connection by the RT-equations, an implicit regularisation of the curvature takes place in each step $s\geq 0$ until the connection reaches its {\it essential} (highest possible) {\it regularity}, $\Gamma \in W^{m,p}_{\A_m,\text{loc}}(\mathcal{M})$, $m\geq 1$, in some suitable subatlas $\A_m$, $W^{2,p}$-compatible with $\A$. Note, since the regularisation requires singular transformations, this regularising subatlas $\A_m$ is contained within the maximal $W^{2,p}_\text{loc}$-extension of the starting atlas $\A$, the lowest atlas regularity that preserves $L^p_\text{loc}$-connection regularity and allows for regularisations by the RT-equations, (cf. Theorem \ref{Thm_Smoothing_low}). Since the atlas $\A'$ preserves the $W^{1,p}_\text{loc}$-regularity of the regularised connection $\Gamma$, it follows that all transition maps of $\A$ have regularity $W^{3,p}_\text{loc}$. One can always extend $\A'$ to its maximal $W^{3,p}_\text{loc}$-atlas and, by the results in \cite{Hirsch}, subsequently restricting this atlas to a maximal $C^\infty$-subatlas, while still preserving $W^{1,p}_\text{loc}$-connection regularity, cf. \cite{ReintjesTemple_essreg}. In this paper we always work with these maximal $C^\infty$-subatlases.

\section{Proof of Theorem \ref{Th:Hawking-smooth}} \label{Sec_Hawking_proof}

In this section we provide an independent proof of the `quantified' smooth Hawking Theorem \ref{Th:Hawking-smooth}.

\begin{proof}
Let $q\in I^+(\Sigma)$. Then there also exists a future-directed maximizing unit timelike geodesic 
$\gamma: [0, b] \to M$ starting orthogonally from $\Sigma$ with $\gamma(b)=q$ and $L(\gamma)=b=\tau_\Sigma(q)$.  

Let $h: [0, b] \to \mathbb{R}^+_0$ be a smooth function with $h(0)=1$ and $h(b)=0$. We choose a parallel orthonormal basis $E_1, \dots, E_{n-1}$ of $T\gamma^{\perp}$ along $\gamma$, by parallel-transporting along $\gamma$ initial vectors $e_i=E_i(0)\in T_{\gamma(0)}\Sigma$ ($i=1,\dots,n-1$) that are orthogonal. 
Define variation fields $V_i(t) = h(t)E_i(t)$ with $V_i(b) = 0$ and let $\gamma_s^i$ be the corresponding geodesic variations of $\gamma=\gamma_0$, with a fixed endpoint q. Since $\gamma$ is maximizing, the second variation of arc length must be non-positive for all variation fields $V_i$, i.e., for $i=1,\dots,n-1$
\begin{align}
\frac{d^2}{ds^2}L(\gamma^i_s)\bigg\vert_{s=0} \le 0.
\end{align}
Denoting by $I$ the index form along $\gamma$ and using \cite[Cor.\ 10.27]{ON83}, 
this implies 
\begin{align*}
0\ge I(V_i, V_i) = \langle \gamma'(0),\mathrm{II}(V_i(0),V_i(0)) \rangle - \int_0^b \left( \langle R(V_i, \gamma')\gamma', V_i \rangle + \langle V_i',  V_i' \rangle \right) dt.
\end{align*}

Note that $\gamma'(0)=N$ (the unit normal to $\Sigma$), $V_i(0)=e_i$ and $V_i'=h' E_i$. 
Summing over $i$ and using \eqref{eq:Ricci}, \eqref{eq:mean-curvature-def}  we arrive at
\begin{equation}\label{eq:second-var-II-Ric}
\begin{split}
0 &\ge \sum_{i=1}^{n-1}\langle \mathrm{II}(e_i,e_i),N \rangle -
\int_0^b \left( h^2 \sum_{i=1}^{n-1} \langle R(E_i, \gamma')\gamma', E_i \rangle + (h')^2 \sum_{i=1}^{n-1} \langle E_i, E_i \rangle \right) dt \\
&= -H + \int_0^b (h^2 \mathrm{Ric}(\gamma',\gamma') - (n-1)(h')^2)\,dt
\end{split}
\end{equation}

Applying assumptions (i) and (ii), we obtain
\begin{equation}\label{6}
-\beta \le (n-1)\int_0^b\bigl((h'(t))^2-\kappa\,h(t)^2\bigr)\,dt .
\end{equation}
Let
\begin{equation}
 J[h]:=\int_0^b\bigl((h'(t))^2-\kappa\,h(t)^2\bigr)\,dt ,
\end{equation}
where $h\in C^\infty([0,b])$ satisfies the boundary conditions $h(0)=1$ and $h(b)=0$.
Then \eqref{6} says
\begin{equation}\label{eq:b3}
J[h]\ge -\frac{\beta}{n-1}.
\end{equation}
To obtain the sharpest possible bound on $b$, we minimize $J[h]$ subject to these boundary conditions.
This is a standard variational problem. The Euler-Lagrange equation for $J$ is
\begin{equation}
 h''+\kappa h=0.    
\end{equation}

Let $h$ be the unique solution of this ODE with $h(0)=1$ and $h(b)=0$. Then
\[
J[h]=\int_0^b\bigl((h')^2-\kappa h^2\bigr)\,dt
=\int_0^b\bigl((h')^2+h\,h''\bigr)\,dt
=\bigl[h(t)h'(t)\bigr]_{t=0}^{t=b}
= -h'(0),
\]
hence the minimal value is $J_{\min}=-h'(0)$ and \eqref{eq:b3} yields
$h'(0)\le \frac{\beta}{n-1}$. 
The maximal admissible length $b$ (which we will denote by $b_{\kappa,\beta}$ for specific choices of the respective constants) is obtained in the extremal case
\begin{equation}\label{eq:extr}
 h'(0)=\frac{\beta}{n-1}.
\end{equation}

\medskip\noindent
\textbf{Case 1: $\kappa>0$.}
Set $k:=\sqrt{\kappa}$. The minimizer of $J$ with $h(0)=1$ and $h(b)=0$ is
\begin{equation}
    h(t)=\frac{\sin(k(b-t))}{\sin(kb)}, \qquad kb\in(0,\pi).
\end{equation}
Imposing the extremality condition \eqref{eq:extr} 
therefore yields
\begin{equation}
    \cot(kb)=-\frac{\beta}{(n-1)k}.
\end{equation}
Since $\cot:(0,\pi)\to\mathbb{R}$ is strictly decreasing and surjective, there exists a unique
$kb\in(0,\pi)$ solving this equation, namely
\begin{equation}
kb=\cot^{-1}\!\left(-\frac{\beta}{(n-1)k}\right).
\end{equation}
This yields the claimed $b=b_{\kappa,\beta}>0$.

\medskip\noindent
\textbf{Case 2: $\kappa=0$}. Then $h''=0$, and the minimizer is $h(t)=1-t/b$, so imposing extremality condition \eqref{eq:extr} yields
\begin{equation}
    b=b_{\kappa,\beta}=-\frac{n-1}{\beta}
\end{equation}
if $\beta<0$. If $\beta\ge 0$, then it
has no solution for finite $b$. Equivalently, for every finite $b>0$ one has
\begin{equation}
J[h]=-\;h'(0)=\frac{1}{b}>0,
\end{equation}
so \eqref{eq:b3} is satisfied for all finite $b$, hence
$b_{\kappa,\beta}=\infty$.

\medskip\noindent
\textbf{Case 3: $\kappa<0$.}
Set $k:=\sqrt{-\kappa}>0$. The equation becomes $h''-k^2h=0$, and the solution with boundary conditions
$h(0)=1$ and $h(b)=0$ is
\begin{equation}
h(t)=\frac{\sinh(k(b-t))}{\sinh(kb)}.
\end{equation}
Then $h'(0)=-k\,\coth(kb)$, and imposing the extremality condition \eqref{eq:extr} 
gives
$\coth(kb)=-\frac{\beta}{(n-1)k}$, i.e.,
\begin{equation}\label{eq:final-case}
b=b_{\kappa,\beta}=\frac{1}{k}\coth^{-1}\!\left(-\frac{\beta}{(n-1)k}\right),
\end{equation}
under the usual condition that the argument of $\coth^{-1}$ is $>1$. 

Now assume otherwise, i.e., $\frac{\beta}{(n-1)k}\ge -1$. 
Then again \eqref{eq:final-case} has no positive solution $b$ and equivalently we have
\begin{equation}
    k\,\coth(kb) > k \ge -\frac{\beta}{n-1},
\end{equation}
so \eqref{eq:b3} is satisfied for every finite $b>0$, implying $b_{\kappa,\beta}=\infty$.
\end{proof}

Observe that there is an alternative way of deriving equation \eqref{eq:second-var-II-Ric} using  Synge's formula for the second variation of arc length \cite[Thm.\ 10.4]{ON83} avoiding the use of the index form.

\section*{Funding}
M. Kunzinger, R. Steinbauer, and I. Vega-Gonz\'alez were funded in part by the Austrian Science Fund (FWF) [Grants DOI \href{https://doi.org/10.55776/PAT1996423}{10.55776/PAT1996423}, and \href{https://doi.org/10.55776/EFP6}{10.55776/EFP6}]. 
M. Reintjes was supported by the Hong Kong Research Grants Council (ECS-21306524 and GRF-11303326).
For open access purposes, the authors have applied a CC BY public copyright license to any author accepted manuscript version arising from this submission.

\section*{Acknowledgements} 
We thank Argam Ohanyan and Melanie Graf for helpful discussions. MK and RS are also grateful to the organizers of the workshop \emph{Optimal Transport and Metric Geometry Across Structures} at EPFL, Mathias Braun, Nicola Gigli, and Robert McCann, for creating a stimulating environment in which some of the ideas for this work emerged. The authors also benefited from occasional discussions with large language models during the preparation of this work. The final formulation of all results and proofs is due solely to the authors.

\printbibliography

@article {BS06,
    AUTHOR = {Bernal, Antonio N. and S\'anchez, Miguel},
     TITLE = {Further results on the smoothability of {C}auchy hypersurfaces
              and {C}auchy time functions},
   JOURNAL = {Lett. Math. Phys.},
  FJOURNAL = {Letters in Mathematical Physics},
    VOLUME = {77},
      YEAR = {2006},
    NUMBER = {2},
     PAGES = {183--197},
      ISSN = {0377-9017,1573-0530},
   MRCLASS = {53C50 (53C80 81T20)},
  MRNUMBER = {2254187},
MRREVIEWER = {Paul\ E.\ Ehrlich},
       DOI = {10.1007/s11005-006-0091-5},
       URL = {https://doi.org/10.1007/s11005-006-0091-5},
}

@article {CGHKS25,
    AUTHOR = {Calisti, Matteo and Graf, Melanie and Hafemann, Eduardo and
              Kunzinger, Michael and Steinbauer, Roland},
     TITLE = {Hawking's singularity theorem for {L}ipschitz {L}orentzian
              metrics},
   JOURNAL = {Comm. Math. Phys.},
  FJOURNAL = {Communications in Mathematical Physics},
    VOLUME = {406},
      YEAR = {2025},
    NUMBER = {9},
     PAGES = {Paper No. 207, 31},
      ISSN = {0010-3616,1432-0916},
   MRCLASS = {83C75 (46N50 52C20 53C50)},
  MRNUMBER = {4940197},
       DOI = {10.1007/s00220-025-05380-9},
       URL = {https://doi.org/10.1007/s00220-025-05380-9},
}

@article {Graf,
    AUTHOR = {Graf, Melanie},
     TITLE = {Singularity theorems for {$C^1$}-{L}orentzian metrics},
   JOURNAL = {Comm. Math. Phys.},
  FJOURNAL = {Communications in Mathematical Physics},
    VOLUME = {378},
      YEAR = {2020},
    NUMBER = {2},
     PAGES = {1417--1450},
      ISSN = {0010-3616,1432-0916},
   MRCLASS = {53C50 (53C20 83C75)},
  MRNUMBER = {4134950},
MRREVIEWER = {Clemens\ Saemann},
       DOI = {10.1007/s00220-020-03808-y},
       URL = {https://doi.org/10.1007/s00220-020-03808-y},
}

@book {gkos2001geometricgeneralized,
    AUTHOR = {Grosser, Michael and Kunzinger, Michael and Oberguggenberger,
              Michael and Steinbauer, Roland},
     TITLE = {Geometric theory of generalized functions with applications to
              general relativity},
    SERIES = {Mathematics and its Applications},
    VOLUME = {537},
 PUBLISHER = {Kluwer Academic Publishers, Dordrecht},
      YEAR = {2001},
     PAGES = {xvi+505},
      ISBN = {1-4020-0145-2},
   MRCLASS = {46T30 (35A30 35Q75 46-02 46F30 58-02 83C99)},
  MRNUMBER = {1883263},
MRREVIEWER = {Stevan\ Pilipovi\'c},
       DOI = {10.1007/978-94-015-9845-3},
       URL = {https://doi.org/10.1007/978-94-015-9845-3},
}

@book{HawkingEllis,
	place={Cambridge},
	series={Cambridge Monographs on Mathematical Physics},
	title={The Large Scale Structure of Space-Time},
	publisher={Cambridge University Press},
	author={Hawking, S. W. and Ellis, G. F. R.},
	year={1973},
	collection={Cambridge Monographs on Mathematical Physics},
}

@book {Hirsch,
    AUTHOR = {Hirsch, Morris W.},
     TITLE = {Differential topology},
    SERIES = {Graduate Texts in Mathematics},
    VOLUME = {No. 33},
 PUBLISHER = {Springer-Verlag, New York-Heidelberg},
      YEAR = {1976},
     PAGES = {x+221},
   MRCLASS = {57DXX (58-01)},
  MRNUMBER = {448362},
MRREVIEWER = {Ulrich\ Koschorke},
}

@article {KSSV:14,
    AUTHOR = {Kunzinger, Michael and Steinbauer, Roland and Stojkovi\'c,
              Milena and Vickers, James A.},
     TITLE = {A regularisation approach to causality theory for
              {$C^{1,1}$}-{L}orentzian metrics},
   JOURNAL = {Gen. Relativity Gravitation},
  FJOURNAL = {General Relativity and Gravitation},
    VOLUME = {46},
      YEAR = {2014},
    NUMBER = {8},
     PAGES = {Art. 1738, 18},
      ISSN = {0001-7701,1572-9532},
   MRCLASS = {83C75},
  MRNUMBER = {3245957},
       DOI = {10.1007/s10714-014-1738-7},
       URL = {https://doi.org/10.1007/s10714-014-1738-7},
}

@article {LLS:21,
    AUTHOR = {Lange, Christian and Lytchak, Alexander and S\"amann, Clemens},
     TITLE = {Lorentz meets {L}ipschitz},
   JOURNAL = {Adv. Theor. Math. Phys.},
  FJOURNAL = {Advances in Theoretical and Mathematical Physics},
    VOLUME = {25},
      YEAR = {2021},
    NUMBER = {8},
     PAGES = {2141--2170},
      ISSN = {1095-0761,1095-0753},
   MRCLASS = {53B30 (83C99)},
  MRNUMBER = {4489478},
MRREVIEWER = {Stefan\ Suhr},
       DOI = {10.4310/atmp.2021.v25.n8.a4},
       URL = {https://doi.org/10.4310/atmp.2021.v25.n8.a4},
}

@article {Minguzzi-cone-structures,
    AUTHOR = {Minguzzi, Ettore},
     TITLE = {Causality theory for closed cone structures with applications},
   JOURNAL = {Rev. Math. Phys.},
  FJOURNAL = {Reviews in Mathematical Physics. A Journal for Both Review and
              Original Research Papers in the Field of Mathematical Physics},
    VOLUME = {31},
      YEAR = {2019},
    NUMBER = {5},
     PAGES = {1930001, 139},
      ISSN = {0129-055X,1793-6659},
   MRCLASS = {53C50 (34A60 53C22 83C75 93D30)},
  MRNUMBER = {3955368},
MRREVIEWER = {Clemens\ Saemann},
       DOI = {10.1142/S0129055X19300012},
       URL = {https://doi.org/10.1142/S0129055X19300012},
}

@article {Minguzzi-Cauchy-surface,
    AUTHOR = {Minguzzi, Ettore},
     TITLE = {On the regularity of {C}auchy hypersurfaces and temporal
              functions in closed cone structures},
   JOURNAL = {Rev. Math. Phys.},
  FJOURNAL = {Reviews in Mathematical Physics. A Journal for Both Review and
              Original Research Papers in the Field of Mathematical Physics},
    VOLUME = {32},
      YEAR = {2020},
    NUMBER = {10},
     PAGES = {2050033, 17},
      ISSN = {0129-055X,1793-6659},
   MRCLASS = {53C40 (34A60 53C50 83C75 93D30)},
  MRNUMBER = {4179231},
MRREVIEWER = {J\'onatan\ Herrera},
       DOI = {10.1142/S0129055X20500336},
       URL = {https://doi.org/10.1142/S0129055X20500336},
}

@book {ON83,
    AUTHOR = {O'Neill, Barrett},
     TITLE = {Semi-{R}iemannian geometry},
    SERIES = {Pure and Applied Mathematics},
    VOLUME = {103},
      NOTE = {With applications to relativity},
 PUBLISHER = {Academic Press, Inc. [Harcourt Brace Jovanovich, Publishers],
              New York},
      YEAR = {1983},
     PAGES = {xiii+468},
      ISBN = {0-12-526740-1},
   MRCLASS = {53-01 (53B30 53C50 83-02)},
  MRNUMBER = {719023},
MRREVIEWER = {N.\ V.\ Mitskevich},
}

@article {Reintjes_SCC,
    AUTHOR = {Reintjes, Moritz},
     TITLE = {Strong cosmic censorship with bounded curvature},
   JOURNAL = {Classical Quantum Gravity},
  FJOURNAL = {Classical and Quantum Gravity},
    VOLUME = {41},
      YEAR = {2024},
    NUMBER = {17},
     PAGES = {Paper No. 175002, 13},
      ISSN = {0264-9381,1361-6382},
   MRCLASS = {83C75 (83C05)},
  MRNUMBER = {4797036},
MRREVIEWER = {K.\ S.\ Virbhadra},
}

@article {ReintjesTemple_geo,
    AUTHOR = {Reintjes, Moritz and Temple, Blake},
     TITLE = {Shock wave interactions and the {R}iemann-flat condition: the geometry behind metric smoothing and the existence of locally inertial frames in general relativity},
   JOURNAL = {Arch. Ration. Mech. Anal.},
  FJOURNAL = {Archive for Rational Mechanics and Analysis},
    VOLUME = {235},
      YEAR = {2020},
    NUMBER = {3},
     PAGES = {1873--1904},
      ISSN = {0003-9527,1432-0673},
   MRCLASS = {83C05 (35Q75)},
  MRNUMBER = {4065653},
MRREVIEWER = {Giovanni\ Preti},
       DOI = {10.1007/s00205-019-01456-8},
       URL = {https://doi.org/10.1007/s00205-019-01456-8},
}

@article {ReintjesTemple_ell1,
    AUTHOR = {Reintjes, Moritz and Temple, Blake},
     TITLE = {The regularity transformation equations: an elliptic mechanism for smoothing gravitational metrics in general relativity},
   JOURNAL = {Adv. Theor. Math. Phys.},
  FJOURNAL = {Advances in Theoretical and Mathematical Physics},
    VOLUME = {24},
      YEAR = {2020},
    NUMBER = {5},
     PAGES = {1203--1245},
      ISSN = {1095-0761,1095-0753},
   MRCLASS = {83C75 (35Q75 58J90 83C05)},
  MRNUMBER = {4285602},
MRREVIEWER = {A.\ Burtscher},
       DOI = {10.4310/ATMP.2020.v24.n5.a5},
       URL = {https://doi.org/10.4310/ATMP.2020.v24.n5.a5},
}

@article {ReintjesTemple_ell2,
    AUTHOR = {Reintjes, Moritz and Temple, Blake},
     TITLE = {Optimal metric regularity in general relativity follows from the {RT}-equations by elliptic regularity theory in
              {$L^p$}-spaces},
   JOURNAL = {Methods Appl. Anal.},
  FJOURNAL = {Methods and Applications of Analysis},
    VOLUME = {27},
      YEAR = {2020},
    NUMBER = {3},
     PAGES = {199--241},
      ISSN = {1073-2772,1945-0001},
   MRCLASS = {83C05 (76L05 83C75)},
  MRNUMBER = {4307029},
       DOI = {10.4310/MAA.2020.v27.n3.a1},
       URL = {https://doi.org/10.4310/MAA.2020.v27.n3.a1},
}

@article {ReintjesTemple_ell3,
    AUTHOR = {Reintjes, Moritz and Temple, Blake},
     TITLE = {How to smooth a crinkled map of space-time: {U}hlenbeck
              compactness for {$L^\infty$} connections and optimal
              regularity for general relativistic shock waves by the
              {R}eintjes-{T}emple equations},
   JOURNAL = {Proc. A.},
  FJOURNAL = {Proceedings A},
    VOLUME = {476},
      YEAR = {2020},
    NUMBER = {2241},
     PAGES = {20200177, 22},
      ISSN = {1364-5021,1471-2946},
   MRCLASS = {35Q75 (83C05)},
  MRNUMBER = {4172534},
MRREVIEWER = {L.\ Ravera},
       DOI = {10.1098/rspa.2020.0177},
       URL = {https://doi.org/10.1098/rspa.2020.0177},
}

@article {ReintjesTemple_ell4,
    AUTHOR = {Reintjes, Moritz and Temple, Blake},
     TITLE = {On the optimal regularity implied by the assumptions of
              geometry {I}: {C}onnections on tangent bundles},
   JOURNAL = {Methods Appl. Anal.},
  FJOURNAL = {Methods and Applications of Analysis},
    VOLUME = {29},
      YEAR = {2022},
    NUMBER = {4},
     PAGES = {303--396},
      ISSN = {1073-2772,1945-0001},
   MRCLASS = {83C75 (53B05 58J05)},
  MRNUMBER = {4584011},
MRREVIEWER = {Yuguang\ Shi},
       DOI = {10.4310/maa.2022.v29.n4.a1},
       URL = {https://doi.org/10.4310/maa.2022.v29.n4.a1},
}

@article {ReintjesTemple_ell5,
    AUTHOR = {Reintjes, Moritz and Temple, Blake},
     TITLE = {On the optimal regularity implied by the assumptions of
              geometry {II}: {C}onnections on vector bundles},
   JOURNAL = {Adv. Theor. Math. Phys.},
  FJOURNAL = {Advances in Theoretical and Mathematical Physics},
    VOLUME = {28},
      YEAR = {2024},
    NUMBER = {5},
     PAGES = {1425--1486},
      ISSN = {1095-0761,1095-0753},
   MRCLASS = {83C75 (53B05 53B15 53C07 58J05)},
  MRNUMBER = {4819557},
MRREVIEWER = {A.\ Burtscher},
       DOI = {10.4310/atmp.241030211154},
       URL = {https://doi.org/10.4310/atmp.241030211154},
}

@article {ReintjesTemple_ell6,
    AUTHOR = {Reintjes, Moritz and Temple, Blake},
     TITLE = {Optimal regularity and {U}hlenbeck compactness for general
              relativity and {Y}ang-{M}ills theory},
   JOURNAL = {Proc. A.},
  FJOURNAL = {Proceedings A},
    VOLUME = {479},
      YEAR = {2023},
    NUMBER = {2271},
     PAGES = {Paper No. 20220444, 14},
      ISSN = {1364-5021,1471-2946},
   MRCLASS = {83C05 (53B05 53C50 70S15 83C75)},
  MRNUMBER = {4574825},
MRREVIEWER = {Michael\ Kunzinger},
}

@article {ReintjesTemple_essreg,
  title={The essential regularity of singular connections in Geometry},
  author={Reintjes, Moritz and Temple, Blake},
  journal={arXiv preprint arXiv:2412.08928},
  year={2024}
}

@article {GrantKuSaSt,
    AUTHOR = {Grant, James D. E. and Kunzinger, Michael and S\"amann,
              Clemens and Steinbauer, Roland},
     TITLE = {The future is not always open},
   JOURNAL = {Lett. Math. Phys.},
  FJOURNAL = {Letters in Mathematical Physics},
    VOLUME = {110},
      YEAR = {2020},
    NUMBER = {1},
     PAGES = {83--103},
      ISSN = {0377-9017,1573-0530},
   MRCLASS = {53C50 (83C75)},
  MRNUMBER = {4047145},
MRREVIEWER = {Peter\ R.\ Law},
       DOI = {10.1007/s11005-019-01213-8},
       URL = {https://doi.org/10.1007/s11005-019-01213-8},
}

@article {benedikt,
    AUTHOR = {Schinnerl, Benedict and Steinbauer, Roland},
     TITLE = {A note on the {G}annon-{L}ee theorem},
   JOURNAL = {Lett. Math. Phys.},
  FJOURNAL = {Letters in Mathematical Physics},
    VOLUME = {111},
      YEAR = {2021},
    NUMBER = {6},
     PAGES = {Paper No. 142, 17},
      ISSN = {0377-9017},
   MRCLASS = {83C75 (46F99 53C50)},
  MRNUMBER = {4342989},
MRREVIEWER = {Clemens Saemann},
       DOI = {10.1007/s11005-021-01481-3},
       URL = {https://doi.org/10.1007/s11005-021-01481-3},
}

@article {bernard_suhr,
    AUTHOR = {Bernard, Patrick and Suhr, Stefan},
     TITLE = {Cauchy and uniform temporal functions of globally hyperbolic
              cone fields},
   JOURNAL = {Proc. Amer. Math. Soc.},
  FJOURNAL = {Proceedings of the American Mathematical Society},
    VOLUME = {148},
      YEAR = {2020},
    NUMBER = {11},
     PAGES = {4951--4966},
      ISSN = {0002-9939},
   MRCLASS = {53C50 (37B25 37D05)},
  MRNUMBER = {4143406},
MRREVIEWER = {Clemens Saemann},
       DOI = {10.1090/proc/15106},
       URL = {https://doi.org/10.1090/proc/15106},
}

@article {ReintjesTemple_geod,
    AUTHOR = {Reintjes, Moritz and Temple, Blake},
     TITLE = {On weak solutions to the geodesic equation in the presence of
              curvature bounds},
   JOURNAL = {J. Differential Equations},
  FJOURNAL = {Journal of Differential Equations},
    VOLUME = {392},
      YEAR = {2024},
     PAGES = {306--324},
      ISSN = {0022-0396,1090-2732},
   MRCLASS = {53C22 (34A26 83C75)},
  MRNUMBER = {4711905},
MRREVIEWER = {Stefan\ Suhr},
       DOI = {10.1016/j.jde.2024.02.014},
       URL = {https://doi.org/10.1016/j.jde.2024.02.014},
}

@article {Andersson-Galloway,
    AUTHOR = {Andersson, Lars and Galloway, Gregory J.},
     TITLE = {d{S}/{CFT} and spacetime topology},
   JOURNAL = {Adv. Theor. Math. Phys.},
  FJOURNAL = {Advances in Theoretical and Mathematical Physics},
    VOLUME = {6},
      YEAR = {2002},
    NUMBER = {2},
     PAGES = {307--327},
      ISSN = {1095-0761,1095-0753},
   MRCLASS = {53C50 (83C05)},
  MRNUMBER = {1937858},
MRREVIEWER = {Alan\ D.\ Rendall},
       DOI = {10.4310/ATMP.2002.v6.n2.a4},
       URL = {https://doi.org/10.4310/ATMP.2002.v6.n2.a4},
}

@article {Borde,
    AUTHOR = {Borde, Arvind},
     TITLE = {Open and closed universes, initial singularities, and
              inflation},
   JOURNAL = {Phys. Rev. D (3)},
  FJOURNAL = {Physical Review. D. Third Series},
    VOLUME = {50},
      YEAR = {1994},
    NUMBER = {6},
     PAGES = {3692--3702},
      ISSN = {0556-2821},
   MRCLASS = {83C75},
  MRNUMBER = {1295003},
MRREVIEWER = {Robert\ J.\ Low},
       DOI = {10.1103/PhysRevD.50.3692},
       URL = {https://doi.org/10.1103/PhysRevD.50.3692},
}

@article {Graf_Volume,
    AUTHOR = {Graf, Melanie},
     TITLE = {Volume comparison for {$\mathcal{C}^{1,1}$}-metrics},
   JOURNAL = {Ann. Global Anal. Geom.},
  FJOURNAL = {Annals of Global Analysis and Geometry},
    VOLUME = {50},
      YEAR = {2016},
    NUMBER = {3},
     PAGES = {209--235},
      ISSN = {0232-704X,1572-9060},
   MRCLASS = {53C20 (53C50 83C75)},
  MRNUMBER = {3554372},
MRREVIEWER = {Peter\ R.\ Law},
       DOI = {10.1007/s10455-016-9508-2},
       URL = {https://doi.org/10.1007/s10455-016-9508-2},
}

@article{braun_mccann_causal,
  author  = {Braun, Mathias and McCann, Robert},
  title   = {Causal convergence conditions through variable timelike Ricci curvature bounds},
  journal = {Memoirs of the European Mathematical Society},
  note    = {In press, 109 pp.}
}

@article {minguzzi-suhr24,
    AUTHOR = {Minguzzi, E. and Suhr, S.},
     TITLE = {Lorentzian metric spaces and their {G}romov-{H}ausdorff
              convergence},
   JOURNAL = {Lett. Math. Phys.},
  FJOURNAL = {Letters in Mathematical Physics},
    VOLUME = {114},
      YEAR = {2024},
    NUMBER = {3},
     PAGES = {Paper No. 73, 63},
      ISSN = {0377-9017,1573-0530},
   MRCLASS = {53C50 (51F99 83C45)},
  MRNUMBER = {4752400},
MRREVIEWER = {S.\ M. B. Kashani},
       DOI = {10.1007/s11005-024-01813-z},
       URL = {https://doi.org/10.1007/s11005-024-01813-z},
}

@article {bykov-minguzzi-suhr,
    AUTHOR = {Bykov, A. and Minguzzi, E. and Suhr, S.},
     TITLE = {Lorentzian metric spaces and {GH}-convergence: the unbounded
              case},
   JOURNAL = {Lett. Math. Phys.},
  FJOURNAL = {Letters in Mathematical Physics},
    VOLUME = {115},
      YEAR = {2025},
    NUMBER = {3},
     PAGES = {Paper No. 63, 66},
      ISSN = {0377-9017,1573-0530},
   MRCLASS = {53C50 (51F99 83C45)},
  MRNUMBER = {4914034},
       DOI = {10.1007/s11005-025-01941-0},
       URL = {https://doi.org/10.1007/s11005-025-01941-0},
}

@article{AS:15,
  author    = {Morales {\'A}lvarez, Pablo and S{\'a}nchez, Miguel},
  title     = {{Myers and Hawking Theorems: Geometry for the Limits of the Universe}},
  journal   = {Milan Journal of Mathematics},
  volume    = {83},
  pages     = {295--311},
  year      = {2015},
  doi       = {10.1007/s00032-015-0241-2}
}

@unpublishe{BR:26,
    AUTHOR = {Braun, Mathias and Rotolo, Carlo},
     TITLE = {A synthetic {G}annon--{L}ee incompleteness theorem},
      YEAR = {2026},
      NOTE = {Preprint, arXiv:2602.14246 [gr-qc]},
}

@article{CM:22,
    AUTHOR = {Cavalletti, Fabio and Mondino, Andrea},
     TITLE = {A review of {L}orentzian synthetic theory of timelike {R}icci
              curvature bounds},
   JOURNAL = {Gen. Relativ. Gravit.},
  FJOURNAL = {General Relativity and Gravitation},
    VOLUME = {54},
      YEAR = {2022},
    NUMBER = {11},
     PAGES = {Paper No. 137, 39},
      ISSN = {0001-7701},
   MRCLASS = {53C23 (49Q22 53C50 83C99)},
  MRNUMBER = {4554410},
       DOI = {10.1007/s10714-022-03004-4},
}

@article{CM:24,
    AUTHOR = {Cavalletti, Fabio and Mondino, Andrea},
     TITLE = {Optimal transport in {L}orentzian synthetic spaces, synthetic
              timelike {R}icci curvature lower bounds and applications},
   JOURNAL = {Camb. J. Math.},
  FJOURNAL = {Cambridge Journal of Mathematics},
    VOLUME = {12},
      YEAR = {2024},
    NUMBER = {2},
     PAGES = {417--534},
      ISSN = {2168-0930},
       DOI = {10.4310/CJM.2024.v12.n2.a3},
}

@article{CMM:25a,
    AUTHOR = {Cavalletti, Fabio and Manini, Davide and Mondino, Andrea},
     TITLE = {Optimal transport on null hypersurfaces and the null energy
              condition},
   JOURNAL = {Comm. Math. Phys.},
  FJOURNAL = {Communications in Mathematical Physics},
    VOLUME = {406},
      YEAR = {2025},
    NUMBER = {9},
     PAGES = {Paper No. 212, 62},
      ISSN = {0010-3616},
       DOI = {10.1007/s00220-025-05345-y},
}

@unpublished{CMM:25b,
    AUTHOR = {Cavalletti, Fabio and Manini, Davide and Mondino, Andrea},
     TITLE = {On the geometry of synthetic null hypersurfaces},
      YEAR = {2025},
      NOTE = {Preprint, arXiv:2506.04934 [math.DG]},
}

@article {CG:12,
    AUTHOR = {Chru{\'s}ciel, Piotr T. and Grant, James D. E.},
     TITLE = {On {L}orentzian causality with continuous metrics},
   JOURNAL = {Classical Quantum Gravity},
  FJOURNAL = {Classical and Quantum Gravity},
    VOLUME = {29},
      YEAR = {2012},
    NUMBER = {14},
     PAGES = {145001, 32},
      ISSN = {0264-9381},
     CODEN = {CQGRDG},
   MRCLASS = {53C50 (83A05)},
  MRNUMBER = {2949547},
MRREVIEWER = {Miguel S{\'a}nchez},
       DOI = {10.1088/0264-9381/29/14/145001},
       URL = {http://dx.doi.org/10.1088/0264-9381/29/14/145001},
}

@Book{Fil:88,
  author	= {Filippov, Aleksei F.},
  title		= {Differential equations with discontinuous righthand
		  sides},
  series	= {Mathematics and its Applications (Soviet Series)},
  volume	= {18},
  publisher	= {Kluwer Academic Publishers Group},
  ADDRESS       = {Dordrecht},
  year		= {1988},
  pages		= {x+304},
  isbn		= {90-277-2699-X},
  mrclass	= {34-02 (58F99)},
  mrnumber	= {1028776},
  doi		= {10.1007/978-94-015-7793-9},
  url		= {https://doi.org/10.1007/978-94-015-7793-9}
}

@article {FS:12,
    AUTHOR = {Fathi, Albert and Siconolfi, Antonio},
     TITLE = {On smooth time functions},
   JOURNAL = {Math. Proc. Cambridge Philos. Soc.},
  FJOURNAL = {Mathematical Proceedings of the Cambridge Philosophical
              Society},
    VOLUME = {152},
      YEAR = {2012},
    NUMBER = {2},
     PAGES = {303--339},
      ISSN = {0305-0041},
     CODEN = {MPCPCO},
   MRCLASS = {53C50 (37J50 53C80)},
  MRNUMBER = {2887877},
MRREVIEWER = {Mikhail B. Sevryuk},
       DOI = {10.1017/S0305004111000661},
       URL = {http://dx.doi.org/10.1017/S0305004111000661},
}

@Article{GGKS:18,
  AUTHOR = {Graf, Melanie and Grant, James D. E. and Kunzinger, Michael
              and Steinbauer, Roland},
     TITLE = {The {H}awking--{P}enrose {S}ingularity {T}heorem for
              {$C^{1,1}$}-{L}orentzian {M}etrics},
   JOURNAL = {Comm. Math. Phys.},
  FJOURNAL = {Communications in Mathematical Physics},
    VOLUME = {360},
      YEAR = {2018},
    NUMBER = {3},
     PAGES = {1009--1042},
      ISSN = {0010-3616},
   MRCLASS = {53 (83)},
  MRNUMBER = {3803816},
       DOI = {10.1007/s00220-017-3047-y},
       URL = {https://doi.org/10.1007/s00220-017-3047-y},
}

@article {GKOS:22,
    AUTHOR = {Graf, Melanie and Kontou, Eleni-Alexandra and Ohanyan, Argam
              and Schinnerl, Benedict},
     TITLE = {Hawking-type singularity theorems for worldvolume energy
              inequalities},
   JOURNAL = {Ann. Henri Poincar\'e},
  FJOURNAL = {Annales Henri Poincar\'e. A Journal of Theoretical and
              Mathematical Physics},
    VOLUME = {26},
      YEAR = {2025},
    NUMBER = {11},
     PAGES = {3871--3906},
      ISSN = {1424-0637,1424-0661},
   MRCLASS = {83C75 (53B30 53C50 70S20)},
  MRNUMBER = {4970217},
       DOI = {10.1007/s00023-024-01502-6},
       URL = {https://doi.org/10.1007/s00023-024-01502-6},
}

@Article{GT:87,
  author    = {Geroch, Robert and Traschen, Jennie},
  title		= {Strings and other distributional sources in general
		  relativity},
  journal	= {Phys.~Rev.~D},
  year		= {1987},
  optkey	= {},
  volume	= {{36}},
  number	= {4},
  pages		= {1017--1031},
  optmonth	= {},
  optnote	= {},
  optannote	= {}
}

@article {GL:18,
	AUTHOR = {Graf, Melanie and Ling, Eric},
	TITLE = {Maximizers in {L}ipschitz spacetimes are either timelike or
	null},
	JOURNAL = {Classical Quantum Gravity},
	FJOURNAL = {Classical and Quantum Gravity},
	VOLUME = {35},
	YEAR = {2018},
	NUMBER = {8},
	PAGES = {087001, 6},
	ISSN = {0264-9381},
	MRCLASS = {83C05},
	MRNUMBER = {3789523},
	DOI = {10.1088/1361-6382/aab259},
	URL = {https://doi.org/10.1088/1361-6382/aab259},
}

@article {HW:51,
    AUTHOR = {Hartman, Philip and Wintner, Aurel},
     TITLE = {On the problems of geodesics in the small},
   JOURNAL = {Amer. J. Math.},
  FJOURNAL = {American Journal of Mathematics},
    VOLUME = {73},
      YEAR = {1951},
     PAGES = {132--148},
      ISSN = {0002-9327},
   MRCLASS = {53.0X},
  MRNUMBER = {0040740 (12,742e)},
MRREVIEWER = {A. Fialkow},
}

@article{Haw:67,
 ISSN = {00804630},
 URL = {http://www.jstor.org/stable/2415769},
 author = {Hawking, Stephen W.},
 journal = {Proceedings of the Royal Society of London. Series A, Mathematical and Physical Sciences},
 number = {1461},
 pages = {187--201},
 publisher = {The Royal Society},
 title = {The Occurrence of Singularities in Cosmology. III. Causality and Singularities},
 volume = {300},
 year = {1967},
 DOI = {https://doi.org/10.1098/rspa.1967.0164}
}

@article{K:24,
    AUTHOR = {Ketterer, Christian},
     TITLE = {Characterization of the null energy condition via displacement
              convexity of entropy},
   JOURNAL = {J. Lond. Math. Soc. (2)},
  FJOURNAL = {Journal of the London Mathematical Society. Second Series},
    VOLUME = {109},
      YEAR = {2024},
    NUMBER = {1},
     PAGES = {Paper No. e12846, 24},
      ISSN = {0024-6107},
       DOI = {10.1112/jlms.12846},
}

@article {KS:18,
    AUTHOR = {Kunzinger, Michael and S\"{a}mann, Clemens},
     TITLE = {Lorentzian length spaces},
   JOURNAL = {Ann. Global Anal. Geom.},
  FJOURNAL = {Annals of Global Analysis and Geometry},
    VOLUME = {54},
      YEAR = {2018},
    NUMBER = {3},
     PAGES = {399--447},
      ISSN = {0232-704X},
   MRCLASS = {53C23 (53B30 53C50 53C80)},
  MRNUMBER = {3867652},
MRREVIEWER = {Benjam\'{\i}n Olea},
       DOI = {10.1007/s10455-018-9633-1},
       URL = {https://doi.org/10.1007/s10455-018-9633-1},
}

@article {KSSV:15,
    AUTHOR = {Kunzinger, Michael and Steinbauer, Roland and Stojkovi{\'c},
              Milena and Vickers, James A.},
     TITLE = {Hawking's singularity theorem for {$C^{1,1}$}-metrics},
   JOURNAL = {Classical Quantum Gravity},
  FJOURNAL = {Classical and Quantum Gravity},
    VOLUME = {32},
      YEAR = {2015},
    NUMBER = {7},
     PAGES = {075012, 19},
      ISSN = {0264-9381},
   MRCLASS = {83C75 (83C05)},
  MRNUMBER = {3322151},
MRREVIEWER = {Peter R. Law},
       DOI = {10.1088/0264-9381/32/7/075012},
       URL = {http://dx.doi.org/10.1088/0264-9381/32/7/075012},
}

@article {KSV:15,
    AUTHOR = {Kunzinger, Michael and Steinbauer, Roland and Vickers, James
              A.},
     TITLE = {The {P}enrose singularity theorem in regularity {$C^{1,1}$}},
   JOURNAL = {Classical Quantum Gravity},
  FJOURNAL = {Classical and Quantum Gravity},
    VOLUME = {32},
      YEAR = {2015},
    NUMBER = {15},
     PAGES = {155010, 12},
      ISSN = {0264-9381},
   MRCLASS = {83C75 (53Z05)},
  MRNUMBER = {3368951},
MRREVIEWER = {Jos\'e Nat\'ario},
       DOI = {10.1088/0264-9381/32/15/155010},
       URL = {http://dx.doi.org/10.1088/0264-9381/32/15/155010},
}

@Article{KOSS:22,
  title={The Hawking--Penrose singularity theorem for C 1-Lorentzian metrics},
  author={Kunzinger, Michael and Ohanyan, Argam and Schinnerl, Benedict and Steinbauer, Roland},
  journal={Communications in Mathematical Physics},
  volume={391},
  number={3},
  pages={1143--1179},
  year={2022},
  publisher={Springer}
}

@article {KOV:22,
    AUTHOR = {Kunzinger, Michael and Oberguggenberger, Michael and Vickers,
              James A.},
     TITLE = {Synthetic versus distributional lower {R}icci curvature
              bounds},
   JOURNAL = {Proc. Roy. Soc. Edinburgh Sect. A},
  FJOURNAL = {Proceedings of the Royal Society of Edinburgh. Section A.
              Mathematics},
    VOLUME = {154},
      YEAR = {2024},
    NUMBER = {5},
     PAGES = {1406--1430},
      ISSN = {0308-2105,1473-7124},
   MRCLASS = {49Q22 (46T30 53C21)},
  MRNUMBER = {4806282},
MRREVIEWER = {Rotem\ Assouline},
       DOI = {10.1017/prm.2023.70},
       URL = {https://doi.org/10.1017/prm.2023.70},
}

@Article{LM:07,
  author	= {LeFloch, Philippe G. and Mardare, Cristinel},
  title		= {Definition and stability of {L}orentzian manifolds with
		  distributional curvature},
  journal	= {Port. Math. (N.S.)},
  fjournal	= {Portugaliae Mathematica. Nova S\'{e}rie},
  volume	= {64},
  year		= {2007},
  number	= {4},
  pages		= {535--573},
  issn		= {0032-5155},
  mrclass	= {53C50 (46F05 53C05 53C80 83C35)},
  mrnumber	= {2374400},
  mrreviewer	= {Francisco J. Palomo},
  doi		= {10.4171/PM/1794},
  url		= {https://doi.org/10.4171/PM/1794}
}

@article {LV:09,
    AUTHOR = {Lott, John and Villani, C\'{e}dric},
     TITLE = {Ricci curvature for metric-measure spaces via optimal transport},
   JOURNAL = {Ann. of Math. (2)},
  FJOURNAL = {Annals of Mathematics. Second Series},
    VOLUME = {169},
      YEAR = {2009},
    NUMBER = {3},
     PAGES = {903--991},
      ISSN = {0003-486X},
   MRCLASS = {53C23 (28A33 53C21)},
  MRNUMBER = {2480619},
       DOI = {10.4007/annals.2009.169.903},
}

@article{McC:24,
    AUTHOR = {McCann, Robert J.},
     TITLE = {A synthetic null energy condition},
   JOURNAL = {Comm. Math. Phys.},
  FJOURNAL = {Communications in Mathematical Physics},
    VOLUME = {405},
      YEAR = {2024},
    NUMBER = {2},
     PAGES = {Paper No. 38, 24},
      ISSN = {0010-3616},
       DOI = {10.1007/s00220-023-04908-1},
}

@article {MCC:20,
    AUTHOR = {McCann, Robert J.},
     TITLE = {Displacement convexity of {B}oltzmann's entropy characterizes
              the strong energy condition from general relativity},
   JOURNAL = {Camb. J. Math.},
  FJOURNAL = {Cambridge Journal of Mathematics},
    VOLUME = {8},
      YEAR = {2020},
    NUMBER = {3},
     PAGES = {609--681},
      ISSN = {2168-0930},
   MRCLASS = {53C50 (49Q22 53C21 58Z05 82C35 83C99)},
  MRNUMBER = {4192570},
       DOI = {10.4310/CJM.2020.v8.n3.a4},
       URL = {https://doi.org/10.4310/CJM.2020.v8.n3.a4},
}

@article {Min:15,
    AUTHOR = {Minguzzi, Ettore},
     TITLE = {Convex neighborhoods for {L}ipschitz connections and sprays},
   JOURNAL = {Monatsh. Math.},
  FJOURNAL = {Monatshefte f\"ur Mathematik},
    VOLUME = {177},
      YEAR = {2015},
    NUMBER = {4},
     PAGES = {569--625},
      ISSN = {0026-9255},
   MRCLASS = {53B15 (26A16 53B40)},
  MRNUMBER = {3371365},
MRREVIEWER = {Camelia M. Frigioiu},
       DOI = {10.1007/s00605-014-0699-y},
       URL = {http://dx.doi.org/10.1007/s00605-014-0699-y},
}

@article {Min:19a,
    AUTHOR = {Minguzzi, Ettore},
     TITLE = {Lorentzian causality theory},
   JOURNAL = {Living Rev.\ Relativ.},
      VOLUME = {22},
      YEAR = {2019},
    NUMBER = {3},
     PAGES = {220 pp},
       DOI = {10.1007/s41114-019-0019-x},
       URL = {https://doi.org/},
}

@article{MS:23,
    AUTHOR = {Mondino, Andrea and Suhr, Stefan},
     TITLE = {An optimal transport formulation of the {E}instein equations of
              general relativity},
   JOURNAL = {J. Eur. Math. Soc. (JEMS)},
  FJOURNAL = {Journal of the European Mathematical Society (JEMS)},
    VOLUME = {25},
      YEAR = {2023},
    NUMBER = {3},
     PAGES = {933--994},
      ISSN = {1435-9855},
   MRCLASS = {83C05 (49Q22 53C50)},
  MRNUMBER = {4560128},
       DOI = {10.4171/JEMS/1188},
}

@article {Pen:65,
    AUTHOR = {Penrose, Roger},
     TITLE = {Gravitational collapse and space-time singularities},
   JOURNAL = {Phys. Rev. Lett.},
  FJOURNAL = {Physical Review Letters},
    VOLUME = {14},
      YEAR = {1965},
     PAGES = {57--59},
      ISSN = {0031-9007},
   MRCLASS = {83.53},
  MRNUMBER = {0172678},
MRREVIEWER = {P. G. Bergmann},
       DOI = {10.1103/PhysRevLett.14.57},
       URL = {https://doi.org/10.1103/PhysRevLett.14.57},
}

@article {PSSS:14,
    AUTHOR = {Podolsk\'y, Ji\v{r}\'{\i} and S\"amann, Clemens and Steinbauer, Roland and {\v S}varc, Robert},
     TITLE = {The global existence, uniqueness and {$C^1$}-regularity
              of geodesics in nonexpanding impulsive gravitational waves},
   JOURNAL = {Classical Quantum Gravity},
  FJOURNAL = {Classical and Quantum Gravity},
    VOLUME = {32},
      YEAR = {2015},
    NUMBER = {2},
     PAGES = {025003, 23},
      ISSN = {0264-9381},
   MRCLASS = {83C35 (34A36 34A37 83C15)},
  MRNUMBER = {3291775},
MRREVIEWER = {Sergey I. Tertychniy},
       URL = {https://doi.org/10.1088/0264-9381/32/2/025003},
}

@article {PSSS:16,
    AUTHOR = {Podolsk\'y, Ji\v{r}\'{\i} and S\"amann, Clemens and Steinbauer, Roland and {\v S}varc, Robert},
     TITLE = {The global uniqueness and {$C^1$}-regularity of geodesics in
              expanding impulsive gravitational waves},
   JOURNAL = {Classical Quantum Gravity},
  FJOURNAL = {Classical and Quantum Gravity},
    VOLUME = {33},
      YEAR = {2016},
    NUMBER = {19},
     PAGES = {195010, 23},
      ISSN = {0264-9381},
   MRCLASS = {83C35},
  MRNUMBER = {3557142},
MRREVIEWER = {Ettore Minguzzi},
       URL = {https://doi.org/10.1088/0264-9381/33/19/195010},
}

@article {Sae:16,
    AUTHOR = {S{\"a}mann, Clemens},
     TITLE = {Global hyperbolicity for spacetimes with continuous metrics},
   JOURNAL = {Ann. Henri Poincar\'e},
  FJOURNAL = {Annales Henri Poincar\'e. A Journal of Theoretical and
              Mathematical Physics},
    VOLUME = {17},
      YEAR = {2016},
    NUMBER = {6},
     PAGES = {1429--1455},
      ISSN = {1424-0637},
   MRCLASS = {83C05 (53C50)},
  MRNUMBER = {3500220},
MRREVIEWER = {A. Burtscher},
       DOI = {10.1007/s00023-015-0425-x},
       URL = {http://dx.doi.org/10.1007/s00023-015-0425-x},
}

@article{SS:18,
  author={S\"{a}mann, Clemens and Steinbauer, Roland},
  title={On geodesics in low regularity},
  journal={Journal of Physics: Conference Series},
  volume={968},
  number={1},
  pages={012010},
  url={http://stacks.iop.org/1742-6596/968/i=1/a=012010},
  year={2018}
}

@article {Sen:98,
    AUTHOR = {Senovilla, Jos\'e M. M.},
     TITLE = {Singularity theorems and their consequences},
   JOURNAL = {Gen. Relativity Gravitation},
  FJOURNAL = {General Relativity and Gravitation},
    VOLUME = {30},
      YEAR = {1998},
    NUMBER = {5},
     PAGES = {701--848},
      ISSN = {0001-7701},
   MRCLASS = {83C75 (83-02)},
  MRNUMBER = {1623229},
MRREVIEWER = {Robert J. Low},
       DOI = {10.1023/A:1018801101244},
       URL = {https://doi.org/10.1023/A:1018801101244},
}

@article {SV:09,
    AUTHOR = {Steinbauer, Roland and Vickers, James A.},
     TITLE = {On the {G}eroch-{T}raschen class of metrics},
   JOURNAL = {Classical Quantum Gravity},
  FJOURNAL = {Classical and Quantum Gravity},
    VOLUME = {26},
      YEAR = {2009},
    NUMBER = {6},
     PAGES = {065001, 19},
      ISSN = {0264-9381},
   MRCLASS = {83C15},
  MRNUMBER = {2486329},
MRREVIEWER = {Stevan Pilipovi\'{c}},
       DOI = {10.1088/0264-9381/26/6/065001},
       URL = {https://doi.org/10.1088/0264-9381/26/6/065001},
}

@Article{Ste:14,
  author	= {Steinbauer, Roland},
  title		= {Every {L}ipschitz metric has {$C^1$}-geodesics},
  journal	= {Classical Quantum Gravity},
  fjournal	= {Classical and Quantum Gravity},
  volume	= {31},
  year		= {2014},
  number	= {5},
  pages		= {057001, 3},
  issn		= {0264-9381},
  mrclass	= {58E10 (34A36 83C20)},
  mrnumber	= {3168257},
  mrreviewer	= {Addolorata Salvatore},
  doi		= {10.1088/0264-9381/31/5/057001},
  url		= {https://doi.org/10.1088/0264-9381/31/5/057001}
}

@article{Ste:23,
    AUTHOR = {Steinbauer, Roland},
     TITLE = {The singularity theorems of general relativity and their low
              regularity extensions},
   JOURNAL = {Jahresber. Dtsch. Math.-Ver.},
  FJOURNAL = {Jahresbericht der Deutschen Mathematiker-Vereinigung},
    VOLUME = {125},
      YEAR = {2023},
    NUMBER = {2},
     PAGES = {73--119},
      ISSN = {0012-0456},
   MRCLASS = {83C75 (53B30 83-02)},
       DOI = {10.1365/s13291-022-00263-7},
}

@article {Sturm2006II,
    AUTHOR = {Sturm, Karl-Theodor},
     TITLE = {On the geometry of metric measure spaces. {II}},
   JOURNAL = {Acta Math.},
  FJOURNAL = {Acta Mathematica},
    VOLUME = {196},
      YEAR = {2006},
    NUMBER = {1},
     PAGES = {133--177},
      ISSN = {0001-5962},
   MRCLASS = {53C23 (28A33 51F99 53C70)},
  MRNUMBER = {2237207},
       DOI = {10.1007/s11511-006-0003-7},
}

@article {S:06a,
    AUTHOR = {Sturm, Karl-Theodor},
     TITLE = {On the geometry of metric measure spaces. {I}},
   JOURNAL = {Acta Math.},
  FJOURNAL = {Acta Mathematica},
    VOLUME = {196},
      YEAR = {2006},
    NUMBER = {1},
     PAGES = {65--131},
      ISSN = {0001-5962},
   MRCLASS = {53C23 (28A33 51F99 53C70)},
  MRNUMBER = {2237206},
       DOI = {10.1007/s11511-006-0002-8},
}

@article{BC:24,
    AUTHOR = {Braun, Mathias and Calisti, Matteo},
     TITLE = {Timelike {R}icci bounds for low regularity spacetimes by
              optimal transport},
   JOURNAL = {Commun. Contemp. Math.},
  FJOURNAL = {Communications in Contemporary Mathematics},
    VOLUME = {26},
      YEAR = {2024},
    NUMBER = {9},
     PAGES = {Paper No. 2350049},
      ISSN = {0219-1997},
       DOI = {10.1142/S0219199723500499},
}

@article{MR:25,
  title={On the equivalence of distributional and synthetic {R}icci curvature lower bounds},
  author={Mondino, Andrea and Ryborz, Vanessa},
  journal={Journal of Functional Analysis},
  volume={289},
  pages={111035},
  year={2025},
  publisher={Elsevier}
}

\end{document}